\documentclass[letterpaper, 11pt]{article}
\usepackage{graphicx}
\usepackage{amsmath}
\usepackage{amsfonts}
\usepackage{hyperref}
\usepackage{amsthm}
 \usepackage{amscd}
\usepackage{mathrsfs}
\usepackage{caption}
\usepackage{subcaption}
\usepackage{float}
\usepackage[top=1in,left=1in,right=1in,bottom=1in]{geometry}

\newtheorem{theorem}{Theorem}[section]
\newtheorem{proposition}[theorem]{Proposition}

\newtheorem{corollary}[theorem]{Corollary}

\newtheorem{example}{Example}[section]
\counterwithin{figure}{section}
\hypersetup{
	colorlinks = true,
	linkcolor  = red,
	citecolor  = green,
	filecolor  = cyan,
	urlcolor   = magenta,}

\newcommand{\ds}{\displaystyle}

\numberwithin{equation}{section}

\title{The Marchenko method to solve
 the general system of \\
 derivative nonlinear Schr\"odinger equations}
\author{Tuncay Aktosun and Ramazan Ercan\\
Department of Mathematics\\
University of Texas at Arlington\\
Arlington, TX 76019-0408, USA
\\
\\
Mehmet Unlu\\
Department of Mathematics\\
Recep Tayyip Erdogan University\\
53100 Rize, Turkey
}

\date{}

\begin{document}

\maketitle

\begin{abstract}
A system of linear integral equations is presented, which is the analog of the system of Marchenko integral equations, to solve the inverse
scattering problem for the linear system associated with the derivative NLS equations. The corresponding direct and inverse scattering problems
are analyzed, and the recovery of the potentials and the Jost solutions from the solution to the Marchenko system is described.
When the reflection coefficients are zero, some explicit solution formulas are provided for the potentials and the Jost solutions
in terms of a pair of constant matrix triplets representing the bound-state information for any number of
bound states and any multiplicities. In the reduced case, when the two potentials in the linear system are related to each other
through complex conjugation, the corresponding reduced Marchenko integral equation is obtained. The solution to the derivative NLS equation
is obtained from the solution to the reduced Marchenko integral equation. The theory presented is illustrated with
some explicit examples.

\end{abstract}

{\bf {AMS Subject Classification (2020):}} 35Q55, 37K10, 37K15, 37K30, 34A55, 34L25, 34L40, 47A40

{\bf Keywords:} general derivative NLS system, explicit solutions, inverse scattering transform, Kaup--Newell system, Gerdjikov--Ivanov system,
Chen--Lee--Liu system, first-order linear system, 
energy-dependent potentials, Marchenko method

\newpage

\section{Introduction}
\label{section1}

Our main goal in this paper is to present solutions to the general system of DNLS (derivative nonlinear Schr\"odinger) equations 
\cite{AC1991,APT2003,AS1981,K1984,OS1998a,OS1998b,T2010,TW1999a,TW1999b}
\begin{equation}\label{1.1}
\begin{cases}
i\,\tilde q_t+\tilde q_{xx}+i(4\delta-\epsilon) \tilde q \tilde q_x \tilde r+
4i\delta\tilde q^2 \tilde r_x+
\delta(4\delta+\epsilon)\,\tilde q^3 \tilde r^2
=0,\\
\noalign{\medskip}
i\,\tilde r_t-\tilde r_{xx}+i(4\delta-\epsilon) \tilde q \tilde r \tilde r_x+4i\delta\tilde q_x \tilde r^2
-
\delta(4\delta+\epsilon)\,\tilde q^2 \tilde r^3
=0,
\end{cases}
\end{equation}
where the subscripts denote the respective partial derivatives, the dependent variables $\tilde q$ and $\tilde r$ are complex-valued functions,
the independent variables $x$ and $t$ take values on the real axis $\mathbb R,$ 
and the parameters $\delta$ and $\epsilon$ are complex valued.
For clarity and simplicity, we assume that for each fixed $t\in\mathbb R,$ the scalar quantities $\tilde q(x,t)$ and $\tilde r(x,t)$
belong to the Schwartz class even though our results hold under weaker conditions.
We recall that the Schwartz class consists of functions of $x$ decaying to zero faster than any inverse power of $x$
as $x\to\pm\infty$ while the derivatives of all orders are continuous everywhere.

The integrability of the nonlinear system \eqref{1.1} by the inverse scattering transform method \cite{AC1991,AS1981,GGKM1967,L1968,NMPZ1984} is already known
because of the existence of a corresponding Lax pair. What is new and significant
in our paper is the development and use of the Marchenko method for \eqref{1.1},
and hence the implementation of recovery of $\tilde q(x,t)$ and $\tilde r(x,t)$
from the corresponding time-evolved scattering data by the Marchenko method.
Even though the Marchenko method is available for various other integrable systems such as
the NLS system \cite{AC1991,AKNS1974,AS1981,CD1982,NMPZ1984,ZS1972}
\begin{equation*}
\begin{cases}
i u_t+ u_{xx}-2u^2v
=0,\\
\noalign{\medskip}
i v_t- v_{xx}+2uv^2
=0,
\end{cases}
\end{equation*}
it has not been available for \eqref{1.1} so far.

Let us remark that the Marchenko method \cite{AM1963,AK2006,CS1989,F1963,F1967,N1980,N1983} is sometimes misnamed and called the Gel'fand--Levitan
method or the Gel'fand--Levitan--Marchenko method. 
The input used in the Gel'fand--Levitan method \cite{CS1989,GL1955,L1987,M1986,N1983} is the spectral function and not
the scattering data. It is the Marchenko method that uses the scattering data as input in the solution to an inverse scattering problem,
and hence it is the Marchenko method that is relevant
in the inverse scattering transform.
In our paper, we only deal with the Marchenko method and not with the Gel'fand--Levitan method.

The goal \cite{AC1991,AS1981,GGKM1967,NMPZ1984} in the inverse scattering transform method consists of the determination of
the solution $\tilde q(x,t)$ and $\tilde r(x,t)$ to \eqref{1.1}
when the initial values $\tilde q(x,0)$ and $\tilde r(x,0)$ are known.
The execution of the integrability of \eqref{1.1} in the sense of the inverse scattering transform
 is equivalent to the use of the following three steps.
In the first step, the initial values $\tilde q(x,0)$ and $\tilde r(x,0)$ are 
associated with a corresponding initial scattering data set $\tilde{\mathbf S}(\zeta,0),$
where $\zeta$ is an appropriate spectral parameter.
In the second step, the time evolution $\tilde{\mathbf S}(\zeta,0)\mapsto\tilde{\mathbf S}(\zeta,t)$
of the scattering data set
is described. In the third step, the time-evolved quantities
$\tilde q(x,t)$ and $\tilde r(x,t)$ are recovered
from the time-evolved scattering data set $\tilde{\mathbf S}(\zeta,t).$ The significance of our paper is that
we provide the execution of
these three steps, by describing the corresponding scattering data set, by showing the time evolution of the scattering data
set, and
by presenting the recovery of $\tilde q(x,t)$ and $\tilde r(x,t)$
from the corresponding solution to our Marchenko system of linear integral equations.

A particular strength of our paper is that we do not assume the simplicity of bound states
in the relevant scattering data set, which is usually artificially assumed in the analysis of \eqref{1.1}. On the
contrary, we
deal with any number of bound states having any multiplicities
in an elegant way with the help of a pair
of matrix triplets. 
Another strength of our paper is that we provide explicit solution formulas for \eqref{1.1} in closed form
corresponding to reflectionless scattering data with any number of bound states and any multiplicities.
The use of matrix triplets to describe the bound-state information in the input to the
Marchenko method is the most appropriate and elegant way to handle
bound states with multiplicities, and this is true in
the Marchenko method \cite{ABDV2010,ADV2007,ADV2010,AE2019,AE2022a,AV2006,B2008,B2017} for all other integrable systems as well.
The use of matrix triplets in the reflectionless scattering data sets causes the integral kernels in the
corresponding Marchenko systems to be separable, and hence it allows the construction of
explicit solution formulas in the reflectionless case corresponding to any number of bound states with any multiplicities.
Such formulas are expressed in a compact form in terms of matrix exponentials, and those formulas are valid
for any number of bound states and any multiplicities.

The general DNLS system \eqref{1.1}, when $(\delta,\epsilon)=(-1/4,1),$  yields the Kaup--Newell system \cite{KN1978} (also called DNLS I system)
given by
\begin{equation}
\label{1.2}
\begin{cases}
i  q_t+q_{xx}-i(q^2 r)_x=0,
\\
\noalign{\medskip}
ir_t-r_{xx}-i(q  r^2)_x
=0.
\end{cases}
\end{equation}
It reduces,
when $(\delta,\epsilon)=(0,1),$
to
the Chen--Lee--Liu system \cite{CLL1979} (also called DNLS II system) 
\begin{equation}
\label{1.3}
\begin{cases}
i  \tilde q_t+\tilde q_{xx}-i\tilde q \tilde q_x \tilde r=0,
\\
\noalign{\medskip}
i \tilde r_t-\tilde r_{xx}-i\tilde q \tilde r \tilde r_x
=0,
\end{cases}
\end{equation}
and 
it gives us,
when $(\delta,\epsilon)=(1/4,1),$
the 
Gerdjikov--Ivanov system \cite{GI1983} (also called DNLS III system)
\begin{equation}
\label{1.4}
\begin{cases}
i  \tilde q_t+\tilde q_{xx}+i\tilde q^2 \tilde r_x+\ds\frac{1}{2}\,\tilde q^3 \tilde r^2=0,
\\
\noalign{\medskip}
i \tilde r_t-\tilde r_{xx}+i \tilde q_x \tilde r^2-\ds\frac{1}{2}\,\tilde q^2 \tilde r^3
=0.
\end{cases}
\end{equation}
We are not interested in analyzing particular cases of the
general DNLS system \eqref{1.1} separately because that would be tedious and is
not necessary. Instead, we present our method to solve \eqref{1.1}
with the presence of the two free parameters $\delta$ and $\epsilon,$ and 
from that solution we are able to extract
the solution to any specific case by assigning particular values to 
the parameters. 

Even though it is possible to apply the Marchenko method
directly on \eqref{1.1} in the presence of the two free parameters $\delta$ and $\epsilon,$
for clarity and simplicity we instead proceed as follows.
We view the Kaup--Newell system \eqref{1.2} as the unperturbed system
and view \eqref{1.1} with the two parameters as the perturbed system. 
We use a tilde to denote the quantities related to the
perturbed system \eqref{1.1}, and the quantities without a tilde are related to
the 
unperturbed system.
This explains why we have written \eqref{1.2}
without the use of a tilde even though we have used
a tilde in \eqref{1.3} and \eqref{1.4}.

Because our method originates in analyzing a pair of linear systems
corresponding to the integrable system
\eqref{1.1}, it turns out that it is more appropriate
for us to express 
\eqref{1.1} by using three complex-valued parameters
$a,$ $b,$ $\kappa$ instead of
the two complex parameters
$\delta$ and $\epsilon$ in \eqref{1.1}. 
Letting
\begin{equation}\label{1.5}
\delta=\ds\frac{\kappa(a-b-1)}{4},\quad \epsilon=\kappa,
\end{equation}
from \eqref{1.1} we obtain the equivalent system
\begin{equation}
\label{1.6}
\begin{cases}
i\,\tilde q_t+\tilde q_{xx}+i\kappa(a-b-2) \tilde q \tilde q_x \tilde r+
i\kappa(a-b-1) \tilde q^2 \tilde r_x+
\ds\frac{\kappa^2(a-b)(a-b-1)}{4}\,\tilde q^3 \tilde r^2
=0,\\
\noalign{\medskip}
i\,\tilde r_t-\tilde r_{xx}+i\kappa(a-b-2) \tilde q \tilde r \tilde r_x+
i\kappa(a-b-1) \tilde q_x \tilde r^2-
\ds\frac{\kappa^2(a-b)(a-b-1)}{4}\,\tilde q^2 \tilde r^3
=0,
\end{cases}
\end{equation}
with still containing only two arbitrary parameters because the new parameters $a$ and
$b$ appear in \eqref{1.6} in the combined form $a-b.$ 
The advantage of using three relevant parameters in the corresponding linear domain,
even though there are only two relevant parameters in the nonlinear domain, will soon be apparent.

If we use the parameters $\delta$ and $\epsilon,$ from \eqref{1.1} it difficult to see 
why we single out \eqref{1.2} as the unperturbed
system with the choice $(\delta,\epsilon)=(1/4,1).$
On the other hand, the simplicity of \eqref{1.2} is easily seen
from the equivalent formulation of
\eqref{1.1} as \eqref{1.6} with the new parameters $a,$ $b,$ and $\kappa.$
Although we could use
any particular case of \eqref{1.6} as
the unperturbed problem instead of 
\eqref{1.2}, it is advantageous and the simplest
to use \eqref{1.2} as the unperturbed nonlinear problem.
This is because \eqref{1.2} is obtained
from \eqref{1.6} by using the simplest choice
$(a,b,\kappa)=(0,0,1).$ 
We note that the Chen--Lee--Liu system \eqref{1.3} corresponds to using
$(a,b,\kappa)=(1,0,1)$ in \eqref{1.6},
and the Gerdjikov--Ivanov system \eqref{1.4}
is obtained by using $(a,b,\kappa)=(1,-1,1).$
In Example~\ref{example9.6}, we elaborate on the issue
that any particular case of the nonlinear system \eqref{1.6} could be used as
the unperturbed problem instead of the particular nonlinear system
\eqref{1.2}. As already mentioned, in the analysis of
the linear system associated with the nonlinear system
\eqref{1.6}, it is advantageous to choose
\eqref{1.2} as the unperturbed system.

Let $(\mathcal X,\mathcal T)$ be the AKNS pair \cite{AC1991,AKNS1974,AS1981,NMPZ1984} associated with the unperturbed
nonlinear system \eqref{1.2} so that the matrix equality
\begin{equation}
\label{1.7}
\mathcal X_t-\mathcal T_x+\mathcal X \mathcal T-\mathcal T \mathcal X=0,
\end{equation}
yields \eqref{1.2}. Thus, corresponding to \eqref{1.2} we have
the pair of unperturbed linear systems given by
\begin{equation}
\label{1.8}
\Psi_x=\mathcal X \Psi,\quad \Psi_t=\mathcal T\Psi.
\end{equation}
It can be verified directly that we can choose the AKNS pair
$(X,T)$ in \eqref{1.7} as
\begin{equation}
\label{1.9}
\mathcal X=\begin{bmatrix}
-i\zeta^2& \zeta  q\\
\noalign{\medskip}
\zeta r& i\zeta^2
\end{bmatrix},
\quad
\mathcal T=\begin{bmatrix}
-2i\zeta^4-iqr \zeta^2 
& 2q\zeta^3+(iq_x+q^2 r)\zeta\\
\noalign{\medskip}
2r\zeta^3+(-i r_x+q r^2)\zeta& 2i\zeta^4+iqr\zeta^2
\end{bmatrix},
\end{equation}
where we use $\zeta$ to denote the spectral parameter.

Similarly, let $(\tilde{\mathcal X},\tilde{\mathcal T})$ be the AKNS pair associated with the perturbed
nonlinear system \eqref{1.6} so that the matrix equality
\begin{equation}
\label{1.10}
\tilde{\mathcal X}_t-\tilde{\mathcal T}_x+\tilde{\mathcal X} \tilde{\mathcal T}-\tilde{\mathcal T} \tilde{\mathcal X}=0,
\end{equation}
yields \eqref{1.6}. Hence,
corresponding to \eqref{1.6} we have
the pair of perturbed linear systems given by
\begin{equation}
\label{1.11}
\tilde\psi_x=\tilde{\mathcal X}\tilde\psi,\quad \tilde\psi_t=\tilde {\mathcal T}\tilde\psi.
\end{equation}
It can again be verified directly that we can choose $\tilde{\mathcal X}$ and $\tilde{\mathcal T}$ in \eqref{1.10} as
\begin{equation}
\label{1.12}
\tilde {\mathcal X}=
\begin{bmatrix}
-i\zeta^2+\ds\frac{i b}{2}\tilde q \tilde r&\kappa\, \zeta \tilde q
\\
\noalign{\medskip}
\ds\frac{1}{\kappa}\,\zeta \tilde r&i\zeta^2+\ds\frac{i a}{2}\tilde q \tilde r
\end{bmatrix},
\quad
\tilde{\mathcal T}=\begin{bmatrix}\tilde{\mathcal T}_{11}&\tilde{\mathcal T}_{12}
\\
\noalign{\medskip}
\tilde{\mathcal T}_{21}&\tilde{\mathcal T}_{22}
\end{bmatrix},
\end{equation}
where we have defined
\begin{equation}
\label{1.13}
\tilde{\mathcal T}_{11}:=
-2i\zeta^4-i\tilde q \tilde r \zeta^2+\ds\frac{b}{2}\,(\tilde q \tilde r_x-\tilde q_x \tilde r)+\ds\frac{i b}{2}\left(b-a+\ds\frac{3}{2}\right)\tilde q^2 \tilde r^2,
\end{equation}
\begin{equation}
\label{1.14}
\tilde{\mathcal T}_{12}:= 2\kappa\,\zeta^3 \tilde q
+\kappa\,\zeta
\left[i \tilde q_x+\ds\frac{1}{2}\,\left(b-a+2\right)\tilde q^2 \tilde r
\right],
\end{equation}
\begin{equation}
\label{1.15}
\tilde{\mathcal T}_{21}:=
\ds\frac{2}{\kappa}\,\zeta^3 \tilde r+
\ds\frac{1}{\kappa}\,\zeta\left[-i\tilde r_x+\ds\frac{1}{2}\,\left(b-a+2\right)\tilde q \tilde r^2
\right],
\end{equation}
\begin{equation}
\label{1.16}
\tilde{\mathcal T}_{22}:=
2i\zeta^4+i\tilde q \tilde r \zeta^2+\ds\frac{a}{2}\,(\tilde q \tilde r_x-\tilde q_x \tilde r)+\ds\frac{i a}{2}\left(b-a+\ds\frac{3}{2}\right)\tilde q^2 \tilde r^2.
\end{equation}
Let us remark that the simplicity of \eqref{1.2}, and hence its choice as the unperturbed problem,
is also seen by comparing the matrices $\mathcal X$ and $\tilde{\mathcal X}$
appearing in \eqref{1.9} and \eqref{1.12}, respectively.
Another simple aspect of
$\mathcal X$ is that the matrix $\mathcal X$ has zero trace, which implies that
the left and right transmission coefficients in the corresponding scattering data set are equal,
whereas
the trace of the matrix $\tilde{\mathcal X}$ is nonzero unless
$a+b=0.$

In our paper we apply
our Marchenko method to the unperturbed linear system given in the first equality of \eqref{1.8}, 
and we obtain the solution to the corresponding inverse scattering problem.
We then get the solution to the inverse scattering problem for
the linear system in the first equality of
\eqref{1.11} by relating
the perturbed linear system to
the unperturbed linear system
through the transformation expressed as
\begin{equation}
\label{1.17}
\tilde\Psi=\mathcal G \Psi.
\end{equation}
In \eqref{1.17} the coefficient matrix $\mathcal G$ is given by
\begin{equation}
\label{1.18}
\mathcal G:=
\begin{bmatrix}
E(x,t)^{b}& 0\\
\noalign{\medskip}
0&E(x,t)^{a}
\end{bmatrix},
\end{equation}
with the complex-valued scalar quantity $E(x,t)$ defined as
\begin{equation}
\label{1.19}
E(x,t):=\exp\left(\ds\frac{i}{2}\ds\int_{-\infty}^x dz\,q(z,t)\,r(z,t)\right).
\end{equation}
The perturbed potentials $\tilde q(x,t)$ and $\tilde r(x,t)$ appearing in \eqref{1.12}--\eqref{1.16}
are related to the unperturbed potentials 
$q(x,t)$ and $r(x,t)$ appearing in \eqref{1.2} as
\begin{equation}
\label{1.20}
\tilde q(x,t):=\ds\frac{1}{\kappa}\, q(x,t)\,E(x,t)^{b-a},\quad \tilde r(x,t):=\kappa\,r(x,t) \, E(x,t)^{a-b}.
\end{equation}
As seen from \eqref{1.17} and \eqref{1.18}, for the 
perturbed linear system in the
first equality of \eqref{1.11}, the use of the three parameters
$a,$ $b,$ $\kappa$ is essential.
On the other hand, from \eqref{1.20} we see why
$a$ and $b$ appear not separately but together as $a-b$ in the perturbed nonlinear system \eqref{1.6}.

When the arguments of a function are clearly understood, we may omit those arguments. Hence, we may use $E$ instead of
$E(x,t)$ and similarly we may use $q,$ $r,$ $\tilde q,$ and $\tilde r$ instead of
$q(x,t),$ $r(x,t),$ $\tilde q(x,t),$ and $\tilde r(x,t),$ respectively.
Let us remark that, using \eqref{1.8}, \eqref{1.11}, and \eqref{1.17}, we can express $(\mathcal X,\mathcal T)$ and
$(\tilde{\mathcal X},\tilde{\mathcal T})$
in terms of each other as
\begin{equation*}
\tilde{\mathcal X}=\mathcal G_x \mathcal G^{-1}+\mathcal G \mathcal X \mathcal G^{-1},\quad \tilde{\mathcal T}=\mathcal G_t \mathcal G^{-1}+\mathcal G \mathcal T\, \mathcal G^{-1},
\end{equation*}
\begin{equation*}
\mathcal X=-\mathcal G^{-1}  \mathcal G_x+\mathcal G^{-1} \tilde{\mathcal X}\mathcal G,\quad \mathcal T
=-\mathcal G^{-1}  \mathcal G_t+\mathcal G^{-1} \tilde{\mathcal T}\,\mathcal G.
\end{equation*}

Our paper is organized as follows. In Section~\ref{section2} we provide the relevant results related to
the direct scattering problem for the unperturbed linear system \eqref{2.1}. The
relevant quantities include the Jost solutions, the scattering coefficients, and the bound-state information.
We use a pair of matrix triplets to describe the bound-state information
with any number of bound states and any multiplicities.
In Section~\ref{section3} we present our Marchenko system of integral equations
relevant to the inverse scattering problem for \eqref{2.1}.
We relate the scattering data set to the kernel of the Marchenko system.
We also describe the recovery of the potentials and the Jost solutions from the
solution to the Marchenko system.
In Section~\ref{section4} we consider the Marchenko system for \eqref{2.1} when
the reflection coefficients are zero. In that case, the Marchenko system has a separable kernel,
and hence it can be solved explicitly by using the methods from linear algebra.
We present some explicit formulas expressing the corresponding potentials and Jost solutions 
in terms of the two matrix triplets used as input to the Marchenko system.
In Section~\ref{section5} we relate
the quantities relevant for the perturbed linear system \eqref{5.1}
to the corresponding relevant quantities for the unperturbed linear system \eqref{2.1}.
In Section~\ref{section6} we present our Marchenko method to obtain the
solution to the perturbed nonlinear system \eqref{1.6}.
In Section~\ref{section7}, in the reflectionless case we present explicit formulas for
the quantities relevant to the perturbed linear system.
In Section~\ref{section8}, we consider the unperturbed linear and nonlinear systems in the
special case where the potentials $q(x,t)$ and $r(x,t)$ are related to each other
through complex conjugation. Using such a reduction, we obtain the corresponding linear and nonlinear equations
and also a scalar Marchenko equation, and we present the recovery of the potential $q(x,t)$ from the solution to
that reduced Marchenko equation.
Finally, in Section~\ref{section9} we provide some explicit examples to illustrate the theory 
presented in the previous sections.

\section{The direct scattering problem for the unperturbed system}
\label{section2}

In this section, we present the basic results related to the direct scattering problem for the unperturbed linear system
given in the first equality of \eqref{1.8}. 
For convenience, we write it as
\begin{equation}\label{2.1}
\ds\frac{d}{dx}\begin{bmatrix}
\alpha\\
\noalign{\medskip}
\beta
\end{bmatrix}=
\begin{bmatrix}
-i\zeta^2 & \zeta \,q(x,t)\\
\noalign{\medskip}
\zeta\, r(x,t) & i\zeta^2
\end{bmatrix}
\begin{bmatrix}
\alpha\\
\noalign{\medskip}
\beta
\end{bmatrix},\qquad x\in\mathbb R,
\end{equation}
where the quantities
$\alpha$ and $\beta$ are the components of the wavefunction
depending on the spacial variable $x,$ the time variable $t,$ and the spectral parameter $\zeta;$
 and the potentials
$q(x,t)$ and $r(x,t)$ are assumed to belong to the Schwartz class
for each fixed $t\in\mathbb R.$
The solution to the direct scattering problem for \eqref{2.1} 
consists of the specification of the scattering data set $\mathbf S(\zeta,t)$
corresponding to the potentials $q(x,t)$ and $r(x,t)$ appearing in \eqref{2.1}. The direct
scattering problem is solved
as follows. Using $q(x,t)$ and $r(x,t)$ as input to \eqref{2.1}, we obtain 
the four particular solutions to \eqref{2.1}, which are known as 
the Jost solutions. From the large spacial asymptotics of those four
Jost solutions, we get the scattering coefficients.
Finally, we obtain $\mathbf S(\zeta,t)$ by supplementing the set of scattering coefficients with
the bound-state information for \eqref{2.1}.

We use 
$\psi(\zeta,x,t),$ $\bar\psi(\zeta,x,t),$ $\phi(\zeta,x,t),$ $\bar\phi(\zeta,x,t)$ to
denote the four Jost solutions to \eqref{2.1}, where they satisfy the respective spacial asymptotics
\begin{equation}\label{2.2}
\psi(\zeta,x,t)=\begin{bmatrix}
o(1)\\
\noalign{\medskip}
 e^{i\zeta^2x}\left[1+o(1)\right]
\end{bmatrix} ,\qquad  x\to+\infty,
\end{equation}
\begin{equation}\label{2.3}
\bar\psi(\zeta,x,t)=\begin{bmatrix}
e^{-i\zeta^2x}\left[1+o(1)\right]\\
\noalign{\medskip}
o(1)
\end{bmatrix} ,\qquad  x\to+\infty,
\end{equation}
\begin{equation}\label{2.4}
\phi(\zeta,x,t)=\begin{bmatrix}
e^{-i\zeta^2x}\left[1+o(1)\right]\\
\noalign{\medskip}
o(1)
\end{bmatrix} ,\qquad   x\to-\infty,
\end{equation}
\begin{equation}\label{2.5}
\bar\phi(\zeta,x,t)=\begin{bmatrix}
o(1)\\
\noalign{\medskip}
e^{i\zeta^2x}\left[1+o(1)\right]
\end{bmatrix} ,\qquad  x\to-\infty.
\end{equation}
We remark that the overbar does not denote complex conjugation.

We have six scattering coefficients associated with \eqref{2.1}, i.e. the transmission coefficients $T(\zeta,t)$ and $\bar T(\zeta,t),$ 
the right reflection coefficients  $R(\zeta,t)$ and $\bar R(\zeta,t),$ and the left reflection coefficients  $L(\zeta,t)$ and $\bar L(\zeta,t).$ 
Since the trace of the
coefficient matrix in \eqref{2.1} is zero, the transmission coefficients from the left and from the right are equal to each other, and 
hence we do not need to use separate notations for the left and right transmission coefficients. 
The six scattering coefficients are obtained from the spacial asymptotics of the Jost solutions given by
\begin{equation}\label{2.6}
\psi(\zeta,x,t)=\begin{bmatrix}
\ds\frac{L(\zeta,t)}{T(\zeta,t)}\,e^{-i\zeta^2 x}\left[1+o(1)\right]\\
\noalign{\medskip}
\ds\frac{1}{T(\zeta,t)}\,e^{i\zeta^2 x}\left[1+o(1)\right]
\end{bmatrix}, \qquad   x\to-\infty,
\end{equation}
\begin{equation}\label{2.7}
\bar\psi(\zeta,x,t)=\begin{bmatrix}
\ds\frac{1}{\bar T(\zeta,t)}\,e^{-i\zeta^2 x}\left[1+o(1)\right]\\
\noalign{\medskip}
\ds\frac{\bar L(\zeta,t)}{\bar T(\zeta,t)}\,e^{i\zeta^2 x}\left[1+o(1)\right]
\end{bmatrix}, \qquad  x\to-\infty,
\end{equation}
\begin{equation}\label{2.8}
\phi(\zeta,x,t)=\begin{bmatrix}
\ds\frac{1}{T(\zeta,t)}\,e^{-i\zeta^2x}\left[1+o(1)\right]\\
\noalign{\medskip}
\ds\frac{R(\zeta,t)}{T(\zeta,t)}\,e^{i\zeta^2 x}\left[1+o(1)\right]
\end{bmatrix}, \qquad   x\to+\infty,
\end{equation}
\begin{equation}\label{2.9}
\bar\phi(\zeta,x,t)=\begin{bmatrix}
\ds\frac{\bar R(\zeta,t)}   {\bar T(\zeta,t)}\,e^{-i\zeta^2 x}\left[1+o(1)\right]\\
\noalign{\medskip}
\ds\frac{1}{\bar T(\zeta,t)}\,e^{i\zeta^2 x}\left[1+o(1)\right]
\end{bmatrix}, \qquad   x\to+\infty.
\end{equation}

Let us use
the subscripts $1$ and $2$ to denote the first and second components, respectively,
of the Jost solutions. Hence,  
we introduce the notation
\begin{equation} \label{2.10}
\begin{bmatrix}
\psi_1(\zeta,x,t)\\
\noalign{\medskip}\psi_2(\zeta,x,t)
\end{bmatrix}:=\psi(\zeta,x,t),\quad \begin{bmatrix}
\bar\psi_1(\zeta,x,t)\\ \noalign{\medskip}\bar\psi_2(\zeta,x,t)
\end{bmatrix}:=\bar\psi(\zeta,x,t),
\end{equation}
 \begin{equation} \label{2.11}
\begin{bmatrix}
\phi_1(\zeta,x,t)\\
\noalign{\medskip}\phi_2(\zeta,x,t)
\end{bmatrix}:=\phi(\zeta,x,t),\quad \begin{bmatrix}
\bar\phi_1(\zeta,x,t)\\ \noalign{\medskip}\bar\phi_2(\zeta,x,t)
\end{bmatrix}:=\bar\phi(\zeta,x,t).
\end{equation}
We also introduce
the auxiliary spectral parameter $\lambda$
in terms of the spectral parameter $\zeta$
as
\begin{equation}\label{2.12}
\lambda=\zeta^{2},\quad \zeta=\sqrt{\lambda},
\end{equation}
with the square root denoting the principal branch of the complex-valued square-root function.
We remark that when $\lambda$ takes values on the real axis, $\zeta$ takes values on the real and imaginary axes.
We use $\mathbb{C^+}$ and $\mathbb{C^-}$ to denote the upper-half and lower-half,
respectively, of the complex plane $\mathbb C,$  and we let  $\mathbb{\overline{C^+}}:=\mathbb{C^+}\cup\mathbb R$ and  
$\mathbb{\overline{C^-}}:=\mathbb{C^-}\cup\mathbb R.$

The quantity $E(x,t)$ defined in \eqref{1.19} is in general complex valued, and hence it
does not necessarily have the unit amplitude.  From \eqref{1.19} it follows that
\begin{equation}
\label{2.13}
\ds\lim_{x\to-\infty}E(x,t)=1, \quad \ds\lim_{x\to+\infty}E(x,t)=e^{i\mu/2},
\end{equation}
where we have defined the complex constant $\mu$ as
\begin{equation}
\label{2.14}
\mu:=\int_{-\infty}^\infty dz\,q(z,t)\,r(z,t).
\end{equation} 
With the help of \eqref{1.2} one can show that
$\mu$ is a complex constant and its value is independent of $t.$
Hence, from \eqref{2.14} we obtain
\begin{equation}
\label{2.15}
\mu=\int_{-\infty}^\infty dz\,q(z,0)\,r(z,0).
\end{equation} 

In the following theorem, we summarize the relevant properties of the Jost solutions and
the scattering coefficients for \eqref{2.1}. 

\begin{theorem}
\label{theorem2.1}
Assume that the potentials $q(x,t)$ and $r(x,t)$ appearing in the first-order system \eqref{2.1}  belong to the Schwartz class
for each fixed $t\in\mathbb R.$ Let $E$ denote
the quantity defined in \eqref{1.19}, and let $\mu$ be the complex constant defined in \eqref{2.14}.
Further, assume that the spectral
parameters $\zeta$ and $\lambda$ are related to each other as in  \eqref{2.12}. Then, we have the following:

\begin{enumerate}

\item[\text{\rm(a)}] For each fixed pair of $x$ and $t$ in $\mathbb R,$ the Jost solutions $\psi(\zeta,x,t)$ and $\phi(\zeta,x,t)$
to \eqref{2.1}  are analytic in the first and third quadrants in the complex $\zeta$-plane
and are continuous in the closures of those regions. Similarly, the Jost solutions $\bar\psi(\zeta,x,t)$ and
$\bar\phi(\zeta,x,t)$ are analytic in the second and fourth quadrants in the complex $\zeta$-plane
and are continuous in the closures of those regions. 

\item[\text{\rm(b)}]  The components of the Jost solutions appearing in \eqref{2.10} and \eqref{2.11} satisfy the following properties. 
The components $\psi_1(\zeta,x,t),$ $\bar\psi_2(\zeta,x,t),$ $\phi_2(\zeta,x,t),$ and  $\bar\phi_1(\zeta,x,t)$ 
are odd in $\zeta;$ and the components 
$\psi_2(\zeta,x,t),$ $\bar\psi_1(\zeta,x,t),$ $\phi_1(\zeta,x,t),$ and  $\bar\phi_2(\zeta,x,t)$ are even in $\zeta.$ 
Moreover, for each fixed pair of $x$ and $t$ in $\mathbb R,$ the four scalar quantities $\psi_1(\zeta,x,t)/\zeta,$ $\psi_2(\zeta,x,t),$ 
$\phi_1(\zeta,x,t),$ and $\phi_2(\zeta,x,t)/\zeta$ are even in 
$\zeta;$ and they are analytic in $\lambda \in\mathbb C^+$ and continuous in 
$\lambda\in\overline{\mathbb C^+}.$ Similarly, for each fixed pair of $x$ and $t$ in $\mathbb R,$ the four scalar quantities 
$\bar\psi_1(\zeta,x,t),$ $\bar\psi_2(\zeta,x,t)/\zeta,$ $\bar\phi_1(\zeta,x,t)/\zeta,$ and 
$\bar\phi_2(\zeta,x,t)$ are even in $\zeta;$ and they are analytic in 
$\lambda \in\mathbb C^-$ and continuous in $\lambda\in\overline{\mathbb C^-}.$

\item[\text{\rm(c)}] 
The transmission coefficients $T(\zeta,t)$ and $\bar T(\zeta,t)$ are independent of $t,$ and hence we have
\begin{equation}
\label{2.16}
T(\zeta,t)=T(\zeta,0),\quad 
\bar T(\zeta,t)=\bar T(\zeta,0), \qquad t\in\mathbb R.
\end{equation}
The quantity $T(\zeta,t)$ is continuous
in $\zeta\in\mathbb R$ and
has a meromorphic extension from $\zeta\in\mathbb R$ to the first and third quadrants in the complex $\zeta$-plane. 
Moreover, $T(\zeta,t)$ is an even function of $\zeta,$ and hence it is a function of 
$\lambda$ in $\overline{\mathbb C^+}.$
The quantity $1/T(\zeta,t)$ is analytic in  $\lambda\in\mathbb C^+$
and continuous in $\lambda$ in $\overline{\mathbb C^+}.$
Furthermore, $T(\zeta,t)$ is meromorphic in $\lambda\in\mathbb C^+$ with a finite number of poles there, where the 
poles are not necessarily simple 
but have finite multiplicities. 
The large $\zeta$-asymptotics of $T(\zeta,t)$ expressed in $\lambda$ is given by
\begin{equation}\label{2.17}
T(\zeta,t)=\ds e^{-i\mu/2}\left[1+O\left(\frac{1}{\lambda}\right)\right],\qquad \lambda\to\infty  \text{\rm{ in }} 
\overline{\mathbb C^+}.
\end{equation}			
Similarly, the quantity $\bar T(\zeta,t)$ is continuous
in $\zeta\in\mathbb R$ and
has a meromorphic extension from $\zeta\in\mathbb R$ to the second and fourth quadrants in the complex 
$\zeta$-plane. Furthermore, $\bar T(\zeta,t)$ is an even function of $\zeta,$ and hence it is a function of
$\lambda$ in $\overline{\mathbb C^-}.$
The quantity $1/\bar T(\zeta,t)$ is analytic in  $\lambda\in\mathbb C^-$
and continuous in $\lambda$ in $\overline{\mathbb C^-}.$
Moreover, $\bar T(\zeta,t)$ is meromorphic in $\lambda\in\mathbb C^-$ with a finite number of poles, where 
the poles are not necessarily simple but have finite multiplicities. The large $\zeta$-asymptotics of $\bar T(\zeta,t)$ expressed in $\lambda$ is given by
\begin{equation}\label{2.18}
\bar T(\zeta,t)=\ds e^{i\mu/2}\left[1+O\left(\frac{1}{\lambda}\right)\right],\qquad \lambda\to\infty 
\text{\rm{ in }} \overline{\mathbb C^-}.
\end{equation}
		
\item[\text{\rm(d)}] For each fixed $t\in\mathbb R,$ the small
$\zeta$-asymptotics of the six scattering coefficients $T(\zeta,t),$ $\bar T(\zeta,t),$ $R(\zeta,t),$ $\bar R(\zeta,t),$ $L(\zeta,t),$ 
and $\bar L(\zeta,t)$  are expressed in $\lambda$ as
\begin{equation}\label{2.19}
T(\zeta,t)=1+O(\lambda),\qquad \lambda\to 0
\text{\rm{ in }} \overline{\mathbb C^+},
\end{equation}
\begin{equation}\label{2.20}
\bar T(\zeta,t)=1+O(\lambda),\qquad \lambda\to 0
\text{\rm{ in }} \overline{\mathbb C^-},
\end{equation}
\begin{equation}\label{2.21}
R(\zeta,t)=-\sqrt{\lambda}\left[\ds\int_{-\infty}^\infty dz\,r(z,t)+O(\lambda)\right], \qquad \lambda\to 0
\text{\rm{ in }} \mathbb R,
\end{equation}
\begin{equation}\label{2.22}
\bar R(\zeta,t)=\sqrt{\lambda}\left[ \ds\int_{-\infty}^\infty dz\,q(z,t)+O(\lambda)\right],\qquad \lambda\to 0
\text{\rm{ in }} \mathbb R,
\end{equation}
\begin{equation*}
L(\zeta,t)=-\sqrt{\lambda}\left[\ds\int_{-\infty}^\infty dz\,q(z,t)+O(\lambda)\right],\qquad \lambda\to 0
\text{\rm{ in }} \mathbb R,
\end{equation*}
\begin{equation*}
\bar L(\zeta,t)=\sqrt{\lambda}\left[\ds\int_{-\infty}^\infty dz\,r(z,t)+O(\lambda)\right], \qquad \lambda\to 0
\text{\rm{ in }} \mathbb R.
\end{equation*}

\item[\text{\rm(e)}] 
The time evolutions of the reflection coefficients are given by
\begin{equation}
\label{2.23}
R(\zeta,t)=R(\zeta,0)\, e^{4 i \lambda^2 t},\quad 
\bar R(\zeta,t)=\bar R(\zeta,0)\, e^{-4 i \lambda^2 t}, \qquad \lambda\in\mathbb R,\quad
t\in\mathbb R,
\end{equation}
\begin{equation}
\label{2.24}
L(\zeta,t)=L(\zeta,0)\, e^{-4 i \lambda^2 t},\quad 
\bar L(\zeta,t)=\bar L(\zeta,0)\, e^{4 i \lambda^2 t}, \qquad \lambda\in\mathbb R,\quad t\in\mathbb R,
\end{equation}
with the understanding that the $\zeta$-domain and the $\lambda$-domain are related to each other
by the relationship expressed in \eqref{2.12}.
For each fixed $t\in\mathbb R,$ each of the four reflection coefficients $R(\zeta,t),$ $\bar R(\zeta,t),$ $L(\zeta,t),$ and $\bar L(\zeta,t)$ is continuous when 
$\lambda\in\mathbb R,$ is an odd function of $\zeta,$ and has the behavior $O(1/\zeta^{5/2})$ as $\lambda\to\pm\infty.$ 
Furthermore, the four function $R(\zeta,t)/\zeta,$ $\bar R(\zeta,t)/\zeta,$ $L(\zeta,t)/\zeta,$ $\bar L(\zeta,t)/\zeta$ are even 
in $\zeta;$ are continuous functions of
$\lambda\in\mathbb R;$ and they behave as $O(1/\lambda^3)$ as $\lambda\to\pm\infty.$

\item[\text{\rm(f)}] For each fixed $t\in\mathbb R,$ the scattering coefficients satisfy
\begin{equation}
\label{2.25}
T(\zeta,t)\,\bar T(\zeta,t)+R(\zeta,t)\,\bar R(\zeta,t)=1,\quad 
T(\zeta,t)\,\bar T(\zeta,t)+L(\zeta,t)\,\bar L(\zeta,t)=1,\qquad \lambda\in\mathbb R.
\end{equation}

\item[\text{\rm(g)}]  For each fixed $t\in\mathbb R,$ the left reflection coefficients are determined when the right reflection coefficients
and the transmission coefficients are known, and we have
\begin{equation}
\label{2.26}
L(\zeta,t)=-\ds\frac{\bar R(\zeta,t)\,T(\zeta,t)}{\bar T(\zeta,t)},
\quad \bar L(\zeta,t)=-\ds\frac{R(\zeta,t)\,\bar T(\zeta,t)}{T(\zeta,t)},\qquad \lambda\in\mathbb R.
\end{equation}
Conversely, as \eqref{2.26} indicates, the right reflection coefficients are determined when the left reflection coefficients
and the transmission coefficients are known.

\item[\text{\rm(h)}]  The bound states for \eqref{2.1} correspond to solutions that are square integrable
in $x.$ Such solutions cannot occur when $\lambda$ is real. In particular, there is no bound
state at $\zeta=0$ or equivalently at $\lambda=0.$ A bound state can only occur
at a complex value of $\zeta$ at which the transmission coefficient $T(\zeta,t)$ has a pole in the interiors of 
the first or third quadrants in the complex
$\zeta$-plane, or at which the transmission coefficient $\bar T(\zeta,t)$ has a pole in the 
interiors of the second or the fourth quadrants. Since the parameter
$\zeta$ appears as $\zeta^2$ in the transmission coefficients $T(\zeta,t)$ and $\bar T(\zeta,t),$ the $\zeta$-values corresponding
to the bound states must be symmetrically located with respect to the origin in the complex $\zeta$-plane.

\item[\text{\rm(i)}] The number of poles of $T(\zeta,t)$ in the upper-half  complex $\lambda$-plane is finite and we use
$\{\lambda_j\}_{j=1}^N$ to denote the set of those distinct poles, where we use $N$ to denote their number without counting the multiplicities. 
Similarly, the number of poles of $\bar T(\zeta,t)$ 
in the lower-half  complex $\lambda$-plane is finite and we use
$\{\bar\lambda_j\}_{j=1}^{\bar N}$ to denote the set of those distinct poles,
where we use $\bar N$ to denote their number without counting the multiplicities. 
The multiplicity of each of those poles is finite, and we use $m_j$ to denote
the multiplicity of the pole at $\lambda=\lambda_j$ and use $\bar m_j$ to denote the multiplicity
of the pole at $\lambda=\bar\lambda_j.$ 
As a consequence of \eqref{2.16}, it follows that each of the quantities
$\lambda_j,$ $m_j,$ $\bar\lambda_j,$ $\bar m_j$ are independent of $t,$ and hence their
values at any time $t$ coincide with the corresponding values at $t=0.$

\item[\text{\rm(j)}] 
For each bound state and multiplicity, there corresponds a norming constant. We use the double-indexed constants $c_{jk}$ for $0\le k\le m_j-1$
to denote the norming
constants at the bound state $\lambda=\lambda_j$ at time $t=0,$ and
we use the double-indexed constants $\bar c_{jk}$ for $0\le k\le \bar m_j-1$
to denote the norming
constants at the bound state $\lambda=\bar\lambda_j$ at time $t=0.$
Thus, the bound-state information at $t=0$ for \eqref{2.1} consists of the two sets given by
\begin{equation}
\label{2.27}
\left\{\lambda_j,m_j,\{c_{jk}\}_{k=0}^{m_j-1}\right\}_{j=1}^N,\quad
\left\{\bar\lambda_j,\bar m_j,\{\bar c_{jk}\}_{k=0}^{\bar m_j-1}\right\}_{j=1}^{\bar N}.
\end{equation}

\item[\text{\rm(k)}] The bound-state information at $t=0$
specified in \eqref{2.27}
can be organized by using a pair of matrix triplets $(A,B,C)$ and $(\bar A,\bar B,\bar C),$
by letting
\begin{equation}\label{2.28}
A:=\begin{bmatrix}
A_1&0&\cdots&0&0\\
0&A_2&\cdots&0&0\\
\vdots&\vdots&\ddots&\vdots&\vdots\\
0&0&\cdots&A_{N-1}&0\\
0&0&\cdots&0&A_N
\end{bmatrix},
\quad
\bar A:=\begin{bmatrix}
\bar A_1&0&\cdots&0&0\\
0&\bar A_2&\cdots&0&0\\
\vdots&\vdots&\ddots&\vdots&\vdots\\
0&0&\cdots&\bar A_{\bar N-1}&0\\
0&0&\cdots&0&\bar A_{\bar N}
\end{bmatrix},
\end{equation}
\begin{equation}\label{2.29}
B=\begin{bmatrix}
B_1\\
B_2\\
\vdots\\
B_N
\end{bmatrix},\quad \bar B=\begin{bmatrix}
\bar B_1\\
\bar B_2\\
\vdots\\
\bar B_{\bar N}
\end{bmatrix},
\end{equation}
\begin{equation}\label{2.30}
C:=\begin{bmatrix}
C_1&C_2&\cdots&C_N
\end{bmatrix},
\quad 
\bar C:=\begin{bmatrix}
\bar C_1&\bar C_2&\cdots&\bar C_{\bar N}
\end{bmatrix},
\end{equation}
where we have defined
\begin{equation}\label{2.31}
A_j:=\begin{bmatrix}
\lambda_j&1&0&\cdots&0&0\\
0&\lambda_j&1&\cdots&0&0\\
0&0&\lambda_j&\cdots&0&0\\
\vdots&\vdots&\vdots&\ddots&\vdots&\vdots\\
0&0&0&\cdots&\lambda_j&1\\
0&0&0&\dots&0&\lambda_j
\end{bmatrix},\qquad 1\le j\le N,
\end{equation}
\begin{equation}\label{2.32}
B_j:=\begin{bmatrix}
0\\ \vdots \\
0\\
1
\end{bmatrix},\quad C_j:=\begin{bmatrix}
c_{j(m_j-1)}&c_{j(m_j-2)}&\cdots&c_{j1}&c_{j0}
\end{bmatrix},\qquad 1\le j\le N,
\end{equation}
with $A_j$ being the $m_j\times m_j$ square matrix in the Jordan canonical form with
$\lambda_j$ appearing in the diagonal entries, $B_j$ being the column vector with $m_j$ components
that are all zero except for the last entry which is $1,$ and $C_j$ being the row vector with $m_j$ components
containing all the norming constants in the indicated order, and
\begin{equation}\label{2.33}
\bar A_j:=\begin{bmatrix}
\bar\lambda_j&1&0&\cdots&0&0\\
0&\bar\lambda_j&1&\cdots&0&0\\
0&0&\bar\lambda_j&\cdots&0&0\\
\vdots&\vdots&\vdots&\ddots&\vdots&\vdots\\
0&0&0&\cdots&\bar\lambda_j&1\\
0&0&0&\dots&0&\bar\lambda_j
\end{bmatrix},\qquad 1\le j\le {\bar N},
\end{equation}
\begin{equation}\label{2.34}
\bar B_j:=\begin{bmatrix}
0\\ \vdots \\
0\\
1
\end{bmatrix},\quad \bar C_j:=\begin{bmatrix}
\bar c_{j(\bar m_j-1)}&\bar c_{j(\bar m_j-2)}&\cdots&\bar c_{j1}&\bar c_{j0}
\end{bmatrix},\qquad 1\le j\le {\bar N},
\end{equation}
with $\bar A_j$ being the $\bar m_j\times\bar m_j$ square matrix in the Jordan canonical form with
$\bar\lambda_j$ appearing in the diagonal entries, $\bar B_j$ being the column vector with $\bar m_j$ components
that are all zero except for the last entry which is $1,$ and $\bar C_j$ being the row vector with $\bar m_j$ components
containing all the norming constants in the indicated order.

\item[\text{\rm(l)}] The four matrices $A,$ $\bar A,$ $B,$ $\bar B$ and each of 
the four matrices $A_j,$ $\bar A_j,$ $B_j,$ $\bar B_j$ are all independent of $t,$ and 
hence their values at any time $t$ coincide with 
the corresponding values at $t=0.$
On the other hand, the norming constants evolve in time as described by
\begin{equation}\label{2.35}
C\mapsto C e^{4i A^2 t},\quad \bar C\mapsto \bar C e^{-4i \bar A^2 t},\quad 
\end{equation}
\begin{equation}\label{2.36}
C_j\mapsto C_j e^{4i A_j^2 t},\quad \bar C_j\mapsto \bar C_j e^{-4i \bar A_j^2 t}.
\end{equation}

\end{enumerate}
	
\end{theorem}

\begin{proof} The proofs of all parts of the theorem can be found in the doctoral thesis \cite{E2018} of
the second author. As an alternative proof, we remark
that the results presented in the theorem for each fixed $t\in\mathbb R$ have similar
proofs when $t=0.$ 
Hence, for the proofs of (a) and (b) we refer the reader to
Theorem~2.2 of \cite{AE2022b};
for the proofs of (c), (d), and (e) to
Theorem~2.5 of \cite{AE2022b};
for the proofs of (h), (i), (j), and (k)
to Section~3 of \cite{AE2022b};
and for the proofs of \eqref{2.16}, \eqref{2.23},
\eqref{2.24}, (e), (f), (j), and (l) to
\cite{E2018}.
\end{proof}

The scattering data set $\mathbf S(\zeta,0)$ at $t=0$ consists of the six scattering coefficients
and the relevant bound-state information. As seen from Theorem~\ref{theorem2.1}(g), we can suppress
the left reflection coefficients from the specification of
$\mathbf S(\zeta,0).$
Furthermore, as seen from (j) and (k) of
Theorem~\ref{theorem2.1}, the knowledge of the bound-state
information contained in the two sets in \eqref{2.27}
is equivalent to the knowledge of
the pair of matrix triplets
$(A,B,C)$ and $(\bar A,\bar B,\bar C).$ 
Hence, we can define the scattering data set $\mathbf S(\zeta,0)$ at $t=0$ for the unperturbed
linear system \eqref{2.1} as
\begin{equation}\label{2.37}
\mathbf S(\zeta,0):=\{T(\zeta,0),\bar T(\zeta,0),R(\zeta,0),\bar R(\zeta,0),(A,B,C),(\bar A,\bar B,\bar C)\}.
\end{equation}
Then, the time-evolved scattering data set $\mathbf S(\zeta,t)$ is obtained from
$\mathbf S(\zeta,0)$ by using \eqref{2.16}, \eqref{2.23}, and \eqref{2.35}
with the understanding that the matrices $A,$ $B,$ $\bar A,$ $\bar B$ are unchanged in time.
Thus, $\mathbf S(\zeta,t)$ for \eqref{2.1} can equivalently be described as
\begin{equation}\label{2.38}
\mathbf S(\zeta,t)=\{T(\zeta,t),\bar T(\zeta,t),R(\zeta,t) ,\bar R(\zeta,t),(A,B,C\, e^{4iA^2 t}),(\bar A,\bar B,\bar C\, e^{-4i\bar A^2 t})\},
\end{equation}
where we recall that $\lambda$ in \eqref{2.23} is related to $\zeta$ as in \eqref{2.12}.

Let us remark on the simplicity and elegance of the use of the matrix triplet pair
$(A,B,C)$ and $(\bar A,\bar B,\bar C)$ rather than the two sets in \eqref{2.27} when 
the bound states have multiplicities. As seen from \eqref{2.36}, 
the time evolutions of the  
norming constants $c_{jk}$ and $\bar c_{jk}$ appearing in
\eqref{2.27} are too complicated to express individually although
the time evolutions of the vectors $C_j$ and $\bar C_j$ are very simple
when expressed in terms of matrix exponentials. Similarly, as seen from \eqref{2.35}, the 
time evolutions of the vectors $C$ and $\bar C$ are also very simple
when expressed in terms of matrix exponentials. Thus, 
the use of the matrix triplet pair allows us to describe in a simple and elegant manner
the bound states and
the time evolutions of the bound-state norming constants
no matter how many bound states we have and no matter what their multiplicities are.

\section{The Marchenko method for the unperturbed system}
\label{section3}

In \cite{AE2022b} we have developed the Marchenko method for the linear system \eqref{2.1} when the potentials
$q(x,t)$ and $r(x,t)$ are independent of the parameter $t.$ In this section we present the extension of
the Marchenko method developed in \cite{AE2022b} from the time-independent case to the
time-evolved case, and this is done by providing the appropriate time evolution of the scattering data for \eqref{2.1}.
In the Marchenko method for \eqref{2.1}, the potentials $q(x,t)$ and $r(x,t)$ are recovered from the 
scattering data set consisting of the time-evolved scattering coefficients and the time-evolved bound-state
information. The input is used to construct the kernel in the Marchenko system of linear integral equations
as well as the nonhomogeneous term in the Marchenko system. The potentials and all other relevant
quantities associated with \eqref{2.1} are then recovered from the solution to the Marchenko system.

In the next theorem, the Marchenko system of linear integral equations for \eqref{2.1} is presented, and
the resulting Marchenko system is shown to be equivalent to an uncoupled system of
Marchenko integral equations.

\begin{theorem}
\label{theorem3.1}
Assume that the potentials $q(x,t)$ and $r(x,t)$ in \eqref{2.1} belong to the Schwartz class for each fixed $t\in\mathbb R.$  
Let $R(\zeta,t)$ and $\bar R(\zeta,t)$ be the corresponding  time-evolved reflection coefficients
appearing in \eqref{2.23}. Let $(A,B,C)$ and $(\bar A,\bar B,\bar C)$ be the pair of matrix triplets representing the bound-state information
for \eqref{2.1}, where the matrices are described in
Section~\ref{section2} and contained in the initial scattering data set in \eqref{2.37}.
We have the following:

\begin{enumerate}

\item[\text{\rm(a)}] 
The Marchenko system of linear integral equations for
\eqref{2.1} is
given by
\begin{equation}\label{3.1}
\begin{split}
\begin{bmatrix}
0&0\\ \noalign{\medskip}0&0
\end{bmatrix}=&\begin{bmatrix}
\bar K_1(x,y,t)&K_1(x,y,t)\\ \noalign{\medskip}\bar K_2(x,y,t)&K_2(x,y,t)
\end{bmatrix}+ \begin{bmatrix}
0&\bar\Omega(x+y,t)\\ \noalign{\medskip}\Omega(x+y,t)&0
\end{bmatrix}\\
\noalign{\medskip}
&+\ds\int_x^\infty dz\begin{bmatrix}
-iK_1(x,z,t)\,\Omega'(z+y,t)&\bar K_1(x,z,t)\,\bar\Omega(z+y,t)\\ \noalign{\medskip}
K_2(x,z,t)\,\Omega(z+y,t)&i\bar K_2(x,z,t)\,\bar\Omega'(z+y,t)
\end{bmatrix},\qquad x<y,
\end{split}
\end{equation}
where $\Omega(y,t)$ and $\bar\Omega(y,t)$ are the quantities
defined as
\begin{equation}\label{3.2}
\Omega(y,t):=\hat R(y,t)+C\,e^{4iA^2 t}\,e^{iAy}\,B,\quad \bar\Omega(y,t):=\hat{\bar R}(y,t)+\bar C\,e^{-4i\bar A^2 t}\,e^{-i\bar A y}\,\bar B,
\end{equation}
with the prime in $\Omega'(y,t)$ and
$\bar\Omega'(y,t)$ denoting the $y$-derivatives, and where we have
\begin{equation}\label{3.3}
\hat R(y,t):=\ds\frac{1}{2\pi}\ds\int_{-\infty}^\infty  
d\lambda\,\ds\frac{R(\zeta,t)}{\zeta}\,e^{i\lambda y},\quad \hat{\bar R}(y,t):=\ds\frac{1}{2\pi}
\ds\int_{-\infty}^\infty  d\lambda\,\ds\frac{\bar R(\zeta,t)}{\zeta}\,e^{-i\lambda y},
\end{equation}
\begin{equation}\label{3.4}
K_1(x,y,t):= 
\ds\frac{1}{2\pi }\int_{-\infty}^\infty d\lambda \left[\ds\frac{e^{i\mu/2}\,\,\psi_1(\zeta,x,t)}{\zeta\,E(x,t)}\right] e^{-i\lambda y},
\end{equation}
\begin{equation}\label{3.5}
K_2(x,y,t):= 
\ds\frac{1}{2\pi }\int_{-\infty}^\infty d\lambda \left[e^{-i\mu/2}\,E(x,t)\,\psi_2(\zeta,x,t)-e^{i\lambda x}\right] e^{-i\lambda y},
\end{equation}
\begin{equation}\label{3.6}
\bar K_1(x,y,t):= 
\ds\frac{1}{2\pi }\int_{-\infty}^\infty 
d\lambda \left[\ds\frac{e^{i\mu/2}\,\bar{\psi}_1(\zeta,x,t)}{E(x,t)}-e^{-i\lambda x}\right] e^{i\lambda y},
\end{equation}
\begin{equation}\label{3.7}
\bar K_2(x,y,t):= 
\ds\frac{1}{2\pi }\int_{-\infty}^\infty d\lambda
\left[\ds\frac{e^{-i\mu/2}\,E(x,t)\,\bar{\psi}_2(\zeta,x,t)}{\zeta}\right] e^{i\lambda y},
\end{equation}
with $\lambda$ being related to $\zeta$ as in 
\eqref{2.12}, $E(x,t)$ and $\mu$ being the quantities defined in \eqref{1.19} and \eqref{2.14}, 
respectively, and $\psi_1(\zeta,x,t),$ $\psi_2(\zeta,x,t),$ $\bar{\psi}_1(\zeta,x,t),$ and $\bar{\psi}_2(\zeta,x,t)$ being the 
components of the Jost solutions given in \eqref{2.10}.

\item[\text{\rm(b)}] 
The coupled Marchenko system \eqref{3.1} is equivalent to the uncoupled system
of  equations
\begin{equation}\label{3.8}
\begin{cases}
	K_1(x,y,t)+\bar\Omega(x+y,t)+i\ds\int_x^\infty dz\int_x^\infty 
	ds\,K_1(x,z,t)\,\Omega'(z+s,t)\,\bar\Omega(s+y,t)=0,
	\\
	\noalign{\medskip}
	\bar K_2(x,y,t)+\Omega(x+y,t)-i\ds\int_x^\infty dz\int_x^\infty 
	ds\,\bar K_2(x,z,t)\,\bar\Omega'(z+s,t)\,\Omega(s+y,t)=0,
	\end{cases}
	\end{equation}
where $x<y,$ with the auxiliary equations given by 
	\begin{equation}\label{3.9}
	\begin{cases}
	\bar K_1(x,y,t)=i\ds\int_x^\infty dz\,K_1(x,z,t)\,\Omega'(z+y,t),\qquad x<y,
	\\ \noalign{\medskip}
	K_2(x,y,t)=-i\ds\int_x^\infty dz\,\bar K_2(x,z,t)\,\bar\Omega'(z+y,t),\qquad x<y.
	\end{cases}
	\end{equation}

\end{enumerate}

\end{theorem}

\begin{proof} The time
evolutions
$\Omega(y,0)\mapsto \Omega(y,t)$ and $\bar\Omega(y,0)\mapsto\bar\Omega(y,t)$ directly follow from 
\eqref{2.23} and \eqref{2.35}.
Since the results are stated for each fixed $t\in\mathbb R,$ 
the result in (a) then follows from Theorem~4.2 of \cite{AE2022b}
and the result in (b) follows from (4.41) and
(4.42) of \cite{AE2022b}.
\end{proof}

In the next theorem, we describe the recovery of the relevant quantities for \eqref{2.1} from the solution to the
Marchenko system \eqref{3.1}, which uses the time-evolved scattering data set $\mathbf S(\zeta,t)$ of \eqref{2.38} as input. Those relevant quantities consist
of the key quantity $E(x,t)$ in \eqref{1.19}, the constant $\mu$
in \eqref{2.14}, the two potentials $q(x,t)$ and $r(x,t)$ in \eqref{2.1}, and
the four Jost solutions to \eqref{2.1}.

\begin{theorem}
\label{theorem3.2}
Let the potentials $q(x,t)$ and $r(x,t)$ in \eqref{2.1} belong to the Schwartz class for each fixed $t\in\mathbb R.$ The relevant quantities
are recovered from the solution to the
Marchenko system \eqref{3.1} or equivalently from the uncoupled counterpart given in \eqref{3.8} and \eqref{3.9}
as follows:

\begin{enumerate}

\item[\text{\rm(a)}] The scalar quantity
$E(x,t)$ defined in \eqref{1.19} is recovered from
the solution to the Marchenko system by using
\begin{equation}\label{3.10}
E(x,t)=\exp\left(2\ds\int_{-\infty}^{x}dz\,Q(z,t)\right),
\end{equation}
where $Q(x,t)$ is the auxiliary scalar quantity constructed from $\bar K_1(x,x,t)$ and $K_2(x,x,t)$ as
\begin{equation}\label{3.11}
Q(x,t):=\bar K_1(x,x,t)-K_2(x,x,t).
\end{equation}

\item[\text{\rm(b)}] The complex-valued scalar constant $\mu$ defined in \eqref{2.14} is
obtained from the solution to the Marchenko system as
\begin{equation}\label{3.12}
\mu=-4i\ds\int_{-\infty}^\infty dz\,Q(z,t).
\end{equation}
Since the value of $\mu$ is independent of $t,$ \eqref{3.12} is equivalent to
\begin{equation}\label{3.13}
\mu=-4i\ds\int_{-\infty}^\infty dz\,Q(z,0).
\end{equation}

\item[\text{\rm(c)}] The potentials $q(x,t)$ and $r(x,t)$ are recovered from the solution to the Marchenko system as
\begin{equation}\label{3.14}
q(x,t)=-2K_1(x,x,t)\exp\left(-4\ds\int_x^\infty dz\,Q(z,t)\right),
\end{equation}
\begin{equation}\label{3.15}
r(x,t)=-2\bar K_2(x,x,t)\exp\left(4\ds\int_x^\infty dz\,Q(z,t)\right).
\end{equation}

\item[\text{\rm(d)}] The Jost solutions $\psi(\zeta,x,t)$ and $\bar\psi(\zeta,x,t)$ 
to \eqref{2.1} are recovered 
from the solution to the Marchenko system as
\begin{equation}\label{3.16}
\psi_1(\zeta,x,t)= \zeta\left(\ds\int_x^\infty dy\,K_1(x,y,t)\,e^{i\zeta^2y}\right)
\exp\left(-2\ds\int_x^\infty dz\,Q(z,t)
\right),
\end{equation}
\begin{equation}\label{3.17}
\psi_2(\zeta,x,t)=\left(e^{i\zeta^2x}+\ds\int_x^\infty dy\,K_2(x,y,t)\,e^{i\zeta^2 y}\right) \exp\left(2\ds\int_x^\infty dz\,Q(z,t)
\right),
\end{equation}
\begin{equation}\label{3.18}
\bar{\psi}_1(\zeta,x,t)=\left (e^{-i\zeta^2x}+\ds\int_x^\infty dy\,\bar K_1(x,y,t)\,e^{-i\zeta^2 y}\right)
\exp\left(-2\ds\int_x^\infty dz\,Q(z,t)
\right),
\end{equation}
\begin{equation}\label{3.19}
\bar{\psi}_2(\zeta,x,t)= \zeta\left(\ds\int_x^\infty dy\,\bar K_2(x,y,t)\,e^{-i\zeta^2y}\right)
\exp\left(2\ds\int_x^\infty dz\,Q(z,t)\right),
\end{equation}
where $\psi_1(\zeta,x,t),$ $\psi_2(\zeta,x,t),$ $\bar{\psi}_1(\zeta,x,t),$ and $\bar{\psi}_2(\zeta,x,t)$ 
are the components of the Jost solutions appearing in \eqref{2.10}.

\item[\text{\rm(e)}] The Jost solutions $\phi(\zeta,x,t)$ and $\bar\phi(\zeta,x,t)$ 
to \eqref{2.1} are recovered with the help of
\begin{equation}\label{3.20}
\phi(\zeta,x,t)=\ds\frac{1}{T(\zeta,t)}\,\bar\psi(\zeta,x,t)+
\ds\frac{R(\zeta,t)}{T(\zeta,t)}\,\psi(\zeta,x,t),
\end{equation}
\begin{equation}\label{3.21}
\bar\phi(\zeta,x,t)=\ds\frac{\bar R(\zeta,t)}{\bar T(\zeta,t)}\,\bar\psi(\zeta,x,t)+
\ds\frac{1}{\bar T(\zeta,t)}\,\psi(\zeta,x,t),
\end{equation}
where on the right-hand sides we use $T(\zeta,t),$ $\bar T(\zeta,t),$ $R(\zeta,t),$
$\bar R(\zeta,t)$ from the scattering data set described in \eqref{2.38} and
use the Jost solutions $\psi(\zeta,x,t)$ and $\bar\psi(\zeta,x,t)$ with components
expressed in \eqref{3.16}--\eqref{3.19}, respectively.

\end{enumerate}
\end{theorem}

\begin{proof}
Since the results pertaining to the Marchenko system are obtained for each fixed $t\in\mathbb R,$ those
results in (a)--(d) directly follow from Theorem~4.4 of \cite{AE2022b}, where the Marchenko theory has been developed in
the time-independent case. 
In particular, from  (4.51) and (4.55) of \cite{AE2022b} we get
\eqref{3.10} and \eqref{3.12}, respectively. 
It is already known from (4.63) of \cite{AE2022b} that the right-hand side of \eqref{3.11} yields
\begin{equation}\label{3.22}
\bar K_1(x,x,t)-K_2(x,x,t)=\ds\frac{i}{4}\,q(x,t)\,r(x,t),
\end{equation}
and hence \eqref{3.10} and \eqref{3.12} are compatible with
\eqref{1.19} and \eqref{2.14}, respectively.
From the
time-independent formulas (4.43) and (4.46) of \cite{AE2022b}, we have
\begin{equation}\label{3.23}
q(x,t)=-2  \,K_1(x,x,t)\,e^{-i\mu}\,E(x,t)^2,\quad
r(x,t)=-2\,\bar K_2(x,x,t)\,e^{i\mu}\,E(x,t)^{-2}.
\end{equation}
Also, from \eqref{3.10} and \eqref{3.12} we get
\begin{equation}\label{3.24}
e^{-i\mu}\,E(x,t)^2=\exp\left(-4\ds\int_x^\infty dz\,Q(z,t)\right),\quad
e^{i\mu}\,E(x,t)^{-2}=\exp\left(4\ds\int_x^\infty dz\,Q(z,t)\right).
\end{equation}
Consequently, using \eqref{3.23} and \eqref{3.24} we obtain \eqref{3.14} and \eqref{3.15}.
The derivation of the equalities in \eqref{3.16}--\eqref{3.19} are established in a similar manner.
The equalities in \eqref{3.20} and \eqref{3.21}
follow from the fact that the two Jost solutions $\psi(\zeta,x,t)$ and
$\bar\psi(\zeta,x,t)$ form a fundamental set
for \eqref{2.1} and the other two
Jost solutions $\phi(\zeta,x,t)$ and
$\bar\phi(\zeta,x,t)$
can be obtained as linear combinations of the fundamental
set of solutions with the help
of \eqref{2.8} and \eqref{2.9}.
\end{proof}
From \eqref{3.14} and \eqref{3.15} it follows that
\begin{equation}\label{3.25}
q(x,t)\,r(x,t)= 4 \, K_1(x,x,t)\,\bar K_2(x,x,t).
\end{equation}
Let us define the auxiliary quantity $P(x,t)$ as
\begin{equation}\label{3.26}
P(x,t):=K_1(x,x,t)\,\bar K_2(x,x,t).
\end{equation}
From \eqref{3.11}, \eqref{3.22}, \eqref{3.25}, and \eqref{3.26}, we see that
the quantity $Q(x,t)$ defined in \eqref{3.11} is related to
$P(x,t)$ as
\begin{equation}\label{3.27}
Q(x,t)=i\,P(x,t).
\end{equation}

In the next corollary, we state the results of Theorem~\ref{3.2} in terms of $P(x,t).$
Let us emphasize that the quantity
$Q(x,t)$ is constructed from $\bar K_1(x,x,t)$ and $K_2(x,x,t),$
whereas the quantity $P(x,t)$ is constructed from
$K_1(x,x,t)$ and $\bar K_2(x,x,t).$
The corollary provides an alternate way to recover the relevant quantities for \eqref{2.1} from the
solution to the corresponding Marchenko system \eqref{3.1}.

\begin{corollary}
\label{corollary3.3}
Let the potentials $q(x,t)$ and $r(x,t)$ in \eqref{2.1} belong to the Schwartz class for each fixed $t\in\mathbb R.$ The relevant quantities for \eqref{2.1}
are recovered from the solution to the
Marchenko system \eqref{3.1}, or equivalently from the uncoupled counterpart given in \eqref{3.8} and \eqref{3.9},
as follows:

\begin{enumerate}

\item[\text{\rm(a)}] The scalar quantity
$E(x,t)$ defined in \eqref{1.19} is recovered from
the solution to the Marchenko system by using
\begin{equation}\label{3.28}
E(x,t)=\exp\left(2i\ds\int_{-\infty}^{x}dz\,P(z,t)\right),
\end{equation}
where $P(x,t)$ is the auxiliary scalar quantity constructed from $K_1(x,x,t)$ and $\bar K_2(x,x,t)$ as
in \eqref{3.26}.

\item[\text{\rm(b)}] The complex-valued scalar constant $\mu$ defined in \eqref{2.14} is
obtained from the solution to the Marchenko system as
\begin{equation*}
\mu=4\ds\int_{-\infty}^\infty dz\,P(z,t).
\end{equation*}

\item[\text{\rm(c)}] The potentials $q(x,t)$ and $r(x,t)$ are recovered from the solution to the Marchenko system as
\begin{equation}\label{3.29}
q(x,t)=-2K_1(x,x,t)\exp\left(-4i \ds\int_x^\infty dz\,P(z,t)\right),
\end{equation}
\begin{equation}\label{3.30}
r(x,t)=-2\bar K_2(x,x,t)\exp\left(4i\ds\int_x^\infty dz\,P(z,t)\right).
\end{equation}

\item[\text{\rm(d)}] The Jost solutions $\psi(\zeta,x,t)$ and $\bar\psi(\zeta,x,t)$ 
to \eqref{2.1} are recovered 
from the solution to the Marchenko system as
\begin{equation}\label{3.31}
\psi_1(\zeta,x,t)= \zeta\left(\ds\int_x^\infty dy\,K_1(x,y,t)\,e^{i\zeta^2y}\right)
\exp\left(-2i\ds\int_x^\infty dz\,P(z,t)
\right),
\end{equation}
\begin{equation}\label{3.32}
\psi_2(\zeta,x,t)=\left(e^{i\zeta^2x}+\ds\int_x^\infty dy\,K_2(x,y,t)\,e^{i\zeta^2 y}\right) \exp\left(2i\ds\int_x^\infty dz\,P(z,t)
\right),
\end{equation}
\begin{equation}\label{3.33}
\bar{\psi}_1(\zeta,x,t)=\left (e^{-i\zeta^2x}+\ds\int_x^\infty dy\,\bar K_1(x,y,t)\,e^{-i\zeta^2 y}\right)
\exp\left(-2i\ds\int_x^\infty dz\,P(z,t)
\right),
\end{equation}
\begin{equation}\label{3.34}
\bar{\psi}_2(\zeta,x,t)= \zeta\left(\ds\int_x^\infty dy\,\bar K_2(x,y,t)\,e^{-i\zeta^2y}\right)
\exp\left(2i\ds\int_x^\infty dz\,P(z,t)\right),
\end{equation}
where $\psi_1(\zeta,x,t),$ $\psi_2(\zeta,x,t),$ $\bar{\psi}_1(\zeta,x,t),$ and $\bar{\psi}_2(\zeta,x,t)$ 
are the components of the Jost solutions appearing in \eqref{2.10}.

\item[\text{\rm(e)}] The Jost solutions $\phi(\zeta,x,t)$ and $\bar\phi(\zeta,x,t)$ 
to \eqref{2.1} are recovered with the help of
\eqref{3.20} and \eqref{3.21} by using $T(\zeta,t),$ $\bar T(\zeta,t),$ $R(\zeta,t),$
$\bar R(\zeta,t)$ from the scattering data set described in \eqref{2.38} and also
using the Jost solutions $\psi(\zeta,x,t)$ and $\bar\psi(\zeta,x,t)$ with components
expressed as in \eqref{3.31}--\eqref{3.34}, respectively.

\end{enumerate}
\end{corollary}

We have the following remarks for Corollary~\ref{corollary3.3}. At first sight, it might seem as if the results in
Corollary~\ref{corollary3.3} are trivial because they are obtained from
Theorem~\ref{theorem3.2} simply by replacing $Q(x,t)$ by $i P(x,t).$
However, Corollary~\ref{corollary3.3} presents an alternate way to recover the relevant
quantities from the solution to the uncoupled Marchenko system \eqref{3.8}.
In fact, when we solve \eqref{3.8} we obtain $K_1(x,y,t)$ and $\bar K_2(x,y,t),$ from which we
can construct $P(x,t)$ by using \eqref{3.26} without having to solve the
auxiliary system \eqref{3.9}.  As seen from \eqref{3.29} and \eqref{3.30}, we can then
also construct the quantities $q(x,t)$ and $r(x,t)$ without having to solve
the
auxiliary system \eqref{3.9}.  The recovery of the quantities $E(x,t),$ $\mu,$ $q(x,t),$ and $r(x,t)$ 
at times
may be more efficient by using the procedure of
Corollary~\ref{corollary3.3} rather than the procedure of Theorem~\ref{theorem3.2} 
because the latter requires the solution to \eqref{3.9}.
On the other hand, even though we have the equality in \eqref{3.27}, the
evaluation of the integral of $Q(x,t)$ may be easier than 
the evaluation of the integral of $P(x,t)$ because the former
integrand is expressed as a difference and the latter integrand as a product,
as seen from \eqref{3.11} and \eqref{3.26}, respectively.

An advantage of using Corollary~\ref{corollary3.3} instead of Theorem~\ref{theorem3.2} becomes apparent
in Section~\ref{section8} when the two potentials $q(x,t)$ and $r(x,t)$
are related to each other via complex conjugation as in \eqref{8.1}.
In that case, the potential $q(x,t)$ can be recovered directly
from the solution to the scalar Marchenko equation \eqref{8.13}.
As we see from the recovery formula given in \eqref{8.14}, no other quantities are needed in the recovery
besides the solution to the scalar Marchenko equation \eqref{8.13}.

\section{Explicit solution formulas for the unperturbed system}
\label{section4}

When the potentials $q(x,t)$ and $r(x,t)$ at $t=0$ in the
unperturbed linear system \eqref{2.1} are reflectionless, i.e. when
the reflection coefficients at $t=0$ are all zero,
we see
from \eqref{2.23} and \eqref{2.24} that the time-evolved reflection
coefficients remain zero for all $t.$ In that case,
the quantities $\Omega(y,t)$ and $\bar\Omega(y,t)$ used as 
input to the Marchenko system \eqref{3.1} yield separable integral kernels.
This 
results in explicit solutions to the Marchenko system \eqref{3.1} and also in explicit solutions to
the unperturbed linear system
\eqref{2.1},
where all those solutions are expressed in terms of the
matrix triplet pair $(A,B,C)$ and
$(\bar A,\bar B,\bar C)$ appearing in \eqref{3.2}.
In this section, we provide the corresponding explicit solution formulas for the unperturbed linear system
\eqref{2.1}.
Certainly, having the solution formulas for the potentials
$q(x,t)$ and $r(x,t)$ for the linear system \eqref{2.1}, we
also have the solution formulas for
the corresponding nonlinear system \eqref{1.2}.

When the Marchenko kernels $\Omega(y,t)$ and $\bar\Omega(y,t)$ correspond to a reflectionless scattering data set, from \eqref{3.2} we get
\begin{equation}\label{4.1}
\Omega(y,t)=C\,e^{iAy+4iA^2t}\,B,\quad 
\bar\Omega(y,t)=\bar C\,e^{-i\bar A y-4i\bar A^2 t}\,\bar B,
\end{equation}
\begin{equation}\label{4.2}
\Omega'(y,t)=i\,C\,A\,e^{iAy+4iA^2t}\,B,\quad 
\bar\Omega'(y,t)=-i\,\bar C\bar A\,e^{-i\bar A y-4i\bar A^2 t}\,\bar B,
\end{equation}
which are all explicitly expressed in terms of the matrix triplets
$(A,B,C)$ and
$(\bar A,\bar B,\bar C)$ appearing in \eqref{2.38}.
In the next theorem we obtain the solution to the Marchenko system \eqref{3.1}
with the input given in \eqref{4.1}.

\begin{theorem}
\label{theorem4.1}
When the time-evolved reflectionless quantities $\Omega(y,t)$ and $\bar\Omega(y,t)$ appearing
in \eqref{4.1} are used as input to the Marchenko system \eqref{3.1},
the resulting system of integral equations is solvable in closed form and has the solution
explicitly expressed in terms of the
matrix triplet pair $(A,B,C)$ and
$(\bar A,\bar B,\bar C)$ as
\begin{equation}\label{4.3}
K_1(x,y,t)
=-\bar C\,e^{-i\bar A x}\,\bar\Gamma(x,t)^{-1}\,e^{-i\bar A y-4i\bar A^2 t}\,\bar B,
\end{equation}
\begin{equation}\label{4.4}
K_2(x,y,t)
=C\,e^{iAx}\,\Gamma(x,t)^{-1}\,e^{iAx+4iA^2t} M \bar A e^{-i\bar A (x+y)-4i\bar A^2 t}\,\bar B,
\end{equation}
\begin{equation}\label{4.5}
\bar K_1(x,y,t)=\bar C\,e^{-i\bar A x}\,\bar\Gamma(x,t)^{-1}\,e^{-i\bar A x-4i\bar A^2 t}\bar M A\,e^{iA(x+y)+4iA^2t}\,B,
\end{equation}
\begin{equation}\label{4.6}
\bar K_2(x,y,t)=-C\,e^{iAx}\,\Gamma(x,t)^{-1}\,e^{iAy+4iA^2t}\,B,
\end{equation}
where $\Gamma(x,t)$ and $\bar\Gamma(x,t)$ are the matrix-valued functions of $x$ and $t$ defined as
\begin{equation}\label{4.7}
\Gamma(x,t):=I-e^{iAx+4iA^2 t}M\bar Ae^{-2i\bar A x-4i\bar A^2 t}\bar M e^{iAx}, 
\end{equation}
\begin{equation}\label{4.8}
\bar\Gamma(x,t):=I-e^{-i\bar A x-4i\bar A^2 t}\bar MAe^{2iAx+4iA^2 t}Me^{-i\bar A x},
\end{equation}
with $M$ and $\bar M$ being the matrix-valued constants defined as
\begin{equation}\label{4.9}	
M:=\int_{0}^\infty dz\,e^{iAz}\,B\,\bar C\,e^{-i\bar A z},\quad \bar M:=\int_{0}^\infty dz\,e^{-i\bar A z}\,\bar B\,C\,e^{i A z}.
\end{equation}
The constant matrices $M$ and $\bar M$ can alternatively be obtained as the unique solutions to
the respective linear systems
\begin{equation}
\label{4.10}
AM-M\,\bar A=iB\,\bar C,\quad 
\bar M A- \bar A\,\bar M=i\bar B\,C.
\end{equation}

\end{theorem}

\begin{proof}
From Theorem~\ref{theorem3.1} we know that the Marchenko system \eqref{3.1} is equivalent to the combination of the uncoupled system \eqref{3.8} and
 the auxiliary system \eqref{3.9}. To obtain \eqref{4.3} we proceed as follows. Using the second equality of
\eqref{4.1} and the first equality of \eqref{4.2} as input to the first line of \eqref{3.8} we get
\begin{equation}\label{4.11}
\begin{split}
K_1(x,y,t)&
+\bar C\,e^{-i\bar A(x+y)-4i\bar A^2 t}\,\bar B
\\&
+i\ds\int_x^\infty dz\int_x^\infty 
ds\,K_1(x,z,t)\left(i\,C\,A\,e^{iA(z+s)+4iA^2t}\,B\right)\bar C\,e^{-i\bar A(s+y)-4i\bar A^2 t}\,\bar B=0.
\end{split}
\end{equation}
Note  that $A,$ $e^{iAy},$ and $e^{4iA^2t}$ commute with each other, 
and also $\bar A,$ $e^{i\bar A y},$ and $e^{4i\bar A^2 t}$ commute with each other. From \eqref{4.11} we see that $K_1(x,y,t)$ has the form
\begin{equation}\label{4.12}
K_1(x,y,t)=H_1(x,t)\,e^{-i\bar A y-4i\bar A^2 t}\,\bar B,
\end{equation}
where $H_1(x,t)$ satisfies
\begin{equation}\label{4.13}
H_1(x,t)\left(I-\ds\int_x^\infty dz\int_x^\infty 
ds\,e^{-i\bar A z-4i\bar A^2 t}\,\bar B\,C\,e^{iAz}\,A\,e^{iAs+4iA^2t}\,B\bar C\,e^{-i\bar A s}\right)
=-\bar C\,e^{-i\bar A x}.
\end{equation}
Using \eqref{4.9} on the left-hand side of \eqref{4.13}, we obtain
\begin{equation*}
H_1(x,t)\left(I-e^{-i\bar A x-4i\bar A^2 t}\,\bar M\,e^{iAx}\,A\,e^{iAx+4iA^2t}\,M\,e^{-i\bar A x}\right)
=-\bar C\,e^{-i\bar A x},
\end{equation*}
or equivalently
\begin{equation}\label{4.14}
H_1(x,t)\,\bar\Gamma(x,t)=-\bar C\,e^{-i\bar A x},
\end{equation}
where $\bar\Gamma(x,t)$ is the matrix defined in \eqref{4.8}. From
\eqref{4.14} we have
\begin{equation}\label{4.15}
H_1(x,t)=-\bar C\,e^{-i\bar A x}\,\bar\Gamma(x,t)^{-1},
\end{equation}
and using \eqref{4.15} in \eqref{4.12} we get \eqref{4.3}. 
The solution formula for $\bar K_2(x,y,t)$ appearing in \eqref{4.6} is obtained in a 
similar manner from the second line of \eqref{3.8}. Then, using the first equality of \eqref{4.1} 
and the second equality of \eqref{4.2} in the second line of \eqref{3.8},
we have
\begin{equation*}
\begin{split}
\bar K_2(x,y,t)&
+C\,e^{iA(x+y)+4iA^2t}\,B
\\&
-i\ds\int_x^\infty dz\int_x^\infty 
ds\,\bar K_2(x,z,t)\left(-i\,\bar C\,\bar A\,e^{-i\bar A(z+s)-4i\bar A^2 t}\,\bar B\right)C\,e^{iA(s+y)+4iA^2t}\,B=0.
\end{split}
\end{equation*}
 From \eqref{4.16} we see that $\bar K_2(x,y,t)$ has the form
\begin{equation}\label{4.16}
\bar K_2(x,y,t)=H_2(x,t)\,e^{iAy+4iA^2t}\,B,
\end{equation}
where $H_2(x,t)$ satisfies
\begin{equation}\label{4.17}
H_2(x,t)\left(I-\ds\int_x^\infty dz\int_x^\infty 
ds\,e^{iAz+4iA^2t}\,B\,\bar C\,e^{-i\bar A z}\,\bar A\,e^{-i\bar A s-4i\bar A^2 t}\,\bar B\,C\,e^{iAs}\right)
=-C\,e^{iAx}.
\end{equation}
Using again \eqref{4.9} on the left-hand side of \eqref{4.17}, we write \eqref{4.17} as
\begin{equation*}
H_2(x,t)\left(I-e^{iAx+4iA^2 t}\,M\,e^{-i\bar A x}\,\bar A 
\,e^{-i\bar A x-4i\bar A^2 t}\,\bar M\,e^{iAx}\right)=-C\,e^{iAx},
\end{equation*}
which can be written as 
\begin{equation}\label{4.18}
H_2(x,t)\,\Gamma(x,t)=-C\,e^{iAx},
\end{equation}
where we remark that $\Gamma(x,t)$ is the matrix defined in \eqref{4.7}. From \eqref{4.18} we get
\begin{equation}\label{4.19}
H_2(x,t)=-C\,e^{iAx}\,\Gamma(x,t)^{-1},
\end{equation}
and using \eqref{4.19} in \eqref{4.16} we obtain \eqref{4.6}. Next, we present the proof of \eqref{4.4}. Using \eqref{4.6} and
the second equality of \eqref{4.2} in the second line of \eqref{3.9}, we have
\begin{equation}\label{4.20}
K_2(x,y,t)=\ds\int_x^\infty dz\,C\,e^{iAx}\,\Gamma(x,t)^{-1}\,e^{iAz+4iA^2t}\,B\,\bar C\,e^{-i\bar A z}\bar A\,e^{-i\bar A y-4i\bar A^2 t}\,\bar B.
\end{equation}
The use of the first equality of \eqref{4.9} in \eqref{4.20} results in
\begin{equation*}
 K_2(x,y,t)
 =C\,e^{iAx}\,\Gamma(x,t)^{-1}\,e^{iAx+4iA^2t}\,M\,e^{-i\bar A x}\,\bar A\,e^{-i\bar A y-4i\bar A^2 t}\,\bar B,
 \end{equation*}
which establishes \eqref{4.4}. To obtain \eqref{4.5}, we use  \eqref{4.3} and the first equality of \eqref{4.2} in the first line of \eqref{3.9}. This yields
\begin{equation}\label{4.21}
\bar K_1(x,y,t)=\ds\int_x^\infty dz\,\bar C\,e^{-i\bar A x}\,\bar\Gamma(x,t)^{-1}\,e^{-i\bar A z-4i\bar A^2 t}\,\bar B\,C\,e^{iAz}\,A\,e^{iAy+4iA^2t}\,B.
\end{equation}
The use of the second equality of \eqref{4.9} in \eqref{4.21} gives us
 \begin{equation*}
 \bar K_1(x,y,t)
 =\bar C\,e^{-i\bar A x}\,\bar\Gamma(x,t)^{-1}\,e^{-i\bar A x-4i\bar A^2 t}\,\bar M\,e^{iAx}\,A\,e^{iAy+4iA^2t}\,B,
\end{equation*}
which yields \eqref{4.5}. Let us finally prove that the constant matrices 
$M$ and $\bar M$ in \eqref{4.9} can also be obtained as the unique solutions to \eqref{4.10}.
From \eqref{4.9} we observe that 
\begin{equation}
\label{4.22}
iA\,M-M\,i\bar A=\ds\int_{0}^\infty dz\,\frac{d}{dz}\left(e^{i A z}\,B\,\bar C\,e^{-i \bar A z}\right),
\end{equation}
\begin{equation}
\label{4.23}
-i\bar A\,\bar M+\bar M\,iA=\ds\int_{0}^\infty dz\,\frac{d}{dz}\left(e^{-i\bar A z}\,\bar B\,C\,e^{iAz}\right).
\end{equation}
From \eqref{2.30}, \eqref{2.32}, and Theorem~\ref{theorem2.1}(i) we know that
the eigenvalues of $A$ are all in 
$\mathbb{C^+}$ and the eigenvalues of $\bar A$ are all
in $\mathbb{C^-}.$ Hence, the integrals in \eqref{4.22} and \eqref{4.23} both exist, and furthermore 
the right-hand side of \eqref{4.22} is equal to 
$-B \bar C$ and the right-hand side of \eqref{4.23} is equal to $-\bar B C.$ Thus, we have
\begin{equation}
\label{4.24}
i(AM-M\,\bar A)=-B\,\bar C,\quad 
i(\bar M A-\bar A\,\bar M)=-\bar B\,C,
\end{equation}
which are equivalent to \eqref{4.10}.
The unique solvability of the linear systems in \eqref{4.24} is assured
\cite{D2013} because the eigenvalues of $A$ are located
in $\mathbb C^+$ and
the eigenvalues of $\bar A$ are
in $\mathbb C^-.$
\end{proof}

In the next theorem, in the reflectionless case, we present some formulas for the potentials $q(x,t)$ and $r(x,t)$ in \eqref{2.1} and 
for the related key quantity $E(x,t),$ where those formulas
are expressed explicitly in terms of the matrix triplets
$(A,B,C)$ and
$(\bar A,\bar B,\bar C).$ 

\begin{theorem}
\label{theorem4.2}
Assume that the time-evolved reflectionless Marchenko kernels $\Omega(y,t)$ and $\bar\Omega(y,t)$ in \eqref{4.1}
are used as input to the Marchenko system \eqref{3.1}.  Then, we have the following:

\begin{enumerate}

\item[\text{\rm(a)}] The corresponding key quantity $E(x,t)$ defined in \eqref{1.19} is 
expressed explicitly in terms of the matrix triplets $(A,B,C)$ and $(\bar A,\bar B,\bar C),$ and we have
\begin{equation}\label{4.25}
E(x,t)=\exp\left(2 \ds\int_{-\infty}^x dz\,Q(z,t)
\right),
\end{equation}
where $Q(x,t)$ is the scalar-valued function of $x$ and $t$
defined in \eqref{3.11} with $\bar K_1(x,x,t)$ and 
$K_2(x,x,t)$  explicitly expressed in terms of the matrix triplets as
\begin{equation}
\label{4.26}
\bar K_1(x,x,t):=\bar C e^{-i\bar A x}\,\bar\Gamma(x,t)^{-1} e^{-i\bar A x-4i\bar A^2 t} \bar M A\,e^{2iAx+4iA^2t}B,
\end{equation}
\begin{equation}
\label{4.27}
K_2(x,x,t):=C e^{iAx}\,\Gamma(x,t)^{-1} e^{iAx+4iA^2t} M \bar A\,e^{-2i\bar A x-4i\bar A^2 t} \bar B,
\end{equation}
with $M$ and $\bar M$  being the constant matrices in \eqref{4.9},
and  $\Gamma(x,t)$ and $\bar\Gamma(x,t)$
being the matrix-valued functions of $x$ and $t$
defined in \eqref{4.7} and \eqref{4.8}, respectively.
Alternatively, we have
\begin{equation}\label{4.28}
E(x,t)=\exp\left(4i \ds\int_{-\infty}^x dz\,P(z,t)
\right),
\end{equation}
where $P(x,t)$ is the scalar-valued function of $x$ and $t$
defined in \eqref{3.26} with $K_1(x,x,t)$ and 
$\bar K_2(x,x,t)$  explicitly expressed in terms of the matrix triplets as
\begin{equation}\label{4.29}
K_1(x,x,t)
=-\bar C\,e^{-i\bar A x}\,\bar\Gamma(x,t)^{-1}\,e^{-i\bar A x-4i\bar A^2 t}\,\bar B,
\end{equation}
\begin{equation}\label{4.30}
\bar K_2(x,x,t)=-C\,e^{iAx}\,\Gamma(x,t)^{-1}\,e^{iAx+4iA^2t}\,B.
\end{equation}

\item[\text{\rm(b)}] 
The corresponding potentials $q(x,t)$ and $r(x,t)$ in the 
unperturbed linear
system \eqref{2.1} are expressed explicitly in terms of the matrix triplets $(A,B,C)$ and $(\bar A,\bar B,\bar C),$ and we have
\begin{equation}\label{4.31}
q(x,t)=\left(2\bar C\,e^{-i\bar A x}\,\bar\Gamma(x,t)^{-1}\,e^{-i\bar A x-4i\bar A^2 t}\,\bar B\right)\exp\left(-4 \ds\int_x^\infty dz\,Q(z,t)
\right),
\end{equation}
\begin{equation}\label{4.32}
r(x,t)=\left(2C\,e^{iA x}\,\Gamma(x,t)^{-1}\,e^{iA x+4i A^2 t}\,B\right)\exp\left(4 \ds\int_x^\infty dz\,Q(z,t)
\right),
\end{equation}
or alternatively we have
\begin{equation}\label{4.33}
q(x,t)=\left(2\bar C\,e^{-i\bar A x}\,\bar\Gamma(x,t)^{-1}\,e^{-i\bar A x-4i\bar A^2 t}\,\bar B\right)\exp\left(-4i \ds\int_x^\infty dz\,P(z,t)
\right),
\end{equation}
\begin{equation}\label{4.34}
r(x,t)=\left(2C\,e^{iA x}\,\Gamma(x,t)^{-1}\,e^{iA x+4i A^2 t}\,B\right)\exp\left(4i \ds\int_x^\infty dz\,P(z,t)
\right).
\end{equation}

\end{enumerate}

\end{theorem}

\begin{proof}
We note that \eqref{4.26}, \eqref{4.27}, \eqref{4.29}, and \eqref{4.30} are obtained
from \eqref{4.3}--\eqref{4.6} by using $y=x$ there.
We obtain \eqref{4.25} by using \eqref{3.10} and \eqref{3.11}
with the help of \eqref{4.26} and \eqref{4.27}.
Similarly, we obtain
\eqref{4.28} by using \eqref{3.26} and \eqref{3.28}
with the help of \eqref{4.29} and \eqref{4.30}. Hence, the proof of (a) is complete.
We get \eqref{4.31} from \eqref{3.14} with the help of
\eqref{3.11}, \eqref{4.26}, \eqref{4.27}, and \eqref{4.29}. Similarly, we get
\eqref{4.32} from \eqref{3.15} with the help of \eqref{3.11}, \eqref{4.26}, \eqref{4.27}, and
\eqref{4.30}.
The alternate expressions in \eqref{4.33} and \eqref{4.34} are obtained
in a similar manner, and this is done with the help of \eqref{3.26}, \eqref{4.29}, and \eqref{4.30}.
\end{proof}

In the next theorem we provide the explicit expressions for the transmission coefficients for \eqref{2.1}
corresponding to the reflectionless quantities in \eqref{4.1}.

\begin{theorem}
\label{theorem4.3}
Assume that the potentials $q(x,t)$ and $r(x,t)$ appearing in \eqref{2.1} at $t=0$ belong to the Schwartz class
and that the corresponding reflection coefficients $R(\zeta,0)$ and $\bar R(\zeta,0)$ are zero.
Let $T(\zeta,0)$ and $\bar T(\zeta,0)$ be the transmission coefficients in this reflectionless case.
Suppose that the corresponding bound-state information is given by the two sets
in \eqref{2.27}, or equivalently by the pair of matrix triplets
$(A,B,C)$ and $(\bar A,\bar B,\bar C)$
described in \eqref{2.28}--\eqref{2.30}.
Let the parameter $\lambda$ be related to the spectral parameter
$\zeta$ as in \eqref{2.12}.
Then, we have the following:

\begin{enumerate}

\item[\text{\rm(a)}] 
The time-evolved potentials $q(x,t)$ and $r(x,t)$ remain reflectionless for all $t\in\mathbb R.$

\item[\text{\rm(b)}] The transmission coefficients do not evolve in time, as
indicated in \eqref{2.16}.

\item[\text{\rm(c)}] 
The total number of poles of $T(\zeta,t)$ including multiplicities in the upper-half complex $\lambda$-plane
is equal to the total number of poles of
$\bar T(\zeta,t)$ including multiplicities in the lower-half complex $\lambda$-plane. In other words, we have
\begin{equation}
\label{4.35}
\mathcal N=\bar{\mathcal N},
\end{equation}
where we have defined
\begin{equation*}
\mathcal N:=\ds\sum_{j=1}^N m_j,\quad    \bar{\mathcal N}:=\ds\sum_{k=1}^{\bar N} \bar m_k,
\end{equation*}
and hence the matrices $A$ and $\tilde A$ in the
two matrix triplets have the same size $\mathcal N\times \mathcal N.$

\item[\text{\rm(d)}] 
The corresponding complex constant $e^{i\mu/2},$ where $\mu$ is the complex constant defined in
\eqref{2.14}, is uniquely determined by
the eigenvalues of the matrices $A$ and $\bar A$ and their
corresponding multiplicities. We have
\begin{equation}
\label{4.36}
e^{i\mu/2}=\ds\frac{\ds\prod_{k=1}^{\bar N}(\bar\lambda_k)^{\bar m_k}}{\ds\prod_{j=1}^N(\lambda_j)^{m_j}},
\end{equation}
with the restriction $\mathcal N=\bar{\mathcal N}.$
We note that \eqref{4.36} is equivalent to the determinant expression
\begin{equation}
\label{4.37}
e^{i\mu/2}=\det[\bar A\,A^{-1}].
\end{equation}

\item[\text{\rm(e)}] 
The transmission coefficients $T(\zeta,t)$ and
$\bar T(\zeta,t)$
are determined by
the eigenvalues of the matrices $A$ and $\bar A$ and their
corresponding multiplicities. We have
\begin{equation}
\label{4.38}
T(\zeta,t)=\left(\ds\frac{\ds\prod_{k=1}^{\bar N}\left( (\lambda/\bar\lambda_k)-1\right)^{\bar m_k}}
{\ds\prod_{j=1}^N\left( (\lambda/\lambda_j)-1\right)^{m_j}}\right),
\quad
\bar T(\zeta,t)=\left(\ds\frac{\ds\prod_{j=1}^N\left( (\lambda/\lambda_j)-1\right)^{m_j}}
{\ds\prod_{k=1}^{\bar N}\left( (\lambda/\bar\lambda_k)-1\right)^{\bar m_k}}\right),
\end{equation}
with the restriction $\mathcal N=\bar{\mathcal N}.$
We remark that \eqref{4.38} is equivalent to the pair of equations given by
\begin{equation}
\label{4.39}
T(\zeta,t)=\det[(\lambda\bar A^{-1}-I)(\lambda A^{-1}-I)^{-1}],
\quad
\bar T(\zeta,t)=\det[(\lambda A^{-1}-I)(\lambda \bar A^{-1}-I)^{-1}],
\end{equation}
and hence we have $\bar T(\zeta,t)=1/T(\zeta,t).$

\end{enumerate}

\end{theorem}

\begin{proof}
The proof of (a) is as follows. From \eqref{2.23} we see that $R(\zeta,t)$ and
$\bar R(\zeta,t)$ are both zero for $t\in\mathbb R$ whenever their values at $t=0$ are zero.
Then, from \eqref{2.26} we see that
$L(\zeta,t)$ and $\bar L(\zeta,t)$ are also zero for $t\in\mathbb R.$
The proof of (b) is apparent from \eqref{2.16}.
Since the transmission coefficients are independent of $t,$ the result in (c)
directly follows from Theorem~5.5(a) of \cite{AE2022b}.
For the proofs of the remaining items we proceed as follows.
In the reflectionless case, from \eqref{2.25} we have
\begin{equation}
\label{4.40}
T(\zeta,t)\,\bar T(\zeta,t)=1,\qquad \lambda\in\mathbb R.
\end{equation}
Note that \eqref{4.40} is equivalent to
\begin{equation}
\label{4.41}
T(\zeta,t)\,e^{i\mu/2}\,\left(\ds\frac{\ds\prod_{j=1}^N(\lambda-\lambda_j)^{m_j}}
{\ds\prod_{k=1}^{\bar N}(\lambda-\bar\lambda_k)^{\bar m_k}}\right)=\ds\frac{1}{\bar T(\zeta,t)\,e^{-i\mu/2}}
\,\left(\ds\frac{\ds\prod_{j=1}^N(\lambda-\lambda_j)^{m_j}}
{\ds\prod_{k=1}^{\bar N}(\lambda-\bar\lambda_k)^{\bar m_k}}\right),\qquad \lambda\in\mathbb R.
\end{equation}
From Theorem~\ref{theorem2.1} we know that
$T(\zeta,t)$ is meromorphic in $\lambda\in\mathbb C^+$ 
with poles at $\lambda=\lambda_j$ for $1\le j\le N$ each with
multiplicity $m_j,$ is continuous in $\lambda\in\mathbb R,$ and 
has the large $\lambda$-asymptotics given in \eqref{2.17}.
Again, from Theorem~\ref{theorem2.1} we know that the quantity
$1/\bar T(\zeta,t)$ is analytic in $\lambda\in\mathbb C^-$ 
with zeros at $\lambda=\bar\lambda_k$ for $1\le k\le \bar N$ each with
multiplicity $\bar m_k,$ is continuous in $\lambda\in\mathbb R,$ and 
has the large $\lambda$-asymptotics expressed in \eqref{2.18}. 
Furthermore, we already  know that \eqref{4.35} holds.
Hence, both sides of \eqref{4.41} must be identical to $1,$ from which we conclude that 
\begin{equation}
\label{4.42}
T(\zeta,t)=e^{-i\mu/2}\,\left(\ds\frac{\ds\prod_{k=1}^{\bar N}(\lambda-\bar\lambda_k)^{\bar m_k}}
{\ds\prod_{j=1}^N(\lambda-\lambda_j)^{m_j}}\right),
\quad \bar T(\zeta,t)=e^{i\mu/2}\,\left(\ds\frac{\ds\prod_{j=1}^N(\lambda-\lambda_j)^{m_j}}
{\ds\prod_{k=1}^{\bar N}(\lambda-\bar\lambda_k)^{\bar m_k}}\right),
\end{equation}
where we have the restriction $\mathcal N=\bar{\mathcal N}.$
From \eqref{2.19} and \eqref{2.20} we know that
the left-hand sides of the equalities in \eqref{4.42} become equal to $1$ when
$\lambda=0.$ Thus, evaluating the first equality in \eqref{4.42} at $\lambda=0,$ we get
\begin{equation}
\label{4.43}
1=e^{-i\mu/2}\,\left(\ds\frac{\ds\prod_{k=1}^{\bar N}(-\bar\lambda_k)^{\bar m_k}}
{\ds\prod_{j=1}^N(-\lambda_j)^{m_j}}\right).
\end{equation}
From \eqref{4.35} we know that
$A$ and $\bar A$ have the same number of eigenvalues including their multiplicities.
Hence, we can replace the minus signs in the products on the right-hand side of \eqref{4.43}
with the plus signs. Then, from the resulting equality we obtain \eqref{4.36}. Finally, using
\eqref{4.36} in \eqref{4.42} we get \eqref{4.38}.
\end{proof}

For reflectionless potentials, from Theorem~\ref{theorem4.3}(d), we know that the vectors $C$ and $\bar C$ 
appearing in the matrix triplets
$(A,B,C)$ and $(\bar A,\bar B,\bar C),$
respectively,
have no effect 
on the value of the complex constant $e^{i\mu/2}.$ 
In the next proposition we show that $C$ and $\bar C$ have only some limited effect
on the value of $\mu$ itself.

\begin{theorem}
\label{theorem4.4}
Suppose that the potentials $q(x,t)$ and $r(x,t)$ appearing in \eqref{2.1} at $t=0$ belong to the Schwartz class
and that the corresponding reflection coefficients $R(\zeta,0)$ and $\bar R(\zeta,0)$ are zero.
Let  $\mu$ be the complex constant defined in \eqref{2.14}.
Assume that the corresponding bound-state information is given by the two sets
in \eqref{2.27}, or equivalently by the pair of matrix triplets
$(A,B,C)$ and $(\bar A,\bar B,\bar C)$
described in \eqref{2.28}--\eqref{2.30}. 
For this reflectionless scattering data set, there can correspond a countably infinite number 
of $\mu$-values, and any two such $\mu$-values differ from each other by a constant integer
multiple of $4\pi.$

\end{theorem}

\begin{proof}
The stated nonuniqueness in the value of $\mu$ can be established by taking the complex logarithm
of both sides of \eqref{4.36} or \eqref{4.37}. From \eqref{4.37} we obtain
\begin{equation}
\label{4.44}
\ds\frac{i\mu}{2}=\log\left[\det[\bar A\,A^{-1}]\right]+2\pi i n,\qquad n\in\mathbb Z,
\end{equation}
where $\log$ denotes the principal branch of the complex logarithm function
and $n$ takes any integer values. From \eqref{4.44} we get
\begin{equation}
\label{4.45}
\mu=-2i\,\log\left[\det[\bar A\,A^{-1}]\right]+4\pi n,\qquad n\in\mathbb Z,
\end{equation}
which completes the proof.
\end{proof}

The restriction on the constant $\mu$ in \eqref{4.45} in the reflectionless case deserves further investigation.
In some examples in Section~\ref{section9}, we illustrate \eqref{4.45} by displaying
two distinct values of $\mu$ obtained by varying
$C$ and $\bar C,$ where those two $\mu$-values differ only by
$4\pi.$ In Example~\ref{example9.4} we demonstrate that
we may have three distinct $\mu$-values 
can be obtained by varying
$C$ and $\bar C.$
We pose it as an open problem
whether there is an upper limit on the number of distinct $\mu$-values in \eqref{4.45}
differing by $4\pi$ 
as we vary the bound-state norming constants.
It is also an open problem to explain physically
the meaning of having distinct values of
$\mu$ differing from each other 
 by an integer multiple of $4\pi.$
 In the terminology of integrable systems, the quantity $\mu$
 is a constant of motion, and hence a physical explanation 
of some distinct values of $\mu$ in \eqref{4.45}
may help us to understand the constants of motion better.

In the next theorem, in the reflectionless case, we show that the integrals of the potentials $q(x,t)$ and $r(x,t)$ over $x\in\mathbb R$ each must vanish.

\begin{theorem}
\label{theorem4.5}
Suppose that the potentials $q(x,t)$ and $r(x,t)$ appearing in \eqref{2.1} at $t=0$ belong to the Schwartz class
and that the corresponding reflection coefficients $R(\zeta,0)$ and $\bar R(\zeta,0)$ are zero.
Then, for each fixed $t\in\mathbb R$ we have
\begin{equation}
\label{4.46}
\int_{-\infty}^\infty dz\,q(z,t)=0,\quad
\int_{-\infty}^\infty dz\,r(z,t)=0.
\end{equation}
 \end{theorem}
 
\begin{proof} From \eqref{2.23} we know that the time-evolved
reflection coefficients $R(\zeta,t)$ and $\bar R(\zeta,t)$ are zero for all $t\in\mathbb R$
whenever their initial values at $t=0$ vanish.
Then, from the leading asymptotics as $\zeta\to 0$ in \eqref{2.21} and \eqref{2.22} we conclude that \eqref{4.46} must hold.
\end{proof}

We observe from \eqref{4.46} that, in the reflectionless case,
the corresponding solution pair $q(x,t)$ and $r(x,t)$ to the nonlinear system \eqref{1.6}
are somehow restricted. In fact, with the help of \eqref{1.20}, we see that
the potentials $\tilde q(x,t)$ and $\tilde r(x,t)$ satisfy the
same restriction as in \eqref{4.46}
whenever $b=a.$ When $b=a,$ from \eqref{1.6} we obtain
the one-parameter family of 
DNLS systems given by
\begin{equation}
\label{4.47}
\begin{cases}
i\,\tilde q_t+\tilde q_{xx}-i\kappa\left(\tilde q^2\,\tilde r\right)_x
=0,\\
\noalign{\medskip}
i\,\tilde r_t-\tilde r_{xx}-i\kappa\left(\tilde q\,\tilde r^2\right)_x
=0.
\end{cases}
\end{equation}
Let us note that \eqref{4.47} reduces to
the Kaup--Newell system \eqref{2.1} when $\kappa=1.$
We conclude that any solution pair $\tilde q(x,t)$ and $\tilde r(x,t)$ to
the nonlinear system \eqref{4.47} must satisfy the same restriction as in \eqref{4.46}
if $\tilde q(x,0)$ and $\tilde r(x,0)$ correspond to 
reflectionless potentials in the related linear system.
This indicates the power of the inverse scattering transform method
and the physical intuition it provides to solve some related mathematical problems.
Without the use of the inverse
scattering transform, i.e.
without relating \eqref{4.47} to the scattering theory
for a corresponding linear system, it may not be so easy to
prove that
there are infinitely many solution pairs to \eqref{4.47},
where those solutions have zero integrals over $x\in\mathbb R$ for each $t\in\mathbb R.$
Similarly, without the use of the inverse scattering transform, it may not be
so easy to prove that if
the initial values $\tilde q(x,0)$ and $\tilde r(x,0)$ have zero integrals
over $x\in\mathbb R$ then the solution pair $\tilde q(x,t)$ and $\tilde r(x,t)$
must each also have zero integrals
over $x\in\mathbb R$ for any $t\in\mathbb R.$ With the use of the inverse scattering transform, such mathematical
results are naturally discovered and their proofs are easy. Without the use
of the inverse scattering transform, those mathematical
results may not be so easy to discover or their proofs may not be 
so easy.

In the next theorem, in the reflectionless case, we express the Jost solutions to \eqref{2.1}
explicitly in terms of the matrix triplet pair
appearing in the scattering data set \eqref{2.38}.

\begin{theorem}
\label{theorem4.6}
Assume that the time-evolved reflectionless Marchenko kernels $\Omega(y,t)$ and $\bar\Omega(y,t)$ in \eqref{4.1}
are used as input to the Marchenko system \eqref{3.1}. Let the parameter $\lambda$ be related to the spectral parameter
$\zeta$ as in \eqref{2.12}. Then, the corresponding four Jost solutions to \eqref{2.1} with the 
respective asymptotics in \eqref{2.2}--\eqref{2.5} can
be expressed explicitly in terms of the matrix triplets $(A,B,C)$ and $(\bar A,\bar B,\bar C).$ 
In fact, we have the following:

\begin{enumerate}

\item[\text{\rm(a)}] 
The Jost solutions $\psi(\zeta,x,t)$ and $\bar\psi(\zeta,x,t)$ are
expressed in terms of the two matrix triplets as
\begin{equation}\label{4.48}
\psi_1(\zeta,x,t)= i\sqrt{\lambda}e^{i\lambda x}\left[\bar C\,e^{-i\bar A x}
\bar\Gamma^{-1} \left(\bar A-\lambda I\right)^{-1}e^{-i\bar A x-4i\bar A^2 t}\bar B\right]
e^{-2\int_x^\infty dz\,Q(z,t)},
\end{equation}
\begin{equation}\label{4.49}
\begin{split}
\psi_2&(\zeta,x,t)\\
&= \ds e^{i\lambda x} \left[1-iC\,e^{iAx}\Gamma^{-1}\,e^{iAx+4iA^2t} M\bar A\left(\bar A -\lambda I\right)^{-1}
e^{-2i\bar A x-4i\bar A^2 t}\bar B\right]e^{2\int_x^\infty dz\,Q(z,t)},
\end{split}\end{equation}
\begin{equation}\label{4.50}
\begin{split}
\bar\psi_1(&\zeta,x,t)\\
=&\ds  e^{-i\lambda x}\left[1+i\bar C\,e^{-i\bar A x}\bar\Gamma^{-1}\,e^{-i\bar A x-4i\bar A^2 t}\,
\bar M A\left(A-\lambda I\right)^{-1}e^{2iAx+4iA^2t}
B\right]e^{-2\int_x^\infty dz\,Q(z,t)},
\end{split}
\end{equation}
\begin{equation}\label{4.51}
\bar\psi_2(\zeta,x,t)=  -i\sqrt{\lambda}\,e^{-i\lambda x} \left[C\,e^{iAx}
\Gamma^{-1}\left(A-\lambda I\right)^{-1}e^{iAx+4iA^2t}B\right]e^{2\int_x^\infty dz\,Q(z,t)},
\end{equation}
where $M$ and $\bar M$ are the constant matrices defined in \eqref{4.9};
$\Gamma$ and $\bar \Gamma$ are the matrix-valued functions of $x$ and $t$
appearing in \eqref{4.7} and \eqref{4.8}, respectively; and
$Q(x,t)$ is the scalar-valued function of $x$ and $t$ defined
in \eqref{3.11} with its right-hand side
expressed by using \eqref{4.26} and \eqref{4.27}.

\item[\text{\rm(b)}]
Alternatively, the Jost solutions $\psi(\zeta,x,t)$ and $\bar\psi(\zeta,x,t)$ are expressed 
in terms of the matrix triplet pair as
\begin{equation}\label{4.52}
\psi_1(\zeta,x,t)= i\sqrt{\lambda}e^{i\lambda x}\left[\bar C\,e^{-i\bar A x}
\bar\Gamma^{-1} \left(\bar A-\lambda I\right)^{-1}e^{-i\bar A x-4i\bar A^2 t}\bar B\right]
e^{-2i\int_x^\infty dz\,P(z,t)},
\end{equation}
\begin{equation}\label{4.53}
\begin{split}
\psi_2&(\zeta,x,t)\\
&= \ds e^{i\lambda x} \left[1-iC\,e^{iAx}\Gamma^{-1}\,e^{iAx+4iA^2t} M\bar A\left(\bar A -\lambda I\right)^{-1}
e^{-2i\bar A x-4i\bar A^2 t}\bar B\right]e^{2i\int_x^\infty dz\,P(z,t)},
\end{split}
\end{equation}
\begin{equation}\label{4.54}
\begin{split}
\bar\psi_1(&\zeta,x,t)\\
=\ds &e^{-i\lambda x}\left[1+i\bar C\,e^{-i\bar A x}\bar\Gamma^{-1}\,e^{-i\bar A x-4i\bar A^2 t}\,
\bar M A\left(A-\lambda I\right)^{-1}e^{2iAx+4iA^2t}
B\right]e^{-2i\int_x^\infty dz\,P(z,t)},
\end{split}
\end{equation}
\begin{equation}\label{4.55}
\bar\psi_2(\zeta,x,t)=  -i\sqrt{\lambda}\,e^{-i\lambda x} \left[C\,e^{iAx}
\Gamma^{-1}\left(A-\lambda I\right)^{-1}e^{iAx+4iA^2t}B\right]e^{2i\int_x^\infty dz\,P(z,t)},
\end{equation}
where
$P(x,t)$ is the scalar-valued function of $x$ and $t$ defined
in \eqref{3.26} with its right-hand side
expressed by using \eqref{4.29} and \eqref{4.30}.

\item[\text{\rm(c)}]
For the Jost solutions
$\phi(\zeta,x,t)$ and $\bar\phi(\zeta,x,t),$ we have
\begin{equation}\label{4.56}
\phi(\zeta,x,t)=\bar T(\zeta,t)\,\bar\psi(\zeta,x,t),
\quad
\bar\phi(\zeta,x,t)=T(\zeta,t)\,\psi(\zeta,x,t),
\end{equation}
where $T(\zeta,t)$ and $\bar T(\zeta,t)$ are expressed
in terms of the matrices $A$ and $\bar A$ as in
\eqref{4.39}, and where the Jost
solutions $\psi(\zeta,x,t)$ and $\bar\psi(\zeta,x,t)$
are expressed in terms of the matrix triplet pair as in 
\eqref{4.48}--\eqref{4.51} or
alternatively as in \eqref{4.52}--\eqref{4.55}.

\end{enumerate}

\end{theorem}

\begin{proof}
We obtain \eqref{4.48}--\eqref{4.51} by using \eqref{4.3}--\eqref{4.6} in
\eqref{3.16}--\eqref{3.19} and by explicitly evaluating the integrals there related to the Fourier
transforms. We remark that the exponential terms in \eqref{3.16}--\eqref{3.19} are
expressed in terms of the quantity $Q(x,t)$ with the help of \eqref{3.11}.
Hence, the proof of (a) is complete.
We obtain the alternate expressions in (b) by using \eqref{4.3}--\eqref{4.6} in
\eqref{3.31}--\eqref{3.34} and again by explicitly evaluating the integrals there related to the Fourier
transforms. We mention that the exponential terms in \eqref{3.31}--\eqref{3.34} are
expressed in terms of the quantity $P(x,t)$ with the help of \eqref{3.26}.
Thus, the proof of (b) is also complete.
Having established \eqref{4.48}--\eqref{4.55},
we use \eqref{3.20} and \eqref{3.21} with $R(\zeta,t)\equiv 0$
and $\bar R(\zeta,t)\equiv 0.$ We then obtain
\begin{equation}\label{4.57}
\phi(\zeta,x,t)=\ds\frac{1}{T(\zeta,t)}\,\bar\psi(\zeta,x,t),
\quad
\bar\phi(\zeta,x,t)=\ds\frac{1}{\bar T(\zeta,t)}\,\psi(\zeta,x,t).
\end{equation}
As indicated in Theorem~\ref{theorem4.3}(e), the quantities $T(\zeta,t)$ and $\bar T(\zeta,t)$
are reciprocals of each other in the reflectionless case. Hence, \eqref{4.57} yields \eqref{4.56},
which completes the proof of our theorem.
\end{proof}

\section{Connection between the perturbed and unperturbed systems}
\label{section5}

In order to understand the relationship between the perturbed nonlinear system \eqref{1.6} and
the unperturbed nonlinear system \eqref{1.2}, we need to understand the relationship between the 
corresponding linear systems.
We recall that we use a tilde to denote the quantities related 
to the perturbed system and that 
the linear system corresponding to \eqref{1.2} is given in \eqref{2.1}.
The linear system corresponding to \eqref{1.6} is given in the first equality of
\eqref{1.12} and, for convenience, we write it in a format similar to that of
\eqref{2.1}. We have
\begin{equation}\label{5.1}
\ds\frac{d}{dx}\begin{bmatrix}
\tilde\alpha\\
\noalign{\medskip}
\tilde\beta
\end{bmatrix}=
\begin{bmatrix}
-i\zeta^2+\ds\frac{i b}{2}\,\tilde q(x,t) \,\tilde r(x,t)&\kappa\, \zeta\, \tilde q(x,t)
\\
\noalign{\medskip}
\ds\frac{1}{\kappa}\,\zeta\, \tilde r(x,t)&i\zeta^2+\ds\frac{i a}{2}\,\tilde q(x,t) \,\tilde r(x,t)
\end{bmatrix}
\begin{bmatrix}
\tilde\alpha\\
\noalign{\medskip}
\tilde\beta
\end{bmatrix}.
\end{equation}
In this section, we analyze the relationship between the relevant quantities for
\eqref{5.1} and the relevant quantities for
\eqref{2.1}.

Let us recall that the potentials
$\tilde q(x,t)$ and $\tilde r(x,t)$ appearing in \eqref{5.1} are related 
the potentials
$q(x,t)$ and $r(x,t)$ in \eqref{2.1} as in \eqref{1.20}.
The four Jost solutions to \eqref{5.1} are those solutions satisfying
the asymptotics \eqref{2.2}--\eqref{2.5}, respectively.
On the other hand, contrary to the coefficient matrix in \eqref{2.1},
the 
coefficient matrix in \eqref{5.1} does not have zero trace unless 
the parameters $a$ and $b$ satisfy $a+b=0.$
Consequently, the corresponding left and right transmission coefficients for \eqref{5.1}
are not equal to each other unless $a+b=0.$ Instead of obtaining the scattering coefficients from \eqref{2.6}--\eqref{2.9},
we obtain those coefficients from the spacial asymptotics given by
\begin{equation}\label{5.2}
\tilde\psi(\zeta,x,t)=\begin{bmatrix}
\ds\frac{\tilde L(\zeta,t)}{\tilde T_{\text{\rm l}}(\zeta,t)}\,e^{-i\zeta^2 x}\left[1+o(1)\right]\\
\noalign{\medskip}
\ds\frac{1}{\tilde T_{\text{\rm l}}(\zeta,t)}\,e^{i\zeta^2 x}\left[1+o(1)\right]
\end{bmatrix}, \qquad   x\to-\infty,
\end{equation}
\begin{equation}\label{5.3}
\tilde{\bar\psi}(\zeta,x,t)=\begin{bmatrix}
\ds\frac{1}{\tilde{\bar T}_{\text{\rm l}}(\zeta,t)}\,e^{-i\zeta^2 x}\left[1+o(1)\right]\\
\noalign{\medskip}
\ds\frac{\tilde{\bar L}(\zeta,t)}{\tilde{\bar T}_{\text{\rm l}}(\zeta,t)}\,e^{i\zeta^2 x}\left[1+o(1)\right]
\end{bmatrix}, \qquad  x\to-\infty,
\end{equation}
\begin{equation}\label{5.4}
\tilde\phi(\zeta,x,t)=\begin{bmatrix}
\ds\frac{1}{\tilde T_{\text{\rm r}}(\zeta,t)}\,e^{-i\zeta^2x}\left[1+o(1)\right]\\
\noalign{\medskip}
\ds\frac{\tilde R(\zeta,t)}{\tilde T_{\text{\rm r}}(\zeta,t)}\,e^{i\zeta^2 x}\left[1+o(1)\right]
\end{bmatrix}, \qquad   x\to+\infty,
\end{equation}
\begin{equation}\label{5.5}
\tilde{\bar\phi}(\zeta,x,t)=\begin{bmatrix}
\ds\frac{\tilde{\bar R}(\zeta,t)}   {\tilde{\bar T}_{\text{\rm r}}(\zeta,t)}\,e^{-i\zeta^2 x}\left[1+o(1)\right]\\
\noalign{\medskip}
\ds\frac{1}{\tilde{\bar T}_{\text{\rm r}}(\zeta,t)}\,e^{i\zeta^2 x}\left[1+o(1)\right]
\end{bmatrix}, \qquad   x\to+\infty.
\end{equation}
Thus, for the perturbed linear system \eqref{5.1} we have eight scattering coefficients; namely,
the left transmission coefficients $\tilde T_{\text{\rm l}}(\zeta,t)$ and $\tilde{\bar T}_{\text{\rm l}}(\zeta,t),$
the right transmission coefficients $\tilde T_{\text{\rm r}}(\zeta,t)$ and $\tilde{\bar T}_{\text{\rm r}}(\zeta,t),$
the left reflection coefficients $\tilde L(\zeta,t)$ and $\tilde{\bar L}(\zeta,t),$
and the right reflection coefficients $\tilde R(\zeta,t)$ and $\tilde{\bar R}(\zeta,t).$

We recall that the solutions to the perturbed linear system \eqref{5.1} 
and the solutions to the unperturbed linear system \eqref{2.1} are
related to each other as in \eqref{1.17}. Since the
linear systems \eqref{2.1} and \eqref{5.1} are both homogeneous, any constant multiples of their solutions
are also solutions.
For the Jost solutions to \eqref{2.1} and the Jost solutions to \eqref{5.1},
such constants can be chosen with the help of \eqref{2.13}.
Thus, we obtain the relationships
\begin{equation}\label{5.6}
\tilde\psi(\zeta,x,t)=e^{-ia \mu/2}\begin{bmatrix}
E(x,t)^{b}& 0\\
\noalign{\medskip}
0&E(x,t)^{a}
\end{bmatrix}\psi(\zeta,x,t),
\end{equation}
\begin{equation}\label{5.7}
\tilde{\bar\psi}(\zeta,x,t)=e^{-ib \mu/2}\begin{bmatrix}
E(x,t)^{b}& 0\\
\noalign{\medskip}
0&E(x,t)^{a}
\end{bmatrix}\bar\psi(\zeta,x,t),
\end{equation}
\begin{equation}\label{5.8}
\tilde\phi(\zeta,x,t)=\begin{bmatrix}
E(x,t)^{b}& 0\\
\noalign{\medskip}
0&E(x,t)^{a}
\end{bmatrix}\phi(\zeta,x,t),
\end{equation}
\begin{equation}\label{5.9}
\tilde{\bar\phi}(\zeta,x,t)=\begin{bmatrix}
E(x,t)^{b}& 0\\
\noalign{\medskip}
0&E(x,t)^{a}
\end{bmatrix}\bar\phi(\zeta,x,t).
\end{equation}

Using the spacial asymptotics in \eqref{5.6}--\eqref{5.9},
with the help of \eqref{2.6}--\eqref{2.9} and 
\eqref{5.2}--\eqref{5.5}, we relate
the eight scattering coefficients for \eqref{5.1} to the six
scattering coefficients for \eqref{2.1} as
\begin{equation}\label{5.10}
\tilde T_{\text{\rm l}}(\zeta,t)=e^{i a \mu/2} \,T(\zeta,t),\quad
\tilde{\bar T}_{\text{\rm l}}(\zeta,t)=
e^{i b \mu/2} \,\bar T(\zeta,t),
\end{equation}
\begin{equation}\label{5.11}
\tilde T_{\text{\rm r}}(\zeta,t)=e^{-ib \mu/2} \,T(\zeta,t),\quad
\tilde{\bar T}_{\text{\rm r}}(\zeta,t)=
e^{-ia \mu/2} \,\bar T(\zeta,t),
\end{equation}
\begin{equation}\label{5.12}
\tilde R(\zeta,t)=e^{i (a-b) \mu/2} \,R(\zeta,t),\quad
\tilde{\bar R}(\zeta,t)=
e^{i (b-a) \mu/2} \,\bar R(\zeta,t),
\end{equation}
\begin{equation}\label{5.13}
\tilde L(\zeta,t)=L(\zeta,t),\quad
\tilde{\bar L}(\zeta,t)=
\bar L(\zeta,t).
\end{equation}

We recall that the bound-state information for the unperturbed system \eqref{2.1}
is described by the two sets specified in \eqref{2.27} and that it is
the most convenient to use the bound-state information
not in the form of \eqref{2.27} but 
in the form of the matrix triplet pair $(A,B,C)$
and $(\bar A,\bar B,\bar C).$
Let us use the matrix triplet pair $(\tilde A,\tilde B,\tilde C)$
and $(\tilde{\bar A},\tilde{\bar B},\tilde{\bar C})$
to represent the bound-state information related to the
perturbed system \eqref{5.1}. We then have
\begin{equation}\label{5.14}
\tilde A=A,\quad
\tilde B=B,\quad \tilde C=e^{i (a-b) \mu/2}\,C,
\end{equation}
\begin{equation}\label{5.15}
\tilde{\bar A}=\bar A,\quad
\tilde{\bar  B}=\bar B,\quad \tilde{\bar C}=e^{i (b-a) \mu/2}\,\bar C.
\end{equation}
The first two equalities in \eqref{5.14} and \eqref{5.15} are obtained 
as follows. As seen from the first equalities in \eqref{5.10} and \eqref{5.11},
the poles of $\tilde T_{\text{\rm l}}(\zeta,0),$ $\tilde T_{\text{\rm r}}(\zeta,0),$ and $T(\zeta,0)$
in the upper-half complex $\lambda$-plane coincide and
also the multiplicities of those poles coincide. From Theorem~\ref{theorem2.1}
we know that such poles and multiplicities are the only ingredients
to construct the matrices $A$ and $B$ with the help of \eqref{2.31}.
Hence, the first two equalities in \eqref{5.14}
are justified.
Similarly, from the second equalities in \eqref{5.10} and \eqref{5.11}
it follows that
the poles of $\tilde{\bar T}_{\text{\rm l}}(\zeta,0),$ $\tilde{\bar T}_{\text{\rm r}}(\zeta,0),$ and $\bar T(\zeta,0)$
in the lower-half complex $\lambda$-plane coincide and
also the multiplicities of those poles coincide. Again from Theorem~\ref{theorem2.1}
we know that such poles and multiplicities are the only ingredients 
to construct the matrices $\bar A$ and $\bar B$ with the help of \eqref{2.33}.
Hence, the first two equalities in
\eqref{5.15} are also justified.
The justification of the third equalities in
\eqref{5.14} and \eqref{5.15} follow from the construction of
the norming constants $c_{jk}$ appearing in
\eqref{2.32} and of the norming constants
$\tilde c_{jk}$ appearing in
\eqref{2.34}. The details of those constructions can be found in
Section~3 of \cite{AE2022b} and Examples~6.1 and 6.2 of
that reference. From those constructions it follows that
the norming constant $c_{jk}$ is directly proportional
to $T(\zeta,0),$ directly proportional to the Jost solution
$\phi(\zeta,x,0),$ and inversely proportional to
the Jost solution $\psi(\zeta,x,0);$
the norming constant
$\bar c_{jk}$ is directly proportional
to $\bar T(\zeta,0),$ directly proportional to the Jost solution
$\bar\phi(\zeta,x,0),$ and inversely proportional to
the Jost solution $\bar\psi(\zeta,x,0);$
the norming constant
$\tilde c_{jk}$ is directly proportional
to $\tilde T_{\text{\rm r}}(\zeta,0),$ directly proportional to the Jost solution
$\tilde\phi(\zeta,x,0),$ and inversely proportional to
the Jost solution $\tilde\psi(\zeta,x,0);$ and
the norming constant
$\tilde{\bar c}_{jk}$ is directly proportional
to $\tilde{\bar T}_{\text{\rm r}}(\zeta,0),$ directly proportional to the Jost solution
$\tilde{\bar \phi}(\zeta,x,0),$ and inversely proportional to
the Jost solution $\tilde{\bar\psi}(\zeta,x,0).$
Then, using \eqref{5.6}--\eqref{5.9} and \eqref{5.11}
we justify the third
equalities in \eqref{5.14} and \eqref{5.15}.

\section{The solution to the perturbed nonlinear system}
\label{section6}

In this section we describe the use of our Marchenko method to obtain the solution to
the initial-value problem for \eqref{1.1} or equivalently for \eqref{1.6}.
Thus, we are given
$\tilde q(x,0)$ and $\tilde r(x,0)$
and we would like to determine 
$\tilde q(x,t)$ and $\tilde r(x,t)$
satisfying \eqref{1.6} at any time $t.$
The following are the steps to obtain 
$\tilde q(x,t)$ and $\tilde r(x,t)$
with the help of the Marchenko method described in Section~\ref{section3}.

\begin{enumerate}

\item[\text{\rm(a)}] 
From the coefficients of the third and fourth terms in the first line
of \eqref{1.6} we determine the values of the parameters
$\kappa$ and $a-b.$ Let use $\Delta_1$ and $\Delta_2$ to denote the respective coefficients, i.e. let us use
\begin{equation}\label{6.1}
\Delta_1:=i\kappa(a-b-2),\quad \Delta_2:=i\kappa(a-b-1).
\end{equation}
From \eqref{6.1} we uniquely determine $\kappa$ and $a-b$ as
\begin{equation*}
\kappa=i(\Delta_1-\Delta_2),\quad a-b=\ds\frac{\Delta_1-2 \Delta_2}{\Delta_1-\Delta_2}.
\end{equation*}

\item[\text{\rm(b)}] 
Next, we use $\tilde q(x,0)$ and $\tilde r(x,0)$ as input to \eqref{5.1} at $t=0,$
and we obtain the corresponding Jost solutions $\tilde\psi(\zeta,x,0),$
$\tilde{\bar\psi}(\zeta,x,0),$
$\tilde\phi(\zeta,x,0),$
$\tilde{\bar\phi}(\zeta,x,0)$
to \eqref{5.1} at $t=0.$

\item[\text{\rm(c)}] 
From the four Jost solutions at $t=0$ constructed in step (b),
by using the spacial asymptotics \eqref{5.2}--\eqref{5.5}, we get
the corresponding eight scattering coefficients
$\tilde T_{\text{\rm l}}(\zeta,0),$
$\tilde{\bar T}_{\text{\rm l}}(\zeta,0),$
$\tilde T_{\text{\rm r}}(\zeta,0),$
$\tilde{\bar T}_{\text{\rm r}}(\zeta,0),$
$\tilde R(\zeta,0),$
$\tilde{\bar R}(\zeta,0),$
$\tilde L(\zeta,0),$ $\tilde{\bar L}(\zeta,0).$

\item[\text{\rm(d)}] 
We remark that, using \eqref{2.19}, \eqref{2.20}, and
\eqref{5.11}, we obtain
\begin{equation}\label{6.2}
\tilde T_{\text{\rm r}}(\zeta,0)=e^{-ib \mu/2},
\quad
\tilde{\bar T}_{\text{\rm r}}(\zeta,0)=
e^{-ia \mu/2}.
\end{equation}
Since we know the left-hand sides of the two equalities in \eqref{6.2} from step (c), 
we have $e^{ia \mu/2}$ and $e^{ib \mu/2}$ both at hand.
Alternatively, we can use the left transmission coefficients
to obtain $e^{ia \mu/2}$ and $e^{ib \mu/2},$ and 
this can be achieved with the help of \eqref{2.19}, \eqref{2.20}, and \eqref{5.10}.

\item[\text{\rm(e)}] 
Having the scattering coefficients at $t=0$ from step (c) and 
the values of 
$e^{ia \mu/2}$ and $e^{ib \mu/2}$ from step (d), we use
\eqref{5.10}--\eqref{5.13} at $t=0$
and construct
the six scattering coefficients at $t=0$ for
the unperturbed system \eqref{2.1}. We have
\begin{equation*}
T(\zeta,0)=e^{ib \mu/2}\, \tilde T_{\text{\rm r}}(\zeta,0),
\quad
\bar T(\zeta,0)=e^{ia \mu/2}\,
\tilde{\bar T}_{\text{\rm r}}(\zeta,0),
\end{equation*}
\begin{equation*}
R(\zeta,0)=e^{i (b-a) \mu/2} \,\tilde R(\zeta,0),\quad
\bar R(\zeta,0)=e^{i (a-b) \mu/2} \,
\tilde{\bar R}(\zeta,0),
\end{equation*}
\begin{equation*}
L(\zeta,0)=\tilde L(\zeta,0),
\quad
\bar L(\zeta,0)=
\tilde{\bar L}(\zeta,0).
\end{equation*}

\item[\text{\rm(f)}] 
We already know the value of $a-b$ from step (a),
and we would like to obtain the values of $a$ and $b$ separately.
Since the parameters $a,$ $b,$ and $\mu$ appearing 
on the right-hand sides in \eqref{6.2} may be complex, the use of the
complex logarithm function cannot uniquely determine
$a$ and $b$ from \eqref{6.2}. In order to have unique values for $a$ and $b,$
we proceed as follows. Using the principal
branch of the complex logarithm of the already known
quantity $e^{ia \mu/2},$ we uniquely obtain
the value of $a$ as
\begin{equation}\label{6.3}
a=\log[e^{-ia \mu/2}].
\end{equation}
Since the value of $a-b$ is already known, we obtain 
the value of $b$ uniquely with the help of 
\eqref{6.3} as
\begin{equation*}
b=-(a-b)+\log[e^{-ia \mu/2}].
\end{equation*}

\item[\text{\rm(g)}] Note that \eqref{1.20} implies that
\begin{equation}\label{6.4}
\tilde q(x,t)\,\tilde r(x,t)=q(x,t)\,r(x,t),
\end{equation}
and hence the quantity $E(x,t)$ defined in \eqref{1.19}
can also be expressed in terms of
$\tilde q(x,t)$ and $\tilde r(x,t)$
as
\begin{equation}
\label{6.5}
E(x,t)=\exp\left(\ds\frac{i}{2}\ds\int_{-\infty}^x dz\,\tilde q(z,t)\,\tilde r(z,t)\right).
\end{equation}
Since we already have
$\tilde q(x,0)$ and $\tilde r(x,0)$
at hand, using \eqref{6.5} we construct
$E(x,0)$ as
\begin{equation*}
E(x,0)=\exp\left(\ds\frac{i}{2}\ds\int_{-\infty}^x dz\,\tilde q(z,0)\,\tilde r(z,0)\right).
\end{equation*}

\item[\text{\rm(h)}] We have the Jost solutions at $t=0$ to the
perturbed problem \eqref{5.1} from step (b), the quantities
$e^{ia \mu/2}$ and $e^{ib \mu/2}$ from step (d), the values of $a$ and $b$
from step (f), and the quantity $E(x,0)$
from step (g). Hence, 
using \eqref{5.6}--\eqref{5.9} we construct at $t=0$ the
Jost solutions to the unperturbed problem \eqref{2.1}
as
\begin{equation*}
\psi(\zeta,x,0)
=e^{ia \mu/2}\begin{bmatrix}
\ds\frac{1}{E(x,0)^{b}}& 0\\
\noalign{\medskip}
0&\ds\frac{1}{E(x,0)^{a}}
\end{bmatrix}\tilde\psi(\zeta,x,0),
\end{equation*}
\begin{equation*}
\bar\psi(\zeta,x,0)=e^{ib \mu/2}\begin{bmatrix}
\ds\frac{1}{E(x,0)^{b}}& 0\\
\noalign{\medskip}
0&\ds\frac{1}{E(x,0)^{a}}
\end{bmatrix}\tilde{\bar\psi}(\zeta,x,0),
\end{equation*}
\begin{equation*}
\phi(\zeta,x,0)
=\begin{bmatrix}
\ds\frac{1}{E(x,0)^{b}}& 0\\
\noalign{\medskip}
0&\ds\frac{1}{E(x,0)^{a}}
\end{bmatrix}\tilde\phi(\zeta,x,0),
\end{equation*}
\begin{equation*}
\bar\phi(\zeta,x,0)
=\begin{bmatrix}
\ds\frac{1}{E(x,0)^{b}}& 0\\
\noalign{\medskip}
0&\ds\frac{1}{E(x,0)^{a}}
\end{bmatrix}\tilde{\bar\phi}(\zeta,x,0).
\end{equation*}

\item[\text{\rm(i)}] 
Having the unperturbed transmission coefficients $T(\zeta,0)$ and
$\bar T(\zeta,0)$ as well as the four
unperturbed Jost solutions $\psi(\zeta,x,0),$
$\bar\psi(\zeta,x,0),$
$\phi(\zeta,x,0),$
$\bar\phi(\zeta,x,0)$
to \eqref{2.1} at $t=0,$
we construct the matrix triplets
$(A,B,C)$ and $(\bar A,\bar B,\bar C)$ 
as described in (h),  (i),  (j), and (k) of Theorem~\ref{theorem2.1}.
The details of the construction can be found in Section~3 of \cite{AE2022b}.
As far as our Marchenko method is concerned, this amounts to
including the effects of the bound states in the Marchenko kernels
by using the ``recipe''
\begin{equation}\label{6.6}
\hat R(y,0)\mapsto \hat R(y,0)+C\,e^{iAy}\,B,\quad\hat{\bar R}(y,0)\mapsto\hat{\bar R}(y,0)+\bar C\,e^{-i\bar A y}\,\bar B.
\end{equation}
In fact, the simple and elegant way of including  the bound-state information stated in \eqref{6.6} holds in general
also for other linear systems for which a Marchenko method is available. 
This is one of the strengths of the Marchenko method in the sense
that any number of bound states with any multiplicities can be handled in a simple and elegant manner
by using \eqref{6.6}.
We now have the scattering data set $\mathbf S(\zeta,0)$ at $t=0$ for the unperturbed 
linear system \eqref{2.1}.

\item[\text{\rm(j)}] 
Having the unperturbed scattering data set $\mathbf S(\zeta,0)$ at $t=0$ at hand,
we use \eqref{2.16}, \eqref{2.23}, \eqref{2.24}, \eqref{2.34}, and \eqref{2.35}
in order to obtain the time-evolved scattering data set $\mathbf S(\zeta,t)$ at any time $t.$
Knowing $\mathbf S(\zeta,t),$ we also have the time-evolved Marchenko kernels
$\Omega(\zeta,t)$ and $\bar\Omega(\zeta,t)$ defined in \eqref{3.2}.

\item[\text{\rm(k)}] 
Having at hand the time-evolved Marchenko kernels
$\Omega(\zeta,t)$ and $\bar\Omega(\zeta,t),$ we use them as input in the
Marchenko system \eqref{3.1} or equivalently in the
uncoupled Marchenko system \eqref{3.8} and the auxiliary system \eqref{3.9}.
Then, we
obtain the solutions 
$K_1(x,y,t),$
$K_2(x,y,t),$
$\bar K_1(x,y,t),$
$\bar K_2(x,y,t).$
Next, as described in Theorem~\ref{theorem3.2} we construct all the relevant quantities associated with
\eqref{2.1}; namely, we get the key quantity $E(x,t)$ defined in \eqref{1.19}, the constant
$\mu$ defined in \eqref{2.14}, the potentials $q(x,t)$ and $r(x,t)$ appearing
in \eqref{2.1}, and the Jost solutions $\psi(\zeta,x,t)$ and $\bar\psi(\zeta,x,t)$ to \eqref{2.1}
satisfying \eqref{2.2} and \eqref{2.3}, respectively.
We also obtain the Jost solutions $\phi(\zeta,x,t)$ and $\bar\phi(\zeta,x,t)$ to \eqref{2.1}
satisfying \eqref{2.4} and \eqref{2.5}, respectively,
by using \eqref{3.20} and \eqref{3.21}.

\item[\text{\rm(l)}] 
Finally, we transform the relevant quantities obtained in step (k) for the unperturbed
system \eqref{2.1}, and we obtain the corresponding relevant quantities for 
the perturbed system \eqref{5.1}.
This is accomplished as follows.
We recover the potentials $\tilde q(x,t)$ and $\tilde r(x,t)$ 
by using \eqref{1.20},
where we remark that we already know the value
of the parameter $\kappa$ from step (a) and we know the values of
the parameters $a$ and $b$ from step (f).
We also recover the four Jost solutions to \eqref{5.1}, and this is done with the help of
\eqref{5.6}--\eqref{5.9} and by using the four Jost solutions $\psi(\zeta,x,t),$ $\bar\psi(\zeta,x,t),$ 
$\phi(\zeta,x,t),$ $\bar\phi(\zeta,x,t)$ to \eqref{2.1} already constructed in step (k).
Based on the inverse scattering transform method, it is known that
$\tilde q(x,t)$ and $\tilde r(x,t)$ are the solutions to
the initial-value problem for \eqref{1.1}.

\end{enumerate}

\section{Explicit solution formulas for the perturbed system}
\label{section7}

In this section we present some explicit solution formulas for the general DNLS system \eqref{1.6}
as well as some explicit solution formulas for the corresponding linear system \eqref{5.1}.
This is done by providing the solution formulas in closed form for $\tilde q(x,t)$ and $\tilde r(x,t)$
satisfying \eqref{1.6}, where the formulas are
explicitly expressed in terms of the two matrix triplets
$(A,B,C)$ and $(\bar A,\bar B,\bar C)$ appearing in \eqref{4.1}
and the two
complex-valued parameters $(a-b)$ and $\kappa$ appearing in \eqref{1.6}.
When the solution pair $\tilde q(x,t)$ and $\tilde r(x,t)$ is used as potentials in
\eqref{5.1}, we also present the explicit formulas for the corresponding 
Jost solutions $\tilde\psi(\zeta,x,t),$ $\tilde\phi(\zeta,x,t),$ $\tilde{\bar\psi}(\zeta,x,t),$ and
$\tilde{\bar\phi}(\zeta,x,t),$ which are all
expressed in terms of the two matrix triplets and
the three parameters $a,$ $b,$ and $\kappa.$

The formulas presented for $\tilde q(x,t)$ and $\tilde r(x,t)$
represent the  time-evolved, reflectionless potentials
in the linear system \eqref{5.1}. The formulas we have for the potentials and for
the corresponding Jost solutions
contain matrix exponentials, and those formulas are valid when each matrix triplet has an arbitrary size.
On the other hand, as indicated in Theorem~\ref{theorem4.3}(c), the corresponding
potentials cannot both belong to the Schwartz class unless the triplet
sizes for $(A,B,C)$ and $(\bar A,\bar B,\bar C)$ are equal to each other. The use of matrix exponentials
allows the formulas presented to have a compact form, independent of the number of bound states and of 
multiplicities of those bound states.
The matrix exponentials can certainly be 
explicitly expressed in terms of elementary functions, but the resulting expressions, 
as the matrix
size gets large,
become extremely lengthy and not practical to display.
In such cases, a symbolic software such as Mathematica, may be used to display
the solution formulas by expressing the matrix exponentials in terms of elementary functions.

In the next theorem, we present some explicit formulas for
the potentials $\tilde q(x,t)$ and $\tilde r(x,t)$ and the Jost solutions for \eqref{5.1}
expressed in terms of a matrix triplet pair corresponding to reflectionless scattering data.
We remark that the formulas presented for $\tilde q(x,t)$ and $\tilde r(x,t)$ make up an explicit solution
to the general DNLS system \eqref{1.6}.

\begin{theorem}
\label{theorem7.1}
Suppose that the potentials $\tilde q(x,t)$ and $\tilde r(x,t)$
appearing in \eqref{5.1} at $t=0$ belong to the Schwartz class and that the 
corresponding reflection coefficients $\tilde R(\zeta,0)$ and $\tilde{\bar R}(\zeta,0)$ are zero.
Let $a,$ $b,$ and $\kappa$ be the complex parameters appearing in \eqref{5.1}.
We have the following:

\begin{enumerate}

\item[\text{\rm(a)}] The formulas
\begin{equation}\label{7.1}
\tilde q(x,t)=\left(\ds\frac{2}{\kappa}\,\bar C\,e^{-i\bar A x}\,\bar\Gamma(x,t)^{-1}\,e^{-i\bar A x-4i\bar A^2 t}\,\bar B\right)\exp\left(2(b-a)\ds\int_{-\infty}^x dz\,Q(z,t)
\right),
\end{equation}
\begin{equation}\label{7.2}
\tilde r(x,t)=\left(2\kappa\,C\,e^{i A x}\,\Gamma(x,t)^{-1}\,e^{i A x+4i A^2 t}\,B\right)\exp\left(
2(a-b)\ds\int_{-\infty}^x dz\,Q(z,t)
\right),
\end{equation}
yield an explicit solution pair for \eqref{1.6} with the initial values
$\tilde q(x,0)$ and $\tilde r(x,0).$
Here, $(A,B,C)$ and $(\bar A,\bar B,\bar C)$ are the matrix triplets
appearing in \eqref{2.28}--\eqref{2.30}
with equal matrix triplet sizes and all eigenvalues of
$A$ located in $\mathbb C^+$ and all eigenvalues of
$\bar A$ located in $\mathbb C^-;$
$\Gamma$ and $\bar \Gamma$ are the matrix-valued functions of $x$ and $t$
appearing in \eqref{4.7} and \eqref{4.8}, respectively; and $Q(x,t)$ is the scalar-valued function of $x$ and $t$ defined
in \eqref{3.11} with $\bar K_1(x,x,t)$ and 
$K_2(x,x,t)$  explicitly expressed in terms of the matrix triplets as in \eqref{4.26} and \eqref{4.27}.
Alternatively, the formulas in \eqref{7.1} and \eqref{7.2} can be expressed as
\begin{equation}\label{7.3}
\tilde q(x,t)=\left(\ds\frac{2}{\kappa}\,\bar C\,e^{-i\bar A x}\,\bar\Gamma(x,t)^{-1}\,e^{-i\bar A x-4i\bar A^2 t}\,\bar B\right)\exp\left(2i(b-a)\ds\int_{-\infty}^x dz\,P(z,t)
\right),
\end{equation}
\begin{equation}\label{7.4}
\tilde r(x,t)=\left(2\kappa\,C\,e^{i A x}\,\Gamma(x,t)^{-1}\,e^{i A x+4i A^2 t}\,B\right)\exp\left(
2i(a-b)\ds\int_{-\infty}^x dz\,P(z,t)
\right),
\end{equation}
where $P(x,t)$ is the scalar-valued function of $x$ and $t$
defined in \eqref{3.26} with $K_1(x,x,t)$ and 
$\bar K_2(x,x,t)$ explicitly expressed in terms of the matrix triplets as in \eqref{4.29} and \eqref{4.30},
respectively.

\item[\text{\rm(b)}]  If the expressions
in \eqref{7.1} and \eqref{7.2}, or equivalently \eqref{7.3} and \eqref{7.4}, are used as the potentials in \eqref{5.1}, then the
corresponding Jost solutions $\tilde\psi(\zeta,x,t)$ and $\tilde{\bar\psi}(\zeta,x,t)$ with components 
similarly defined as in \eqref{2.10} are explicitly expressed by the formulas
\begin{equation}\label{7.5}
\tilde\psi_1(\zeta,x,t)= i\sqrt{\lambda}e^{i\lambda x}\left[\bar C\,e^{-i\bar A x}
\bar\Gamma^{-1} \left(\bar A-\lambda I\right)^{-1}e^{-i\bar A x-4i\bar A^2 t}\bar B\right]\Upsilon(-a-1,b-a),
\end{equation}
\begin{equation}\label{7.6}
\tilde\psi_2(\zeta,x,t)= \ds e^{i\lambda x} \left[1-iC e^{iAx}\Gamma^{-1} e^{iAx+4iA^2t} M\bar A\left(\bar A -\lambda I\right)^{-1}
e^{-2i\bar A x-4i\bar A^2 t}\bar B\right]\Upsilon(1-a,0),
\end{equation}
\begin{equation}\label{7.7}
\begin{split}
\tilde{\bar\psi}_1(&\zeta,x,t)\\
=& \ds e^{-i\lambda x}\left[1+i\bar C e^{-i\bar A x}\bar\Gamma^{-1} e^{-i\bar A x-4i\bar A^2 t}
\bar M A\left(A-\lambda I\right)^{-1}e^{2iAx+4iA^2t}
B\right]\Upsilon(-b-1,b-1),
\end{split}
\end{equation}
\begin{equation}\label{7.8}
\tilde{\bar\psi}_2(\zeta,x,t)= -i\sqrt{\lambda}\,e^{-i\lambda x} \left[C\,e^{iAx}
\Gamma^{-1}\left(A-\lambda I\right)^{-1}e^{iAx+4iA^2t}B\right]\Upsilon(-b-1,b-1),
\end{equation}
where $M$ and $\bar M$ are the constant matrices defined in \eqref{4.9},
the parameter $\lambda$ is related to the
spectral parameter $\zeta$ as in \eqref{2.12},
and we have defined the double-indexed scalar quantity $\Upsilon(a,b)$ as
\begin{equation}\label{7.9}
\Upsilon(a,b):=\exp\left(
2a\ds\int_x^\infty dz\,Q(z,t)+2b \ds\int_{-\infty}^x dz\,Q(z,t)
\right).
\end{equation}
Note that $\Upsilon(a,b)$ can alternatively be evaluated by
replacing $Q(z,t)$ in \eqref{7.9} by $i P(z,t),$ as indicated in \eqref{3.27}.

\item[\text{\rm(c)}]  If \eqref{7.1} and \eqref{7.2}, or equivalently \eqref{7.3} and \eqref{7.4}, are used as the potentials
in \eqref{5.1}, then the
corresponding Jost solutions $\tilde\phi(\zeta,x,t)$ and $\tilde{\bar\phi}(\zeta,x,t)$ with components 
similarly defined as in \eqref{2.10} are explicitly expressed by the formulas
\begin{equation}\label{7.10}
\begin{split}
\tilde\phi_1(&\zeta,x,t)
\\
=&e^{-i\lambda x} \bar T(\zeta,t) \left[1+i\bar C e^{-i\bar A x}\bar\Gamma^{-1}\,e^{-i\bar A x-4i\bar A^2 t}
\bar M A\left(A-\lambda I\right)^{-1}e^{2iAx+4iA^2t}
B\right]\Upsilon(-1,b),
\end{split}
\end{equation}
\begin{equation}\label{7.11}
\tilde\phi_2(\zeta,x,t)=-i\sqrt{\lambda} e^{-i\lambda x}\bar T(\zeta,t) \left[C e^{iAx}
\Gamma^{-1}\left(A-\lambda I\right)^{-1}e^{iAx+4iA^2t}B\right]\Upsilon(1,a),
\end{equation}
\begin{equation}\label{7.12}
\tilde{\bar\phi}_1(\zeta,x,t)=i\sqrt{\lambda}\,e^{i\lambda x} \,T(\zeta,t)\left[\bar C\,e^{-i\bar A x}
\bar\Gamma^{-1} \left(\bar A-\lambda I\right)^{-1}e^{-i\bar A x-4i\bar A^2 t}\bar B\right]
\Upsilon(-1,b),
\end{equation}
\begin{equation}\label{7.13}
\begin{split}
\tilde{\bar\phi}_2(&\zeta,x,t)\\
&=e^{i\lambda x}\,T(\zeta,t) \left[1-iC\,e^{iAx}\Gamma^{-1}\,e^{iAx+4iA^2t} M\bar A\left(\bar A -\lambda I\right)^{-1}
e^{-2i\bar A x-4i\bar A^2 t}\bar B\right]\Upsilon(1,a),
\end{split}
\end{equation}
where $T(\zeta,t)$ and $\bar T(\zeta,t)$ are expressed
in terms of the matrices $A$ and $\bar A$ as in \eqref{4.38} or equivalently as
in \eqref{4.39}.

\end{enumerate}

\end{theorem}

\begin{proof}
We get \eqref{7.1} from \eqref{1.20} with the help of \eqref{4.25} and \eqref{4.31}. Similarly, we get
\eqref{7.2} from \eqref{1.20} with the help of \eqref{4.25} and \eqref{4.32}.
The alternate expressions in \eqref{7.3} and \eqref{7.4} are obtained from \eqref{7.1} and \eqref{7.2}, respectively,
by using \eqref{3.27}.
Hence, the proof of (a) is complete.
We obtain \eqref{7.5}--\eqref{7.8} from \eqref{5.6} and \eqref{5.7} with the help of \eqref{4.25} and \eqref{4.48}--\eqref{4.51}.
Thus, the proof of (b) is complete.
We obtain \eqref{7.10}--\eqref{7.13} from \eqref{5.8} and \eqref{5.9} with the help of \eqref{4.25} and \eqref{4.56}
with $T(\zeta,t)$ and $\bar T(\zeta,t)$ expressed in terms of the matrices
$A$ and $\bar A$ as in \eqref{4.39}.
Hence, the proof of (c) is also complete.
\end{proof}

\section{The reductions}
\label{section8}

When the dependent variables $q(x,t)$ and $r(x,t)$
in the integrable nonlinear system \eqref{1.2} 
are related to each other in some way, the system
\eqref{1.2} consisting of two equations may be reduced to a single equation 
in one dependent variable.
In this section we consider the two common types of reductions, namely,
\begin{equation}\label{8.1}
r(x,t)=q(x,t)^\ast,\quad r(x,t)=-q(x,t)^\ast,
\end{equation}
where we recall that we use an asterisk to denote complex conjugation.
We treat the two cases simultaneously by writing \eqref{8.1} as $r(x,t)=\pm q(x,t)^\ast,$ and
we analyze the corresponding reduced equations simultaneously
by keeping in mind that, in our analysis in this section, the upper signs in $\pm$ and $\mp$
refer
to the first case in \eqref{8.1} and the lower signs refer to the second case.

Using \eqref{8.1} in \eqref{1.2} we see that, in each of the two cases, 
the nonlinear system \eqref{1.2} reduces to the single nonlinear equation given by
\begin{equation}\label{8.2}
i\,q_t+q_{xx}\mp i\left(|q|^2 q\right)_x=0,\qquad x\in\mathbb R,\quad t\in\mathbb R.
\end{equation}
The two equations in \eqref{8.2} are usually called the Kaup--Newell equations \cite{KN1978}.
Using \eqref{8.1} in \eqref{2.1} we obtain the corresponding linear system as
\begin{equation}\label{8.3}
\ds\frac{d}{dx}\begin{bmatrix}
\alpha\\
\noalign{\medskip}
\beta
\end{bmatrix}=
\begin{bmatrix}
-i\zeta^2 & \zeta \,q(x,t)\\
\noalign{\medskip}
\pm \zeta\, q(x,t) ^\ast& i\zeta^2
\end{bmatrix}
\begin{bmatrix}
\alpha\\
\noalign{\medskip}
\beta
\end{bmatrix},\qquad x\in\mathbb R.
\end{equation}

In the next theorem, we describe the effect of
the reductions of \eqref{8.1} on the solution to the direct scattering problem 
for \eqref{8.3}. We recall that the direct scattering problem for \eqref{8.3} involves,
when the potential $q(x,t)$ is given,
the determination of the four Jost solutions, the six scattering coefficients, and the bound-state information
described by a pair of matrix triplets.

\begin{theorem}
\label{theorem8.1}
Assume that the potential $q(x,t)$  appearing in the first-order system \eqref{8.3}  belongs to the Schwartz class
for each fixed $t\in\mathbb R.$ 
We have the following:

\begin{enumerate}

\item[\text{\rm(a)}] The Jost solutions $\bar\psi(\zeta,x,t)$ and $\bar\phi(\zeta,x,t)$
to \eqref{8.3} with the respective asymptotics \eqref{2.3} and \eqref{2.5}
are related to the Jost solutions $\psi(\zeta,x,t)$ and $\phi(\zeta,x,t)$
to \eqref{8.3} 
with the respective asymptotics \eqref{2.2} and \eqref{2.4}, and we have
\begin{equation} \label{8.4}
 \begin{bmatrix}
\bar\psi_1(\zeta,x,t)\\ \noalign{\medskip}\bar\psi_2(\zeta,x,t)
\end{bmatrix}=
\begin{bmatrix}
\psi_2(\zeta,x,t)^\ast
\\
\noalign{\medskip}\pm\psi_1(\zeta,x,t)^\ast\end{bmatrix},\quad
 \begin{bmatrix}
\bar\phi_1(\zeta,x,t)\\ \noalign{\medskip}\bar\phi_2(\zeta,x,t)
\end{bmatrix}=
\begin{bmatrix}
\pm\phi_2(\zeta,x,t)^\ast
\\
\noalign{\medskip}\phi_1(\zeta,x,t)^\ast\end{bmatrix},
\end{equation}
where the subscripts are used to denote the components as in 
\eqref{2.10} and \eqref{2.11}.

\item[\text{\rm(b)}] The scattering coefficients $\bar T(\zeta,t),$ $\bar R(\zeta,t),$ and
$\bar L(\zeta,t)$ for \eqref{8.3} appearing in \eqref{2.7} and \eqref{2.9}
are related to the
 scattering coefficients $T(\zeta,t),$ $R(\zeta,t),$ and
$L(\zeta,t)$ for \eqref{8.3} appearing in \eqref{2.6} and \eqref{2.8} as
\begin{equation} \label{8.5}
\bar T(\zeta,t)=T(\zeta,t)^\ast,\quad
\bar R(\zeta,t)=\pm R(\zeta,t)^\ast,\quad
\bar L(\zeta,t)=\pm L(\zeta,t)^\ast.
\end{equation}

\item[\text{\rm(c)}] The matrix triplets $(A,B,C)$ and $(\bar A,\bar B,\bar C)$ 
appearing in the scattering data set \eqref{2.38} and
describing the bound-state information for \eqref{8.3} are related to each other as
\begin{equation} \label{8.6}
(\bar A,\bar B,\bar C)=(A^\ast,B^\ast,\pm C^\ast).
\end{equation}
In fact, since $B$ and $\bar B$ appearing in \eqref{8.6} are both real, we actually have 
$\bar B=B.$
We note that \eqref{8.6} implies that the two sets in \eqref{2.27}
containing the bound-state information for \eqref{8.3} are related to each other as
\begin{equation} \label{8.7}
\bar\lambda_j=\lambda_j^\ast,\quad \bar m_j=m_j,\quad \bar c_{jk}=\pm c_{jk}^\ast,
\end{equation}
where we have
$1\le j\le N$ and $0\le k\le m_j-1.$

\end{enumerate}

\end{theorem}

\begin{proof}
The first equality in \eqref{8.4} is proved by showing that the vector on the
right-hand side of that equality satisfies \eqref{8.3} with the spacial
asymptotics given in \eqref{2.3}.
Similarly, the second
equality in \eqref{8.4} is proved by showing that the vector on the
right-hand side of that equality satisfies \eqref{8.3} with the spacial
asymptotics given in \eqref{2.5}.
Hence, the proof of (a) is complete. The proof of (b) is obtained by using \eqref{8.4} and the fact that the two
vectors on the right-hand sides of the two equalities in \eqref{8.4} satisfy the spacial asymptotics given in 
\eqref{2.7} and \eqref{2.9}, respectively.
The proof of (c) is as follows. From the first equality of \eqref{8.5} we get
\begin{equation} \label{8.8}
\bar\lambda_j=\lambda_j^\ast,\quad \bar m_j=m_j,\quad \bar A_j=A_j^\ast,\quad \bar N=N,\qquad 1\le j\le N.
\end{equation}
and using \eqref{8.8} in \eqref{2.28}, we obtain $\bar A=A^\ast.$
Next, we use the second equality of \eqref{8.8} and determine that
the number of entries in the two column vectors $B_j$ and $\bar B_j$ must be equal.
Then, from the first equalities of \eqref{2.32} and \eqref{2.34} we see that
$\bar B_j=B_j.$ Using $\bar B_j=B_j$ and the fourth equality of \eqref{8.8} in \eqref{2.29}, we obtain $\bar B=B.$ 
Since $B$ and $\bar B$ are real, we also conclude $\bar B=B^\ast.$ Next, in order to prove that
$\bar c_{jk}=\pm c_{jk}^\ast,$ we examine the procedure to
construct $c_{jk}$ and $\bar c_{jk}$ from the data set containing the two transmission coefficients $T(\zeta,0)$ and 
$\bar T(\zeta,0)$ and the four Jost solutions
$\psi(\zeta,x,0),$ $\bar\psi(\zeta,x,0),$ 
$\phi(\zeta,x,0),$ 
$\bar\phi(\zeta,x,0).$ 
That procedure is summarized in Section~3 of \cite{AE2022b} with the
details provided in \cite{AE2022a}. We briefly indicate how that procedure yields 
$\bar c_{jk}=\pm c_{jk}^\ast.$
From (3.1) and (3.2) of
\cite{AE2022b}, we know that 
the transmission coefficients have the expansions
\begin{equation}
\label{8.9}
T(\zeta)= \ds\frac{t_{jm_j}}{(\lambda-\lambda_j)^{m_j}}+\ds\frac{t_{j(m_j-1)}}{(\lambda-\lambda_j)^{m_j-1}}+
\cdots+\ds\frac{t_{j1}}{(\lambda-\lambda_j)}+O\left(1\right),\qquad \lambda\to\lambda_j,
\end{equation}
\begin{equation}
\label{8.10}
\bar T(\zeta)=\ds\frac{\bar t_{j\bar m_j}}{(\lambda-\bar\lambda_j)^{\bar m_j}}+\ds\frac{\bar t_{j(\bar m_j-1)}}
{(\lambda-\bar\lambda_j)^{\bar m_j-1}}+
\cdots+\ds\frac{\bar t_{j1}}{(\lambda-\bar\lambda_j)}+O\left(1\right),\qquad \lambda\to\bar\lambda_j,
\end{equation}
where we refer to the double-indexed quantities $t_{jk}$ and $\bar t_{jk}$
as the ``residues."
By taking the complex conjugate in \eqref{8.9} and using the first two equalities of \eqref{8.8} in \eqref{8.10},
we determine that the residues
satisfy $\bar t_{jk}=t_{jk}^\ast.$ Using (3.3) and (3.5) of \cite{AE2022b}
we construct the double-indexed dependency constants $\gamma_{jk}$ and
$\bar \gamma_{jk}.$ Next, using \eqref{8.4} in (3.3) and (3.5) of \cite{AE2022b}
we establish that $\bar \gamma_{jk}=\pm \gamma_{jk}^\ast.$
As described in Section~3 of \cite{AE2022b}, each norming constant
 $c_{jk}$ consists of a summation of terms containing the products of
$t_{jl}$ and $\gamma_{jp},$ where for each fixed $j$ we have
$1\le l\le m_j$ and $0\le p\le m_j-1.$ 
Similarly,
 each norming constant
 $\bar c_{jk}$ consists of a summation of terms containing the products of
$\bar t_{jl}$ and $\bar \gamma_{jp},$
where for each fixed $j$ we have
$1\le l\le \bar m_j$ and $0\le p\le \bar m_j-1.$ 
Hence, as a result of  $\bar t_{jk}=t_{jk}^\ast$
and  $\bar \gamma_{jk}=\pm \gamma_{jk}^\ast,$
we get $\bar c_{jk}=\pm c_{jk}^\ast.$
Thus, the third equality of \eqref{8.7} is also established.
Next, using that third equality in the second equalities of
\eqref{2.32} and \eqref{2.34} we obtain
$\bar C_j=\pm C_j^\ast.$ Finally, using
$\bar C_j=\pm C_j^\ast$ in \eqref{2.30} we establish $\bar C=\pm C^\ast.$
Thus, the proof is complete. \end{proof}

We recall that the inverse scattering problem for \eqref{8.3} consists of
the determination of the potential $q(x,t)$ 
when the corresponding scattering data set is known.
In order to solve the inverse scattering problem for \eqref{8.3}, we can use a Marchenko 
method, which can be obtained by reducing the Marchenko system \eqref{3.1}
appropriately. 
The reduction of the Marchenko system \eqref{3.1} consisting of four integral equations
to the reduced system of a single Marchenko integral equation can be accomplished as follows.

\begin{enumerate}

\item[\text{\rm(a)}]  By using the reduction for the right reflection coefficient
given in the middle equation in \eqref{8.5} and
the reduction for the bound-state information given in \eqref{8.6},
we observe that the two Marchenko kernels
$\Omega(y,t)$ and $\bar \Omega(y,t)$ defined in \eqref{3.2}
are related to each other as
\begin{equation} \label{8.11}
\bar\Omega(y,t)=\pm \Omega(y,t)^\ast,
\end{equation}
where we have also used \eqref{3.3} in establishing \eqref{8.11}.

\item[\text{\rm(b)}]  By using the reduction for the two Jost solutions
$\psi(\zeta,x,t)$ and $\bar\psi(\zeta,x,t)$
given in the first equality in \eqref{8.4},
we determine that the four quantities
$K_1(x,y,t),$
$K_2(x,y,t),$
$\bar K_1(x,y,t),$
$\bar K_2(x,y,t)$
defined in \eqref{3.4}--\eqref{3.7},
respectively, are related to each other as
\begin{equation} \label{8.12}
\bar K_1(x,y,t)=K_2(x,y,t)^\ast,
\quad \bar K_2(x,y,t)=\pm K_1(x,y,t)^\ast. 
\end{equation}

\item[\text{\rm(c)}]  
Next, we use the reductions 
\eqref{8.11} and \eqref{8.12} in the Marchenko system \eqref{3.1}
or equivalently in the uncoupled Marchenko system
\eqref{3.8} and the auxiliary system \eqref{3.9}.
In fact, it is the best to use the reductions 
\eqref{8.11} and \eqref{8.12} only in the first line of \eqref{3.8},
and this yields the reduced Marchenko integral equation
\begin{equation} \label{8.13}
K_1(x,y,t)\pm \Omega(x+y,t)\pm i\ds\int_x^\infty dz\,
K_1(x,z,t)\,\Omega'(z+s,t)\,\Omega(s+y,t)^\ast=0,\qquad y>x,
\end{equation}
where the only unknown is the quantity $K_1(x,y,t).$

\end{enumerate}

We have derived the Marchenko equation 
for the linear system \eqref{8.3},
and it is given in \eqref{8.13}.
We remark that \eqref{8.13} uses
as input the quantity $\Omega(y,t)$ 
defined in the first equality in \eqref{3.2}, 
and the quantity $\Omega(y,t)$ itself
is constructed from the right reflection
coefficient $R(\zeta,t)$ and the matrix
triplet $(A,B,C)$ alone.

In the next theorem
we describe the recovery of the potential
$q(x,t)$ from the solution to \eqref{8.13}.

\begin{theorem}
\label{theorem8.2}
Assume that the potential $q(x,t)$  appearing in the first-order system \eqref{8.3} belongs to the Schwartz class
for each fixed $t\in\mathbb R.$ 
Then, $q(x,t)$ can be recovered
from the solution $K_1(x,y,t)$ to
the reduced Marchenko integral equation \eqref{8.13}
via
\begin{equation} \label{8.14}
q(x,t)=-2 K_1(x,x,t) \,\exp\left(\mp 4i\ds\int_x^\infty dz\, |K_1(z,z,t)|^2\right).
\end{equation}
\end{theorem}

\begin{proof}
We already know how to construct $q(x,t)$ from the solution to
the Marchenko system \eqref{3.1}.
It turns out that it is possible to construct the
potential $q(x,t)$ by using only the solution $K_1(x,y,t)$ to
the reduced Marchenko integral equation \eqref{8.13}.
The construction takes place as follows.
Using the second equality of \eqref{8.12} in \eqref{3.26}
we obtain the corresponding quantity $P(x,t)$ as
\begin{equation} \label{8.15}
P(x,t)=\pm |K_1(x,x,t)|^2.
\end{equation}
Then, using \eqref{8.15} in \eqref{3.28},
we get the key quantity $E(x,t)$ as
\begin{equation*}
E(x,t)=\exp\left(\pm 2i\ds\int_{-\infty}^x dz\, |K_1(z,z,t)|^2\right).
\end{equation*}
Finally, with the help of \eqref{8.15}, we use  \eqref{3.29} to construct $q(x,t)$ 
and we get \eqref{8.14}.
\end{proof}

In the next theorem, we describe the construction of $q(x,t)$
appearing in \eqref{8.3} in the reflectionless case.
We provide an explicit expression for $q(x,t)$ 
in terms of the matrix triplet $(A,B,C)$ alone.

\begin{theorem}
\label{theorem8.3}
Suppose that the potential $q(x,t)$ appearing in \eqref{8.3} at $t=0$ belongs to the Schwartz class
and that the corresponding reflection coefficient $R(\zeta,0)$ is zero.
Then, $q(x,t)$ can be constructed explicitly 
in terms of the matrix triplet $(A,B,C)$ alone,
and this can be done by using \eqref{8.14},
where $K_1(x,x,t)$ is explicitly constructed 
in terms of the matrix triplet $(A,B,C).$ 

\end{theorem}

\begin{proof}
Since \eqref{8.3} corresponds to having $r=\pm q^\ast$ in \eqref{2.1}, as seen from \eqref{8.14} it is enough
to construct $K_1(x,x,t)$ in terms of the matrix triplet
$(A,B,C),$ where $K_1(x,x,t)$ is listed in \eqref{4.29}.
From Theorem~\ref{theorem8.1} we know that 
$(\bar A,\bar B,\bar C)$ can be expressed in terms of
$(A,B,C)$ as in \eqref{8.6}. Then, using \eqref{8.6} in \eqref{4.29} we get
\begin{equation}\label{8.16}
K_1(x,x,t)
=-C^\ast\,e^{-i A^\ast x}\,\left(\Gamma(x,t)^\ast\right)^{-1}\,e^{-i A^\ast x-4i \left(A^\ast\right)^2 t}\, B,
\end{equation}
where we have used $\bar B=B$ and $\bar \Gamma=\Gamma^\ast.$
The equality $\bar B=B$ is proved in Theorem~\ref{theorem8.1}(c). The proof of 
$\bar \Gamma=\Gamma^\ast$ can be given as follows.
Using \eqref{8.6} in \eqref{4.9} we obtain
$\bar M=M^\ast.$ Then, using
\eqref{8.6} in \eqref{4.8} we confirm that 
$\bar \Gamma=\Gamma^\ast.$ Next,
using \eqref{8.6} in
\eqref{4.7} and in the first equality of \eqref{4.9}, we 
see that $\Gamma$ can be explicitly constructed in terms of
the matrix triplet $(A,B,C).$ Thus,
we also see that the right-hand side of \eqref{8.16}
can be explicitly constructed in terms of $(A,B,C).$
Hence, the proof is complete.
\end{proof}

In the following proposition, we show that if the unperturbed potentials
$q(x,t)$ and $r(x,t)$ are related to each other
as $r(x,t)=\pm q(x,t)^\ast,$ then the perturbed potentials
$\tilde q(x,t)$ and $\tilde r(x,t)$ are also related to each other
in almost the same manner.

\begin{proposition}
\label{proposition8.4}
Assume that the solution pair $q(x,t)$ and $r(x,t)$ to \eqref{1.2} belongs to the Schwartz class for each fixed $t\in\mathbb R.$  
Let us also assume that
$q(x,t)$ and $r(x,t)$ are related to each other as in \eqref{8.1}, i.e. $r(x,t)=\pm q(x,t)^\ast.$ 
Then, the solution pair $\tilde q(x,t)$ and $\tilde r(x,t)$ to \eqref{1.6} are related to
each other as
\begin{equation}\label{8.17}
\tilde r(x,t)=\pm \ds\frac{1}{|\kappa|^2}\,\tilde q(x,t)^\ast,
\end{equation}
where $\kappa$ is the complex parameter appearing in \eqref{1.6}.
\end{proposition}

\begin{proof}
When \eqref{8.1} holds, we observe that the quantity $E(x,t)$ 
defined in \eqref{1.19}
satisfies
\begin{equation}\label{8.18}
E(x,t)^\ast=\ds\frac{1}{E(x,t)}.
\end{equation}
Using \eqref{8.18} in \eqref{1.20} we obtain \eqref{8.17}.
\end{proof}

\section{Explicit examples}
\label{section9}

In this section we provide some explicit examples to illustrate the theory presented in the previous sections.
We remark that, the results in Sections~\ref{section4} 
already contain the explicit solution formulas in the reflectionless case both for the
unperturbed linear problem \eqref{2.1} and the unperturbed nonlinear problem \eqref{1.2}, respectively,
where those formulas are expressed in terms of matrix exponentials.
Similarly, as described in Sections~\ref{section5} and \ref{section7}, we also have
the explicit solution formulas in the reflectionless case both for the
perturbed linear problem \eqref{5.1} and the perturbed nonlinear problem \eqref{1.6},
where again those formulas are expressed in terms of matrix exponentials.
In general, a matrix exponential function of $x$ and $t$
consists of algebraic combinations of polynomials, exponential functions, and trigonometric
functions in those two independent variables. As the size in the matrix exponential becomes large,
the ``unpacking'' of a matrix exponential, i.e. expressing it in terms of elementary functions
consisting of polynomials, exponential functions, and
trigonometric functions, becomes impractical in the sense that the resulting expressions
become extremely lengthy as the matrix size increases. Nevertheless, in this section, in order to demonstrate the power
of our method, we present a few examples where we unpack
the matrix exponentials and present the solutions in terms of elementary functions not containing
any matrix exponentials.

We mention that, in the reflectionless case, we have prepared a Mathematica notebook 
containing explicit solutions for the perturbed nonlinear system \eqref{1.6}
and for the linear sytem \eqref{5.1}, where the 
user can specify the input by providing the three parameters 
$a,$ $b,$ $\kappa$ appearing in \eqref{1.5}
as well as the two matrix triplets
$(A,B,C)$ and $(\bar A,\bar B,\bar C)$ appearing in \eqref{2.28}--\eqref{2.30}.
Certainly, the special choice $(a,b,\kappa)=(0,0,1)$ in the input yields explicit solutions for the
unperturbed nonlinear system \eqref{1.2} and for the linear system \eqref{2.1}. Our Mathematica notebook 
not only unpacks all the matrix exponentials in the solution formulas, but
also verifies that the expressions for $\tilde q(x,t)$ and $\tilde r(x,t)$ not involving matrix exponentials
indeed satisfy the corresponding nonlinear system \eqref{1.6}.

 In the reflectionless case, 
we have the explicit expressions for all
the relevant quantities both for the unperturbed linear system in \eqref{2.1}
and the 
perturbed linear system \eqref{5.1}.
The relevant quantities include the Jost solutions, the potentials, the transmission coefficients,
the key quantity $E(x,t)$ defined in \eqref{1.19}, and
the complex constant $\mu$ defined in \eqref{2.14}.
When the input set consisting of 
$(a,b,\kappa),$ $(A,B,C),$ and $(\bar A,\bar B,\bar C)$ is specified,
our Mathematica notebook also displays the aforementioned relevant quantities explicitly
in terms of elementary functions and without any matrix exponentials, and it verifies that the Jost solutions
satisfy the corresponding linear systems.

In the first example below, we elaborate on Theorem~\ref{theorem4.3}, by choosing our reflectionless
input data containing two matrix triplets of unequal sizes.

\begin{figure}[!h]
     \centering
         \includegraphics[width=2in]{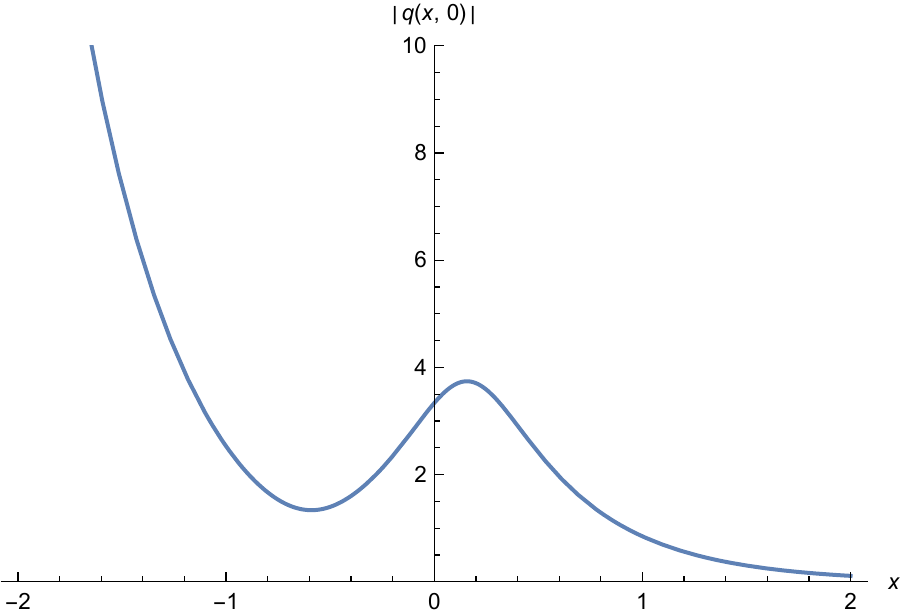}      \hskip .1in
         \includegraphics[width=2in]{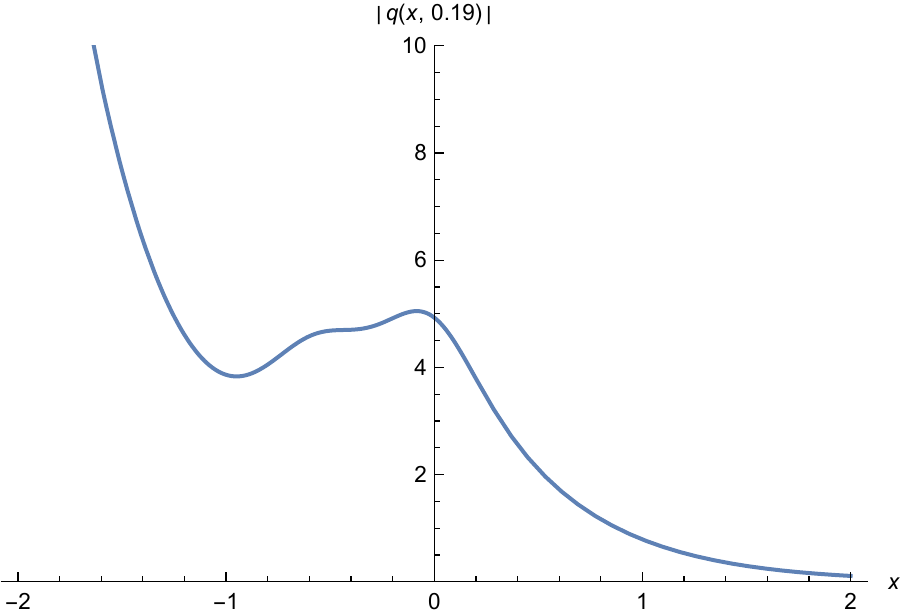} \hskip .1in
         \includegraphics[width=2in]{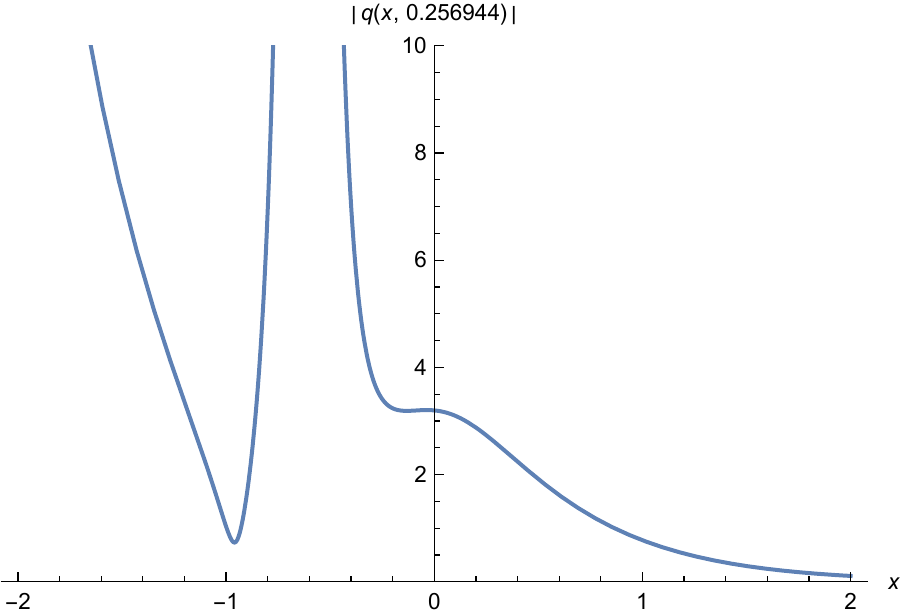}
         \vskip .2in
         \includegraphics[width=2in]{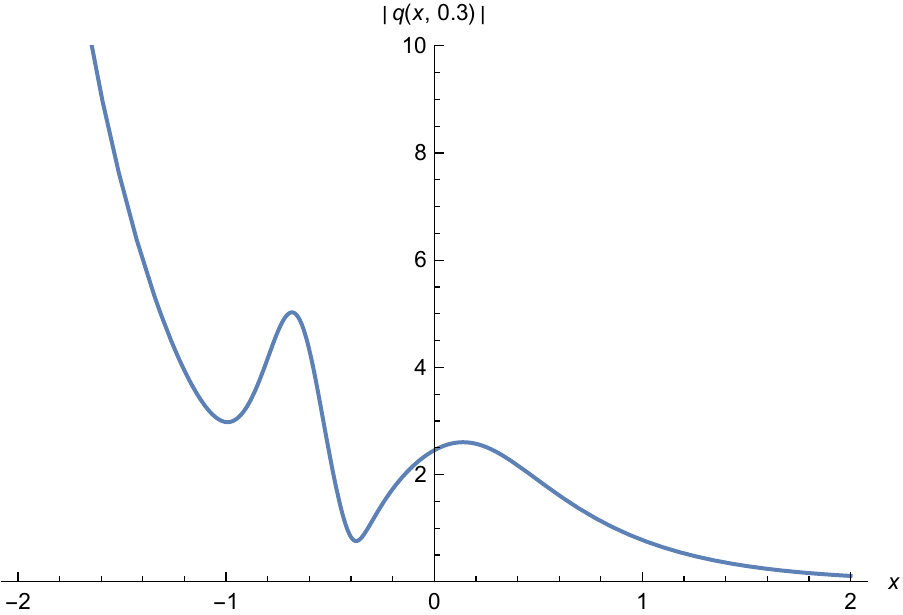} \hskip .1in
         \includegraphics[width=2in]{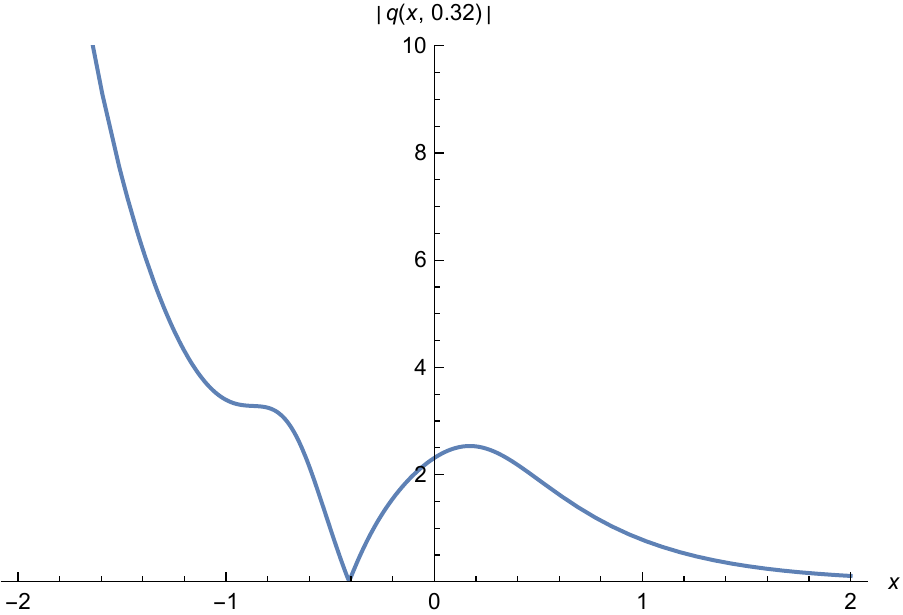} \hskip .1in
         \includegraphics[width=2in]{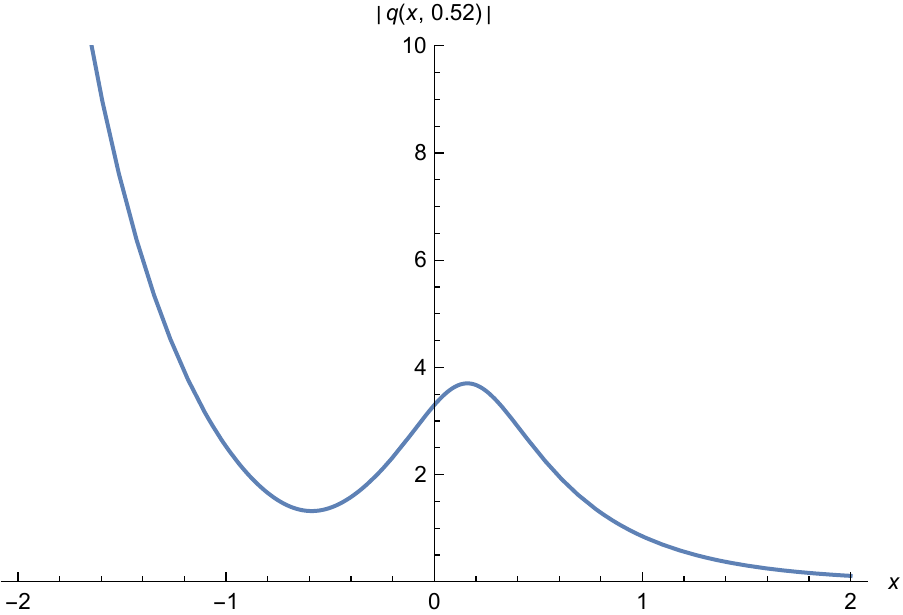}
\caption{The snapshots for $|q(x,t)|$ of \eqref{9.3} at several $t$-values in Example~\ref{example9.1}.}
\label{figure9.1}
\end{figure}

\begin{figure}[!h]
     \centering
         \includegraphics[width=2in]{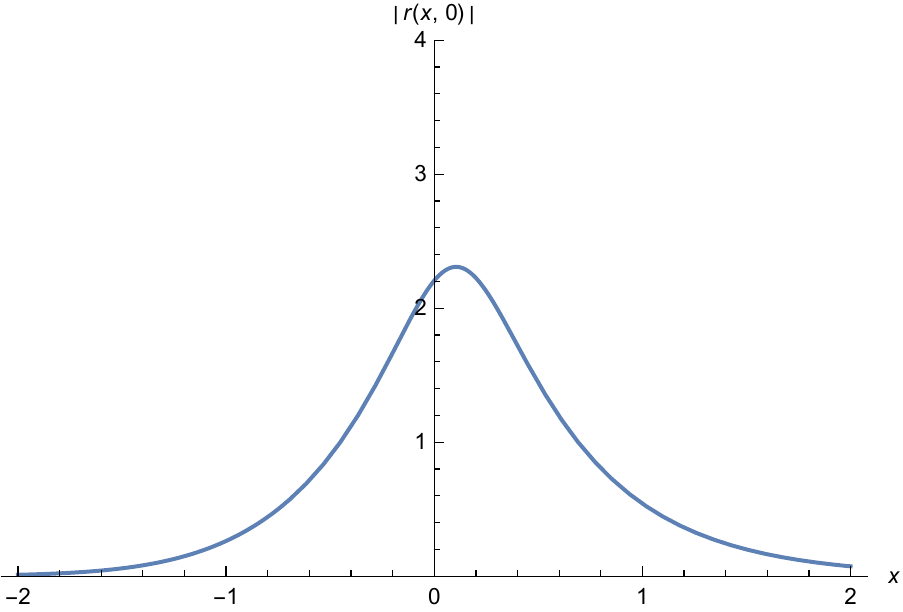}      \hskip .1in
         \includegraphics[width=2in]{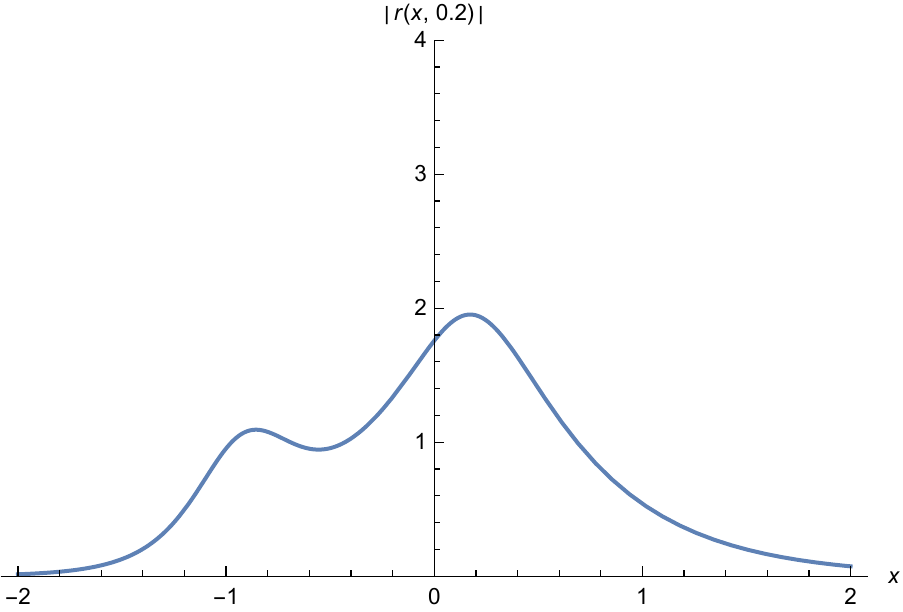} \hskip .1in
         \includegraphics[width=2in]{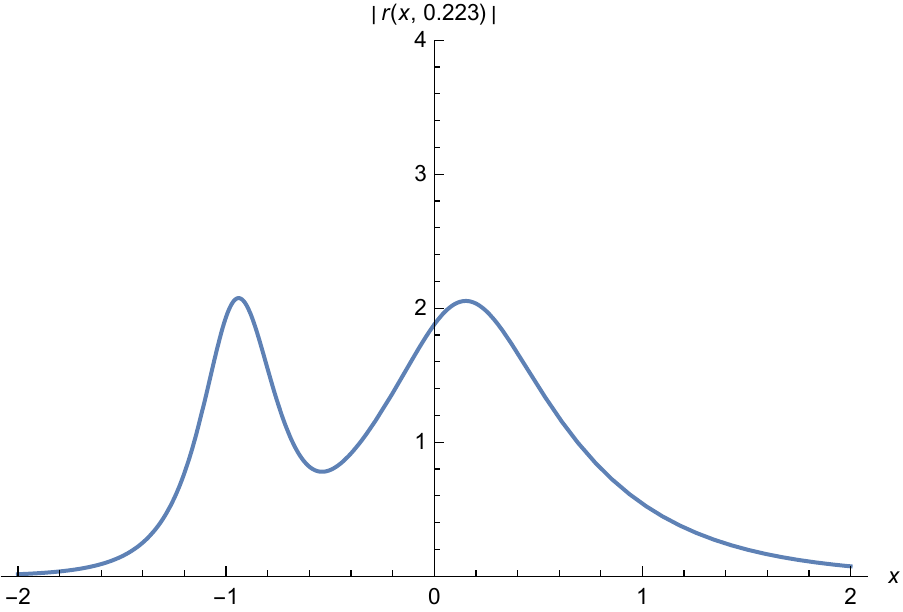}
         \vskip .2in
         \includegraphics[width=2in]{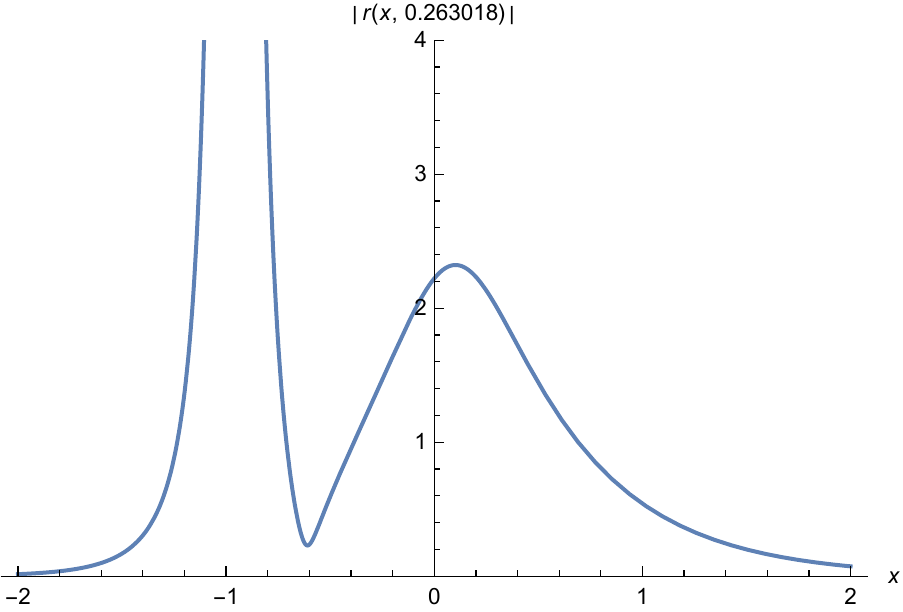} \hskip .1in
         \includegraphics[width=2in]{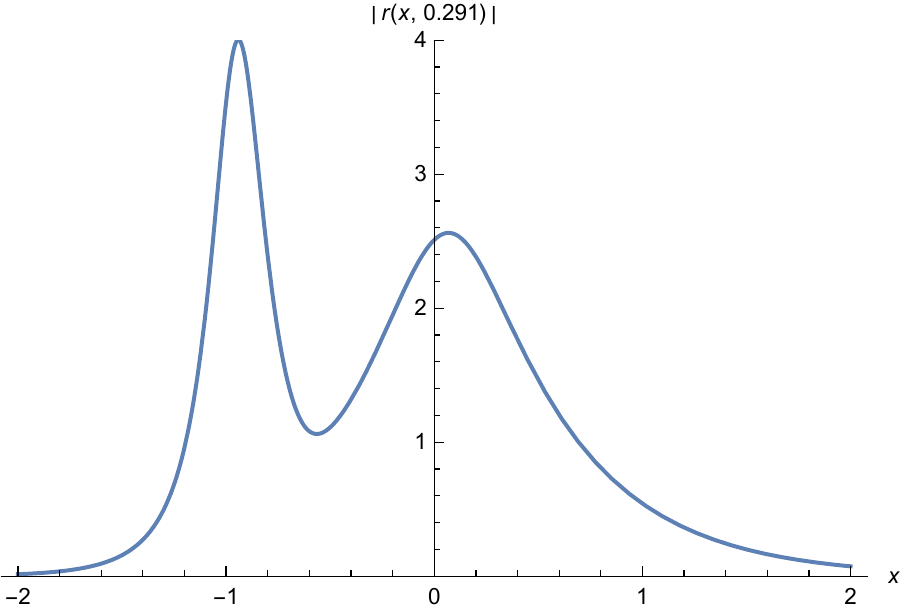} \hskip .1in
         \includegraphics[width=2in]{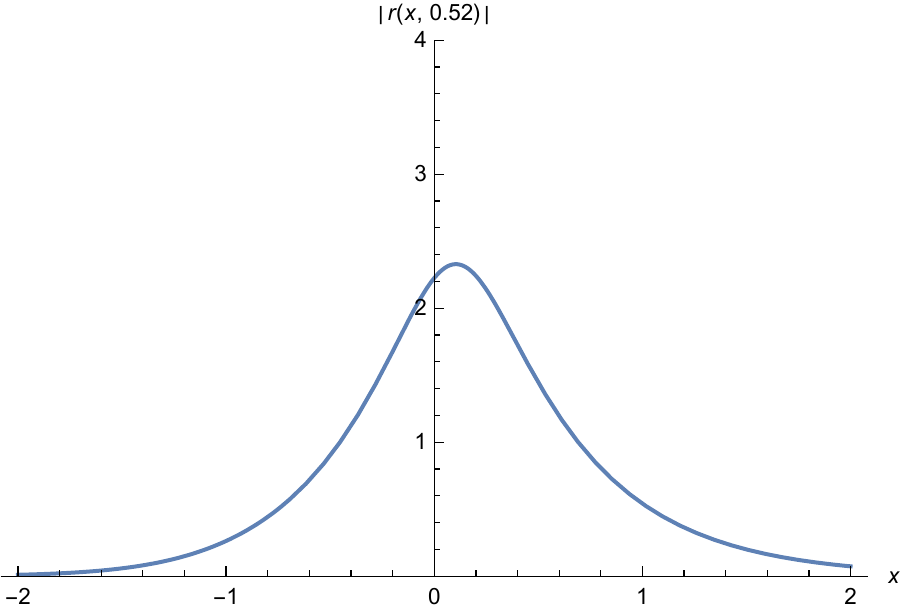}
\caption{The snapshots for $|r(x,t)|$ of \eqref{9.4} at several $t$-values in Example~\ref{example9.1}.}
\label{figure9.2}
\end{figure}

\begin{example}
\label{example9.1}
\normalfont
From Theorem~\ref{theorem4.3}(c), in the reflectionless case we know that the potentials
$q(x,t)$ and $r(x,t)$ cannot both belong to the Schwartz class for all fixed $t\in\mathbb R$
unless 
the sizes of the two matrix triplets $(A,B,C)$ and $(\bar A,\bar B,\bar C)$ 
are equal. To illustrate this, we choose our input data as
\begin{equation}\label{9.1}
a=0, \quad b=0, \quad \kappa=1,
\end{equation}
\begin{equation}\label{9.2}
A=\begin{bmatrix}
i
\end{bmatrix},\quad B=\begin{bmatrix}
1
\end{bmatrix},\quad C=\begin{bmatrix}
2
\end{bmatrix},
\quad
\bar A=\begin{bmatrix}
-i&0\\
0&-2i
\end{bmatrix}, \quad \bar B=\begin{bmatrix}
1\\ 1
\end{bmatrix}, \quad
\bar C=\begin{bmatrix}
3&1
\end{bmatrix},
\end{equation}
where we recall that $a,$ $b,$ and $\kappa$ are the parameters appearing in \eqref{1.5}.
With the help of \eqref{4.31} and \eqref{4.32}, after unpacking all the 
matrix exponentials, we obtain the corresponding potentials $q(x,t)$ and $r(x,t)$ expressed in
terms of elementary functions as
\begin{equation}
\label{9.3}
q(x,t)=\ds\frac{6e^{-2x+4i\,t}\left(18e^{6x}-27i\,e^{2x}-4i\,e^{12i\,t}\right)\left(18e^{6x}+6e^{4x+12i\,t}-ie^{12i\,t}\right)}
{\left(18e^{6x}+27i\,e^{2x}+8i\,e^{12i\,t}\right)^2},
\end{equation}
\begin{equation}
\label{9.4}
r(x,t)=-\ds\frac{72e^{4x-4i\,t}\left(18e^{6x}+27i\,e^{2x}+8i\,e^{12i\,t}\right)}{\left(18i\,e^{6x}+27e^{2x}+4e^{12i\,t}\right)^2}.
\end{equation}
With the help of \eqref{3.11},
\eqref{3.13}, and \eqref{4.25}--\eqref{4.27}, we also get
\begin{equation}
\label{9.5}
E(x,t)=\ds\frac{54e^{2x}+8e^{12i\,t}+36i\,e^{6x}}{27e^{2x}+8e^{12i\,t}-18i\,e^{6x}}, \quad  \mu=2\pi-2i \ln(2),
\end{equation}
where we recall that the scalar function $E(x,t)$ and the constant $\mu$ are the quantities
defined in \eqref{1.19} and \eqref{2.14}, respectively.
One can directly verify that \eqref{1.2} is satisfied by the quantities $q(x,t)$ and $r(x,t)$ appearing in
\eqref{9.3} and \eqref{9.4}, respectively. In this example, the matrix $A$ and the matrix $\bar A$ have unequal sizes, and hence
we cannot expect both $q(x,t)$ and $r(x,t)$ to belong to the
Schwartz class for all fixed $t\in\mathbb R.$ 
In fact, from \eqref{9.3} and \eqref{9.4} we conclude that
\begin{equation}\label{9.6}
q(x,t)=\begin{cases}\ds\frac{3}{8}\, e^{-2x+4it}\left[1+o(1)\right],\qquad x\to-\infty,\\
\noalign{\medskip}
6\, e^{-2x+4it}\left[1+o(1)\right],\qquad x\to+\infty,
\end{cases}
\end{equation}
\begin{equation*}
r(x,t)=\begin{cases}-36i\, e^{4x-24it}\left[1+o(1)\right],\qquad x\to-\infty,\\
\noalign{\medskip}
72\, e^{-2x+4it}\left[1+o(1)\right],\qquad x\to+\infty,
\end{cases}
\end{equation*}
and hence, from the first line of \eqref{9.6} we observe that the function $q(x,t)$ does not belong to the Schwartz class
at any fixed value of $t.$ 
In fact, from the denominator in \eqref{9.3} it follows that $|q(x,t)|$ has singularities when
we have
\begin{equation}\label{9.7}
\left[27 e^{2 x} + 8 \,\cos(12 t)\right]^2 + \left[18 e^{6 x} - 8\, \sin(12 t)\right]^2=0.
\end{equation}
Similarly, from the denominator in \eqref{9.4} it follows that $|r(x,t)|$ has singularities when
we have
\begin{equation}\label{9.8}
\left[27 e^{2 x} + 4 \,\cos(12 t)\right]^2 + \left[18 e^{6 x} +4\, \sin(12 t)\right]^2=0.
\end{equation}
Let us first analyze the singularities of $|q(x,t)|.$
Using $e^{6x}=(e^{2x})^3,$ we can eliminate $x$ in \eqref{9.7}, and hence we see that \eqref{9.7} is satisfied if and only if we have
\begin{equation}\label{9.9}
128\,\cos^3(12 t) + 2187\,\sin(12 t)=0.
\end{equation}
Note that \eqref{9.9} is equivalent to 
\begin{equation*}
\ds\frac{\sin(12 t)}{\cos(12t)}\,\sec^2(12t)=-\ds\frac{128}{2187},
\end{equation*}
which, in turn, is equivalent to
\begin{equation}\label{9.10}
\tan(12 t)\left[1+\tan^2(12t)\right]=-\ds\frac{128}{2187}.
\end{equation}
Since \eqref{9.10} is a cubic equation in $\tan(12t),$ it can be solved explicitly by using algebra.
Consequently, all real solutions to \eqref{9.10} are obtained as
\begin{equation}\label{9.11}
t =\ds\frac{1}{12} \left[\pi n-\tan^{-1}(0.058263\overline{2})\right],\qquad n\in\mathbb Z,
\end{equation}
where an overbar on a digit indicates a roundoff on that digit.
Using the values of $t$ given in \eqref{9.11}, with the help of \eqref{9.7} we evaluate the corresponding $x$-values.
It turns out that the even integer values of $n$ in \eqref{9.11} yield complex $x$-values and hence
they should be excluded. As a result, we determine that the singularities of $|q(x,t)|$
occur in a periodic fashion when we have
\begin{equation}\label{9.12}
(x,t)=\left(-0.60904\overline{7},-0.0048552\overline{7}+\ds\frac{\pi}{12}(2n-1)\right), \qquad n\in\mathbb Z.
\end{equation}
In fact, from \eqref{9.12} we conclude that the singularity when $t>0$ occurs the first time at $n=1,$
which corresponds to $t=0.25694\overline{4}.$
In Figure~\ref{figure9.1} we display the behavior of $|q(x,t)|$ during one period, where from \eqref{9.7} we 
see that the period is equal to $\pi/6.$ A similar analysis can be used to determine the singularities of $|r(x,t)|$ with the help of
\eqref{9.8}, from which we determine that those singularities occur periodically when we have
\begin{equation}\label{9.13}
(x,t)=\left(-0.95482\overline{5},0.0012189 \overline{8}+\ds\frac{\pi}{12}(2n-1)\right), \qquad n\in\mathbb Z.
\end{equation}
From \eqref{9.13} we conclude that the singularity when $t>0$ occurs the first time at $n=1,$
which corresponds to $t=0.26301\overline{8}.$
In Figure~\ref{figure9.2} we display the behavior of $|r(x,t)|$ during one period, where from \eqref{9.13} we 
see that the period is equal to $\pi/6.$ 
With the help of our prepared Mathematica notebook, we can observe the
animations for each of $|q(x,t)|$ and $|r(x,t)|.$
Let us remark that, in our input of \eqref{9.1} and \eqref{9.2}, if we change $C$ and $\bar C$ without changing the 
rest of the input, we may get a different value of the constant $\mu$ than that given in the second equality of
\eqref{9.5}. In fact, if we use
\begin{equation*}
C=\begin{bmatrix}
i
\end{bmatrix},\quad
\bar C=\begin{bmatrix}
i&i
\end{bmatrix},
\end{equation*}
we then get $\mu=-2\pi-2i \ln(2),$
which differs by $4\pi$ from the $\mu$-value given in \eqref{9.5}. We note that
a difference of $4\pi$ in the two $\mu$-values is consistent
with the result stated in Theorem~\ref{theorem4.4},
even though the potentials do not belong to the Schwartz class.
We add the cautious remark that the evaluation of integrals involving the inverse tangent function
by using Mathematica may not yield correct values when the argument of that function is 
complex valued. For example, the use of \eqref{3.13} in Mathematica may not always yield
the correct value of $\mu,$ and thus it is better to use \eqref{2.15} in Mathematica
in the evaluation of $\mu.$
In this example, from \eqref{4.39} we obtain the two transmission coefficients $T(\zeta,t)$ and
$\bar T(\zeta,t)$ as
\begin{equation}\label{9.14}
T(\zeta,t)=\ds\frac{-i(\lambda+i)(\lambda+2i)}{2(\lambda-i)},
\quad \bar T(\zeta,t)=\ds\frac{-2(\lambda-i)}{i(\lambda+i)(\lambda+2i)},
\end{equation}
where we recall that $\lambda= \zeta^2$ as stated in \eqref{2.12}.
Note that the properties of the transmission coefficients
listed in \eqref{9.14} do not agree with \eqref{2.17}--\eqref{2.20} because
the potentials $q(x,t)$ and $r(x,t)$ do not belong to the Schwartz class for each fixed $t\in\mathbb R.$
Finally, let us remark that, in this example, the potential $r(x,t)$ given in \eqref{9.4} satisfies
the second equality in \eqref{4.46} when the integral there is interpreted
as a Cauchy principal value. We determine that the imaginary part of $q(x,t)$ displayed in \eqref{9.3}
satisfies the first equality in \eqref{4.46}, i.e. the Cauchy principal value of the integral of
the imaginary part of $q(x,t)$ over $x\in\mathbb R$ is zero for each fixed $t\in\mathbb R.$
However, the Cauchy principal value of the integral of
the real part of $q(x,t)$ over $x\in\mathbb R$ is equal to $+\infty,$
and hence the first equality in \eqref{4.46} does not hold.

\end{example}

In the next example, we illustrate Theorem~\ref{theorem4.4} with the potentials belonging to the Schwartz class
and demonstrate that a change in the matrices $C$ and $\bar C$
may result in a change in the value of the constant $\mu,$ as stated in \eqref{4.45}.

\begin{figure}[!h]
     \centering
         \includegraphics[width=2in]{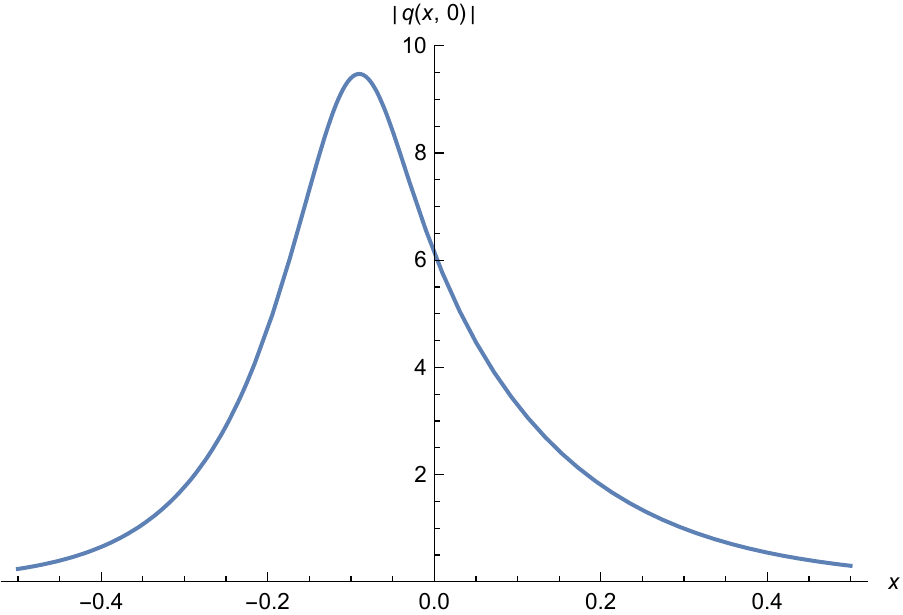}      \hskip .1in
         \includegraphics[width=2in]{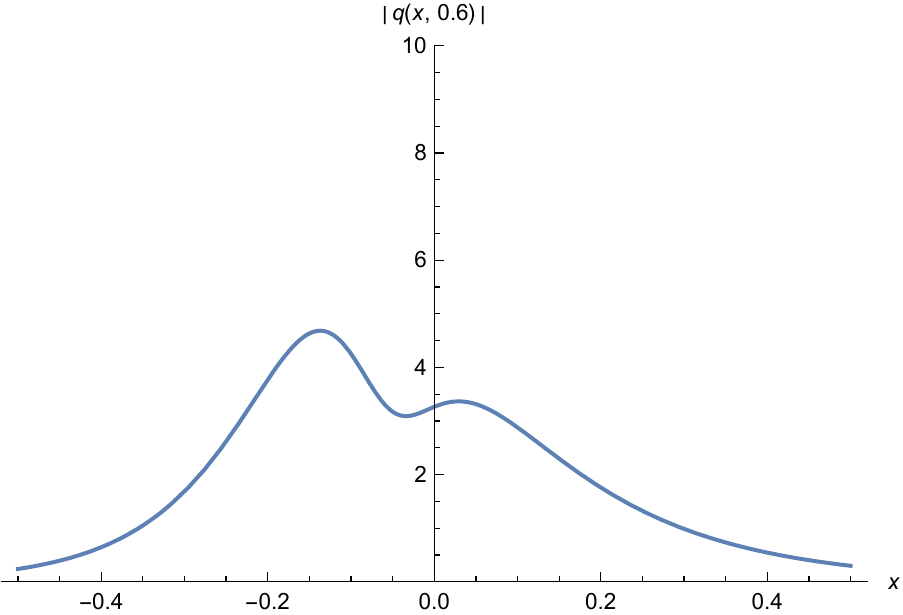} \hskip .1in
         \includegraphics[width=2in]{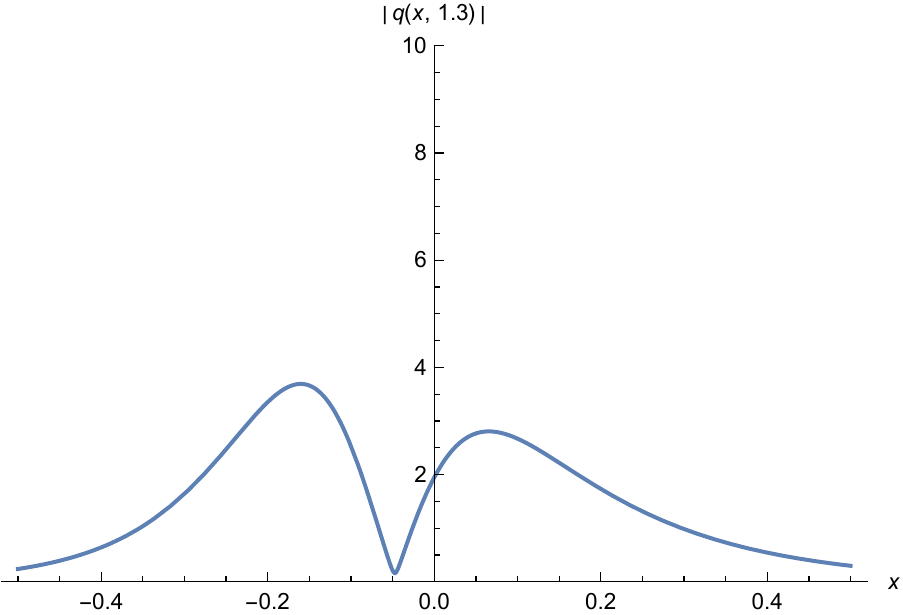}
         \vskip .2in
         \includegraphics[width=2in]{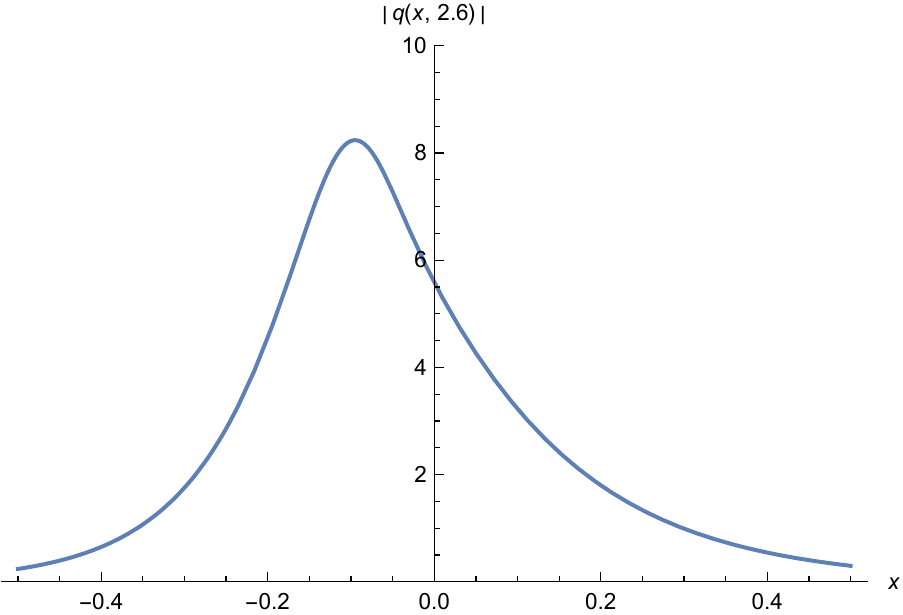} \hskip .1in
         \includegraphics[width=2in]{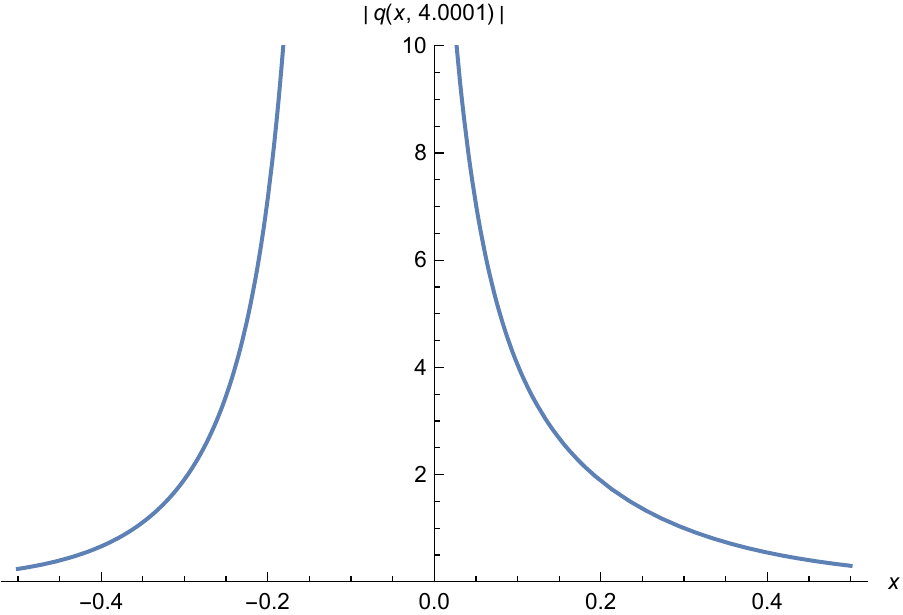} \hskip .1in
         \includegraphics[width=2in]{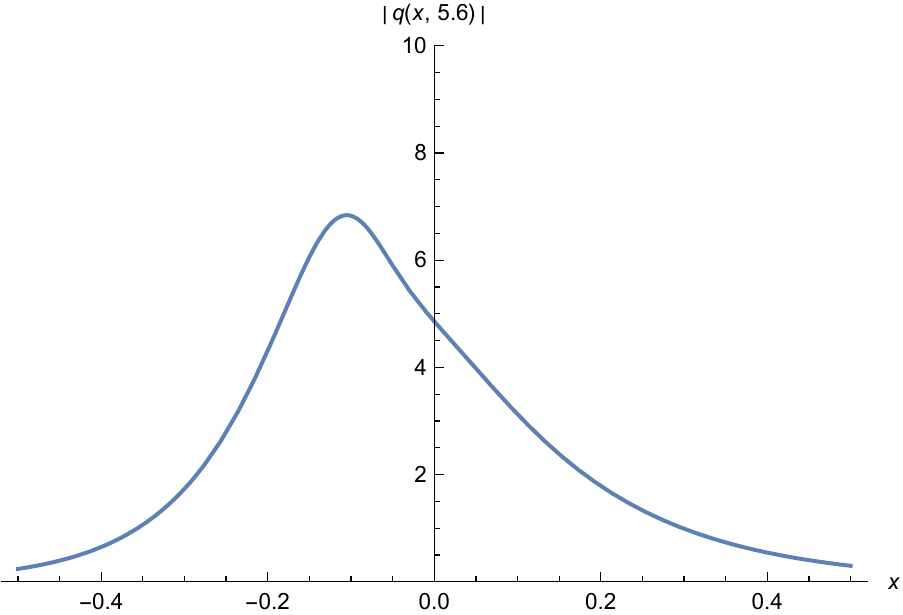}
\caption{The snapshots for $|q(x,t)|$ of \eqref{9.17} at several $t$-values in Example~\ref{example9.2}.}
\label{figure9.3}
\end{figure}

\begin{example}
\label{example9.2}
\normalfont
Using $a,b,\kappa,$ and the matrix triplet $(A,B,C)$ and $(\bar A,\bar B,\bar C)$ given by
\begin{equation}\label{9.15}
a=0, \quad b=0, \quad \kappa=1,
\end{equation}
\begin{equation}\label{9.16}
A=\begin{bmatrix}
5i
\end{bmatrix},\quad B=\begin{bmatrix}
1
\end{bmatrix},\quad C=\begin{bmatrix}
2
\end{bmatrix},
\quad 
\bar A =\begin{bmatrix}
-3i
\end{bmatrix},\quad \bar B=\begin{bmatrix}
1
\end{bmatrix},\quad\bar{ C}=\begin{bmatrix}
3
\end{bmatrix},
\end{equation}
as input in \eqref{4.31} and \eqref{4.32}, after unpacking all the matrix exponentials, 
we obtain the corresponding potentials $q(x,t)$ and $r(x,t)$ as
\begin{equation}
\label{9.17}
q(x,t)=\ds\frac{192\,e^{10(x+10i\,t)}\left(32\,e^{16(x+4i\,t)}-15i\right)}{\left(32\,e^{16(x+4i\,t)}+9i\right)^2},
\end{equation}
\begin{equation}
\label{9.18}
r(x,t)=\ds\frac{128\,e^{6(x-6i\,t)}\left(32\,e^{16(x+4i\,t)}+9i\right)}{\left(32\,e^{16(x+4i\,t)}-15i\right)^2},
\end{equation}
which satisfy \eqref{1.2}. In this example, we obtain the scalar quantity $E(x,t),$ the constant $\mu,$ and the two
transmission coefficients
$T(\zeta,t)$  and $\bar{T}(\zeta,t)$ defined in \eqref{1.19}, \eqref{2.14}, and \eqref{4.39}, respectively, as
\begin{equation}
\label{9.19}
E(x,t)=\ds\frac{-96e^{16(x+4i\,t)}+45i}{160e^{16(x+4i\,t)}+45i}, \quad  \mu=2\pi+2i \ln(5/3),
\end{equation}
\begin{equation}
\label{9.20}
T(\zeta,t)=-\ds\frac{5}{3}\left(\ds\frac{\lambda+3i}{\lambda-5i}\right), \quad 
\bar {T}(\zeta,t)=-\ds\frac{3}{5}\left(\ds\frac{\lambda-5i}{\lambda+3i}\right),
\end{equation}
where we recall that $\lambda= \zeta^2.$ From the denominators in \eqref{9.20} we see that there are two simple bound states
occurring at the $\lambda$-values 
corresponding to the poles of $T(\zeta,t)$ and $\bar T(\zeta,t),$ respectively.
Our prepared Mathematica notebook provides the animations for
$|q(x,t)|$ and $|r(x,t)|.$ In Figure~\ref{figure9.3} we illustrate the periodic behavior of
$|q(x,t)|$ during one period. As observed from
the snapshots in Figure~\ref{figure9.3},
at $t=0$ there are two overlapping solitons, and as time progresses they first separate from each other and
then they overlap again. With further progress in time, their combined amplitude increases to a finite peak value, and
then that combined amplitude decreases when one period is completed.
In time, the behavior described during one period keeps repeating itself.
We remark that during the first period, the combined amplitude reaches a finite peak value
when $t=4.000\overline{1}$ at $x=-0.79\overline{3}.$
The period for the evolution of $|q(x,t)|$ is equal to $5.\overline{6}.$
The behavior of $|r(x,t)|$ is very similar to the behavior of $|q(x,t)|,$
and it is also periodic with the same period. There is only a minor difference
in the behaviors $|q(x,t)|$ and
$|r(x,t)|,$ and that is why we do not include any snapshots for $|r(x,t).$
The minor difference is that, when the two solitons in $|q(x,t)|$
separate the left soliton has a higher amplitude, whereas the right soliton in $|r(x,t)|$ has
a higher amplitude. In this example, if we change the matrices $C$ and $\bar C$ in \eqref{9.16}
and instead use 
\begin{equation*}
C=\begin{bmatrix}
i
\end{bmatrix},
\quad
\bar{ C}=\begin{bmatrix}
i
\end{bmatrix},
\end{equation*}
without changing the rest of the input in \eqref{9.15} and \eqref{9.16}, we get
a different pair of potentials $q(x,t)$ and $r(x,t),$ which are given by
\begin{equation}
\label{9.21}
q(x,t)=\ds\frac{128i\,e^{10(x+10i\,t)}\left(64\,e^{16(x+4i\,t)}+5i\right)}{\left(64\,e^{16(x+4i\,t)}-3i\right)^2},
\end{equation}
\begin{equation}
\label{9.22}
r(x,t)=\ds\frac{128i\,e^{6(x-6i\,t)}\left(64\,e^{16(x+4i\,t)}-3i\right)}{\left(64\,e^{16(x+4i\,t)}+5i\right)^2}.
\end{equation}
The value of the scalar quantity $E(x,t)$ also changes, and the modified value is given by
\begin{equation*}
E(x,t)=-\ds\frac{192e^{16(x+4i\,t)}+15i}{320e^{16(x+4i\,t)}-15i}.
\end{equation*}
The value of the complex constant $\mu$ also changes, and the modified value is given by
\begin{equation}
\label{9.23}
\mu=-2\pi+2i \ln(5/3).
\end{equation}
On the other hand, as seen from \eqref{4.39}, the transmission coefficients $T(\zeta,t)$ and $\bar T(\zeta,t)$
given in \eqref{9.20}
are not affected by the change in $C$ and $\bar C.$
We note that the $\mu$-values in \eqref{9.19} and 
\eqref{9.23} differ from each other by $4\pi,$ which is compatible with \eqref{4.45}.
Let us finally remark on the effect of modifying the parameters $a$ and $b$ appearing
\eqref{9.15} without changing their difference. For example, in the input of \eqref{9.15} and
\eqref{9.16}, let us use
\begin{equation*}
a=5, \quad b=5, \quad \kappa=1,
\end{equation*}
without changing \eqref{9.16}.
In that case, as seen from \eqref{1.20}, the potentials $\tilde q(x,t)$ and $\tilde r(x,t)$
appearing in the nonlinear problem are not affected and they agree with $q(x,t)$ and $r(x,t),$ respectively.
However, the remaining quantities appearing in the 
associated linear problem are affected.
In particular, the Jost solutions are affected as indicated in \eqref{5.6}--\eqref{5.9},
the AKNS pair $(\tilde{\mathcal X},\tilde{\mathcal T})$ is affected as indicated in \eqref{1.9} and
\eqref{1.12},
and the transmission coefficients are affected as indicated in \eqref{5.10} and \eqref{5.11}.
For example, the transmission coefficients $T(\zeta,t)$ and $\bar T(\zeta,t)$ in
\eqref{9.20} are transformed into
$\tilde T(\zeta,t)$ and 
$\tilde{\bar T}(\zeta,t),$ respectively,
which are given by
\begin{equation}\label{9.24}
\tilde T(\zeta,t)=e^{5i \mu/2} \,T(\zeta,t),\quad
\tilde{\bar T}(\zeta,t)=
e^{-5i \mu/2} \,\bar T(\zeta,t).
\end{equation}
From the second equality in \eqref{9.19} we get
$e^{5i \mu/2}=-(3/5)^5,$ and hence
using \eqref{9.20} in \eqref{9.24} we obtain
\begin{equation*}
\tilde T(\zeta,t)=\ds\frac{81}{625}  \left(\ds\frac{\lambda+3i}{\lambda-5i}\right),\quad
\tilde{\bar T}(\zeta,t)=
\ds\frac{625} {81} \left(\ds\frac{\lambda-5i}{\lambda+3i}\right).
\end{equation*}
In this example, the potential pair $q(x,t)$ and $r(x,t)$ displayed
in \eqref{9.17} and \eqref{9.18} satisfies \eqref{4.46}. Similarly,
the potential pair $q(x,t)$ and $r(x,t)$ displayed
in \eqref{9.21} and \eqref{9.22}  also satisfies \eqref{4.46}. 

\end{example}

In the next example we illustrate
the soliton solutions to \eqref{1.2}
when the input matrix triplets each has size $2.$
This example also illustrates the fact that the use of matrix exponentials
in expressing solitons solutions is crucial when the number of
bound states is high. As seen from \eqref{4.31} and \eqref{4.32},
those solutions are expressed in a compact form using matrix
exponentials.
As demonstrated in the next example, 
expressing the solutions in elementary functions
after unpacking the matrix exponentials,
we obtain explicit but
lengthy expressions without gaining much physical insight.

\begin{figure}[!h]
     \centering
         \includegraphics[width=2in]{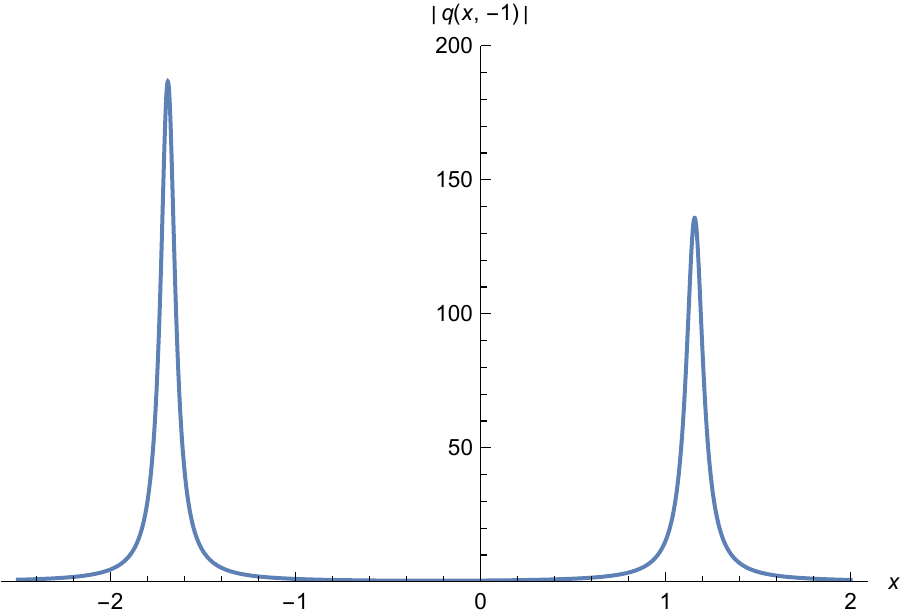}      \hskip .1in
         \includegraphics[width=2in]{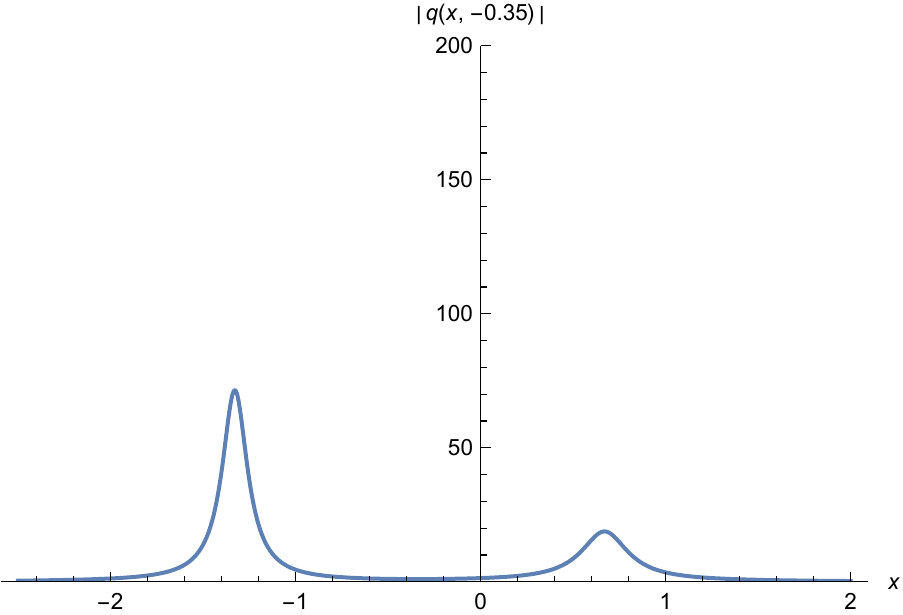} \hskip .1in
         \includegraphics[width=2in]{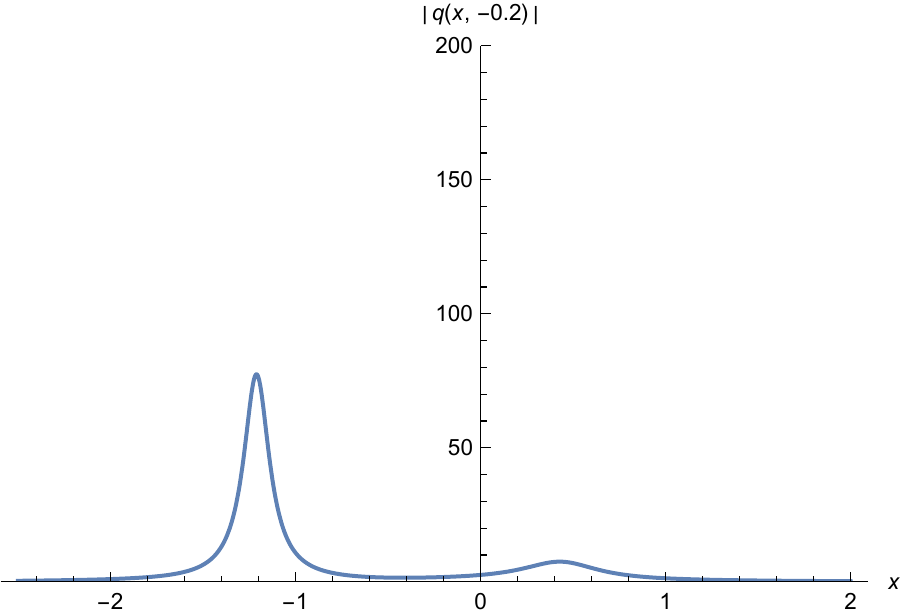}
         \vskip .2in
         \includegraphics[width=2in]{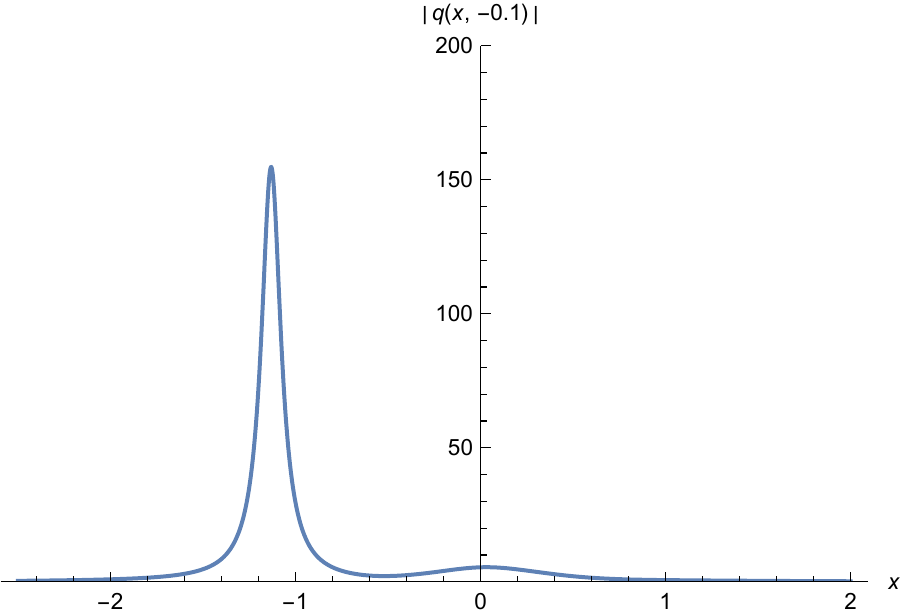} \hskip .1in
         \includegraphics[width=2in]{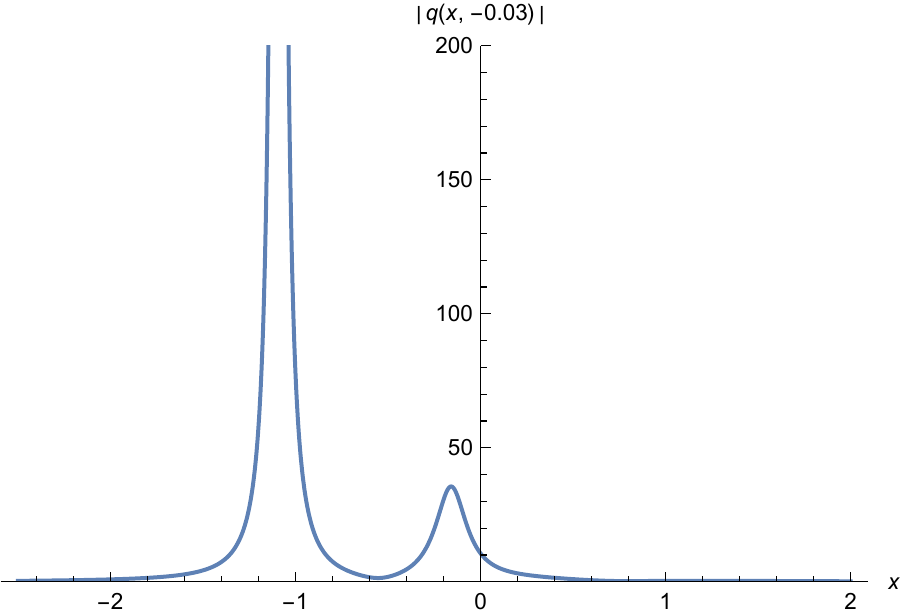} \hskip .1in
         \includegraphics[width=2in]{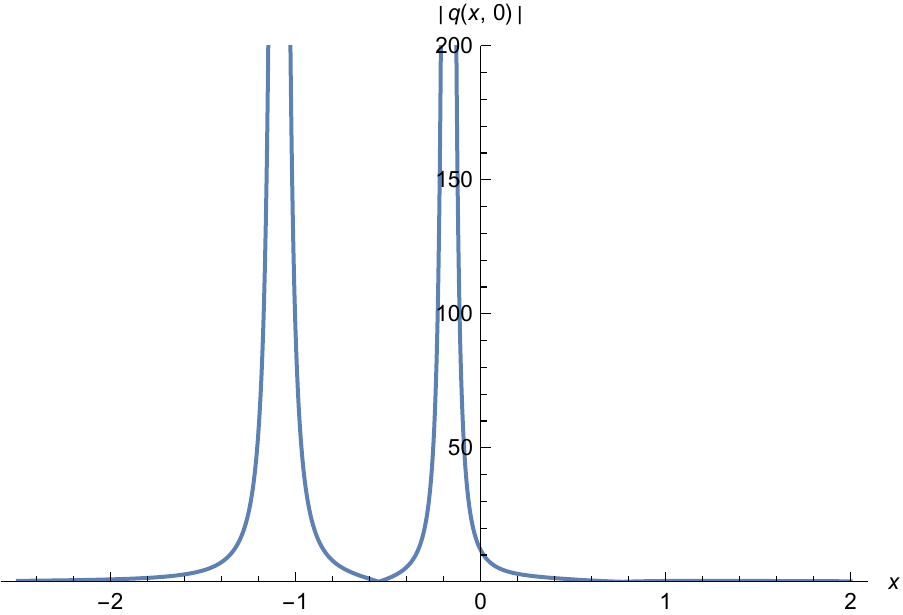}
\caption{The snapshots for $|q(x,t)|$ of \eqref{9.27} at several $t$-values in Example~\ref{example9.3}.}
\label{figure9.4}
\end{figure}

\begin{figure}[!h]
     \centering
         \includegraphics[width=2in]{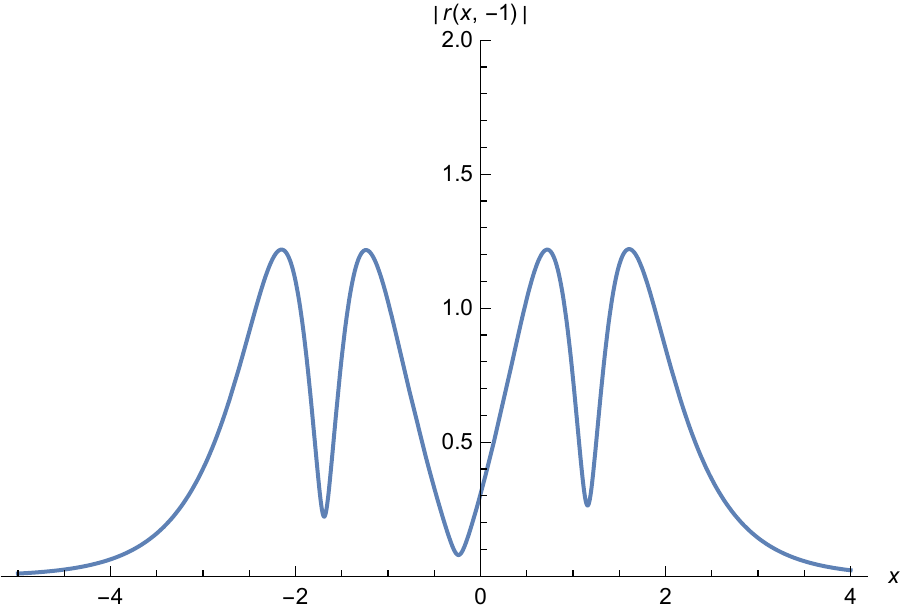}      \hskip .1in
         \includegraphics[width=2in]{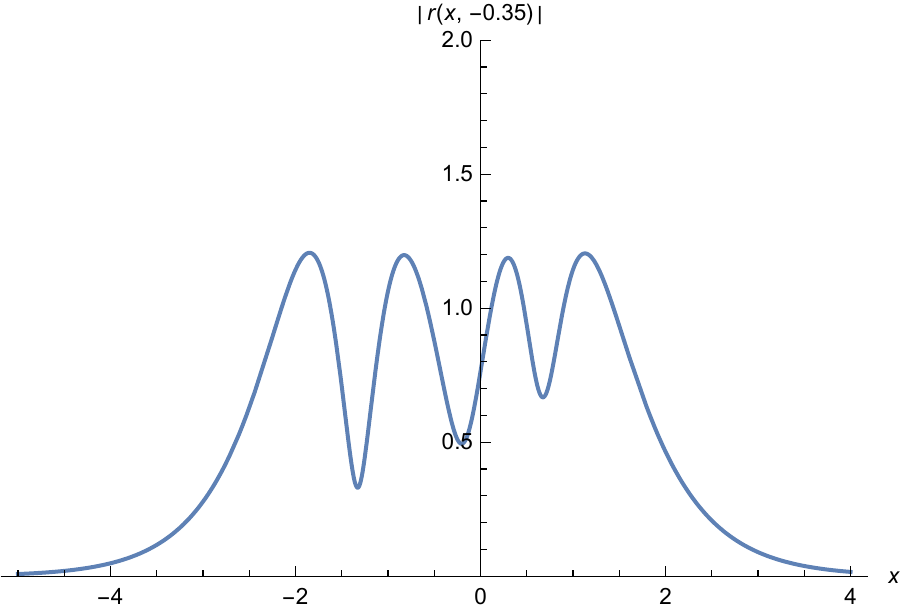} \hskip .1in
         \includegraphics[width=2in]{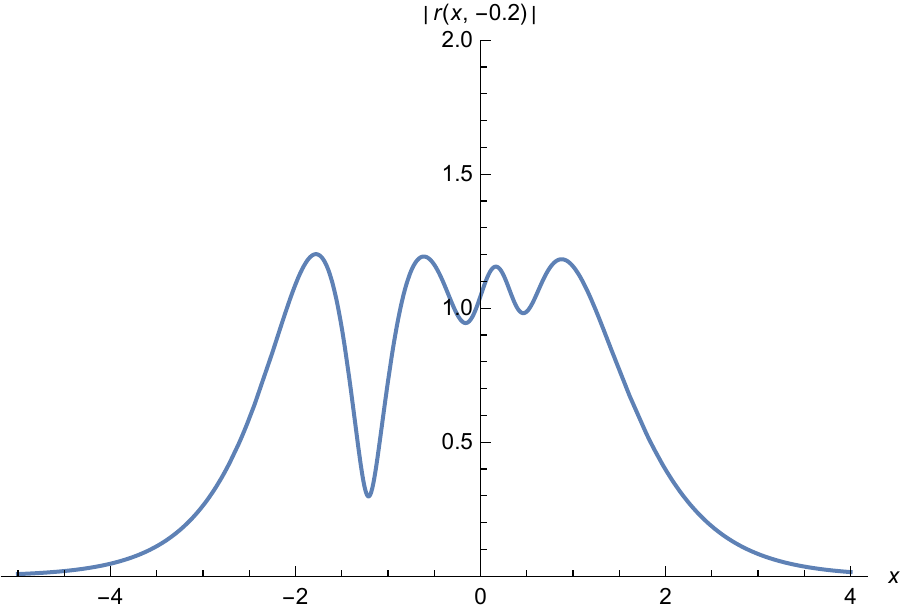}
         \vskip .2in
         \includegraphics[width=2in]{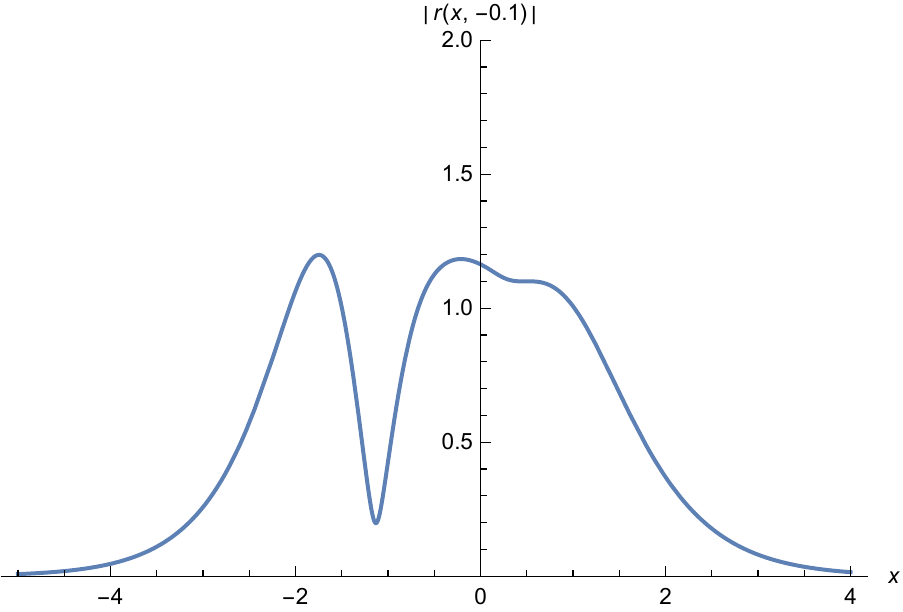} \hskip .1in
         \includegraphics[width=2in]{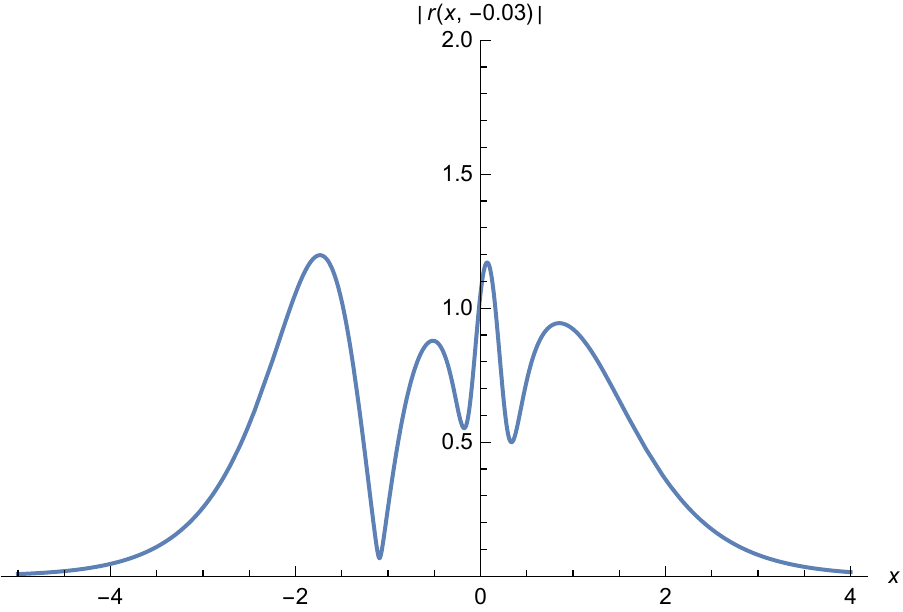} \hskip .1in
         \includegraphics[width=2in]{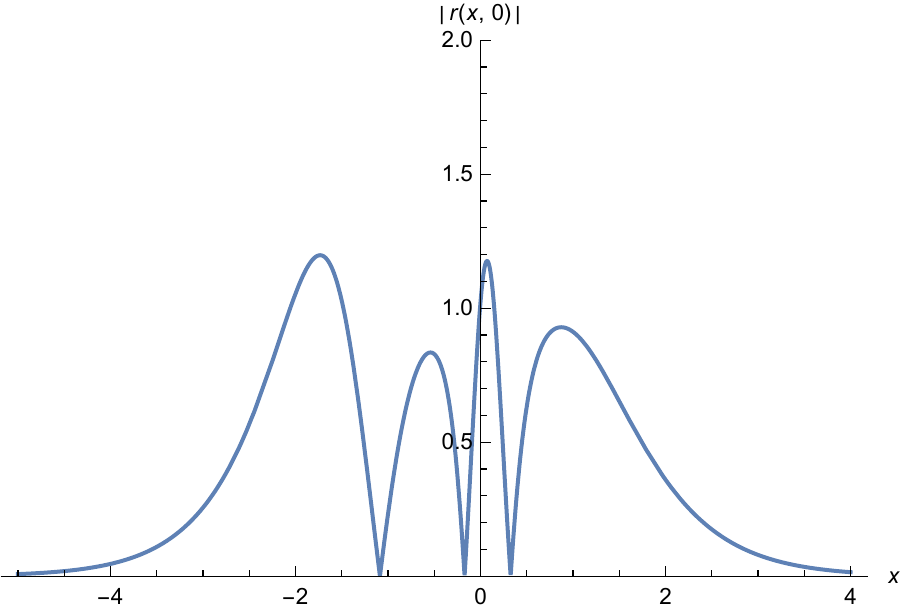}
\caption{The snapshots for $|r(x,t)|$ of \eqref{9.28} at several $t$-values in Example~\ref{example9.3}.}
\label{figure9.5}
\end{figure}

\begin{example}
\label{example9.3}
\normalfont
In this example, we use the reflectionless input data set consisting of
$(a,b,\kappa)$ and the matrix triplets $(A,B,C)$ and $(\bar A,\bar B,\bar C),$ where we have let
\begin{equation}\label{9.25}
a=0, \quad b=0, \quad \kappa=1,
\end{equation}
\begin{equation}\label{9.26}
A=\begin{bmatrix}
i&1\\ \noalign{\medskip}
0&i
\end{bmatrix},\quad B=\begin{bmatrix}
0\\ \noalign{\medskip}
1
\end{bmatrix},\quad C=\begin{bmatrix}
3i&2
\end{bmatrix},
\quad
\bar A =\begin{bmatrix}
-i&1\\ \noalign{\medskip}
0&-i
\end{bmatrix},\quad \bar B=\begin{bmatrix}
0\\ \noalign{\medskip}
1
\end{bmatrix},\quad\bar{ C}=\begin{bmatrix}
2&3i
\end{bmatrix},
\end{equation}
as input in \eqref{4.31} and \eqref{4.32}, we obtain the corresponding potentials $q(x,t)$ and $r(x,t)$ as
\begin{equation}\label{9.27}
q(x,t)=\ds\frac{w_5+w_6}{w^2_7},
\end{equation}
\begin{equation}\label{9.28}
r(x,t)=\ds\frac{w_8+w_9}{w^2_{10}},
\end{equation}
which satisfy \eqref{1.2} and where $w_5,$ $w_6,$ $w_7,$ $w_8,$ $w_9,$ and $w_{10}$ are defined as
\begin{equation*}
w_5:=-32e^{-6x+4i\,t}\left[24t+4e^{4x}(-3i+16t+4i\,x)-i(1+6x)\right],
\end{equation*}
\begin{equation*}
w_6:=9+16e^{4x} \left[5-14x+24x^2+4e^{4x}+8t(i+48t)\right],
\end{equation*}
\begin{equation*}
w_7:=-32\left[-7x+84x^2+4t(48t-11i)\right]+73\cosh(4x)+55\sinh(4x),
\end{equation*}
\begin{equation*}
w_8:=8e^{-6x-4i\,t}\left[9(1-4x+16i\,t)-32e^{4x}(-1+3x+12i\,t)\right],
\end{equation*}
\begin{equation*}
w_9:=9+32e^{4x} \left[7x-12x^2+2e^{4x}+4t(11i-48t)\right],
\end{equation*}
\begin{equation*}
w_{10}:=16\left[5-14x+84x^2+8t(48t+i)\right]+73\cosh(4x)+55\sinh(4x).
\end{equation*}
In this example, we obtain the scalar quantity $E(x,t),$ the constant $\mu,$ and the transmission coefficients
$T(\zeta,t)$  and $\bar{T}(\zeta,t)$ defined in \eqref{1.19}, \eqref{2.14}, and \eqref{4.39}, respectively, as
\begin{equation*}
E(x,t)=\left(\ds\frac{w_{11}+w_{12}}{w_{13}+w_{14}}\right)^{1/2} \exp\bigg(i\,\tan^{-1}(128t\,e^{4x}/w_{15})
-i\,\tan^{-1}(1408t\,e^{4x}/w_{16})\bigg),
\end{equation*}
\begin{equation*}
\mu=0,
\end{equation*}
\begin{equation}\label{9.29}
T(\zeta,t)=\ds\frac{(\lambda+i)^2}{(\lambda-i)^2}, \quad  \bar {T}(\zeta,t)=\ds\frac{(\lambda-i)^2}{(\lambda+i)^2},
\end{equation}
where we recall that $\lambda= \zeta^2$ and we have defined  
\begin{equation*}
w_{11}:=81+4096e^{16x}+\left(288e^{4x}+2048e^{12x}\right)\left(5-14x+24x^2+384t^2\right),
\end{equation*}
\begin{equation*}
w_{12}:=128e^{8x}\left[59-280x+872x^2-1344x^3+1152x^4+128t^2(61-168x+288x^2)+294912t^4\right],
\end{equation*}
\begin{equation*}
w_{13}:=81+4096e^{16x}-\left(576e^{4x}+4096e^{12x} \right)\left(-7x+12x^2+192t^2\right),
\end{equation*}
\begin{equation*}
w_{14}:=128e^{8x}\left[9+392x^2-1344x^3+1152x^4+128t^2(121-168x+288x^2)+294912t^4\right],
\end{equation*}
\begin{equation*}
w_{15}:=9+64e^{8x}+16e^{4x}\left(5-14x+24x^2+384t^2\right),
\end{equation*}
\begin{equation*}
w_{16}:=9+64e^{8x}-32e^{4x}\left(-7x+12x^2+192t^2\right).
\end{equation*}
As seen from the denominators in \eqref{9.29}, there are two bound states corresponding to the
poles of the transmission coefficients, each with multiplicity two.
In this example, we have the even symmetry in time for the absolute values of the potentials 
$q(x,t)$ and $r(x,t)$ given in \eqref{9.27} and \eqref{9.28}, respectively, i.e. we have
\begin{equation*}
|q(x,-t)|=|q(x,t)|,\quad |r(x,-t)|=|r(x,t)|,\qquad
x\in\mathbb R,\quad t\in\mathbb R.
\end{equation*}
In Figures~\ref{figure9.4} and \ref{figure9.5} 
we present some snapshots of 
$|q(x,t)|$ and $|r(x,t)|,$
respectively.
Our prepared Mathematica notebook provides the animations
illustrating the time evolutions of $|q(x,t)|$ and $|r(x,t)|.$
An analysis on the denominators of $|q(x,t)|$ and $|r(x,t)|$
can be carried as in Example~\ref{example9.1}.
We determine that
$|q(x,t)|$ has singularities occurring at certain discrete times at which
$|q(x,t)|$ becomes equal to $+\infty$ 
at one particular $x$-value. On the other hand, 
$|r(x,t)|$ has no singularities when $x\in\mathbb R$ and $t\in\mathbb R.$
The time evolution of 
$|q(x,t)|$ depicted in Figure~\ref{figure9.4} is as follows.
First, the two solitons are far apart from each other. As they approach each other, their
amplitudes keep changing and at certain times they develop singularities. While approaching
each other, also their speeds increase. 
As shown in Figure~\ref{figure9.4}, the two solitons interact nonlinearly with each other, 
and during the interaction their amplitudes get smaller
and their widths get larger, but they do not overlap. Then, they are repelled from each other and move away from each other.
As they move away from each other, their amplitudes keep changing and at some particular discrete times 
they develop singularities where $|q(x,t)|$ becomes equal to $+\infty$
at one particular $x$-value.
The evolution of $|r(x,t)|$ is depicted in Figure~\ref{figure9.5} 
and is as follows. At first, there are two soliton pairs moving toward each other and their speeds increase
as the soliton pairs get closer. In each double soliton, the distance between the two peaks remains unchanged.
The nonlinear interactions take place roughly during the time interval $t\in(-0.3,0.3).$ 
The soliton on the left of the left soliton pair does not interact with the rest.
After the soliton pairs complete their collision, the pairs move backward and their speeds decrease as they move away from each other.
Our prepared Mathematica notebook 
provides the animations of $|q(x,t)|$ and $|r(x,t)|,$
expresses all the relevant quantities by unpacking the matrix exponentials,
and verifies that both the linear and nonlinear systems given in \eqref{2.1}
and \eqref{1.2}, respectively, are satisfied.

\end{example}

In the next example, again in the reflectionless case, we use the same input data set
used in Example~\ref{example9.3}, but for three different sets for the matrices
$C$ and $\bar C.$ In all the three cases, the corresponding potentials
$q(x,t)$ and $r(x,t)$ have no singularities and they all
belong to the Schwartz class for each fixed $t\in\mathbb R.$

\begin{figure}[!h]
     \centering
         \includegraphics[width=2in]{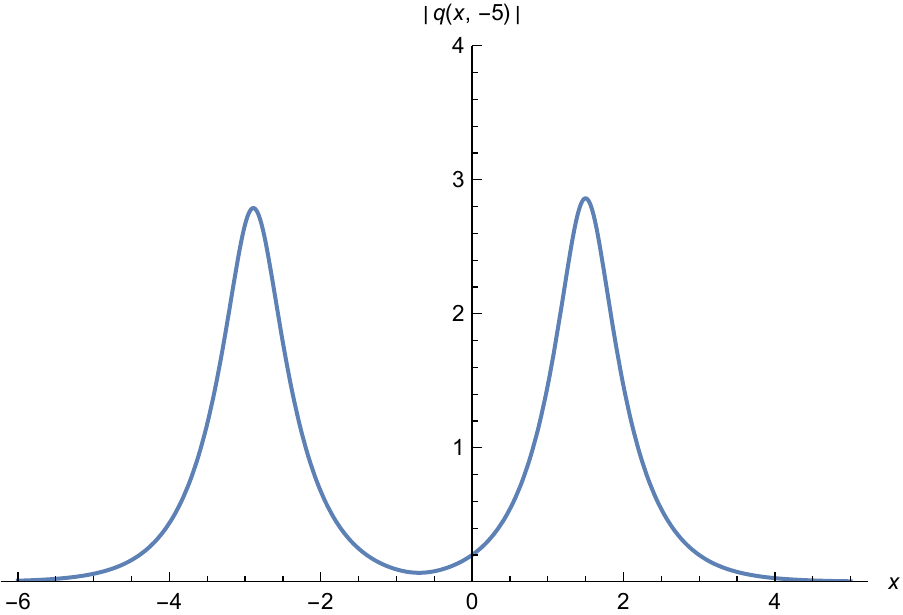}      \hskip .1in
         \includegraphics[width=2in]{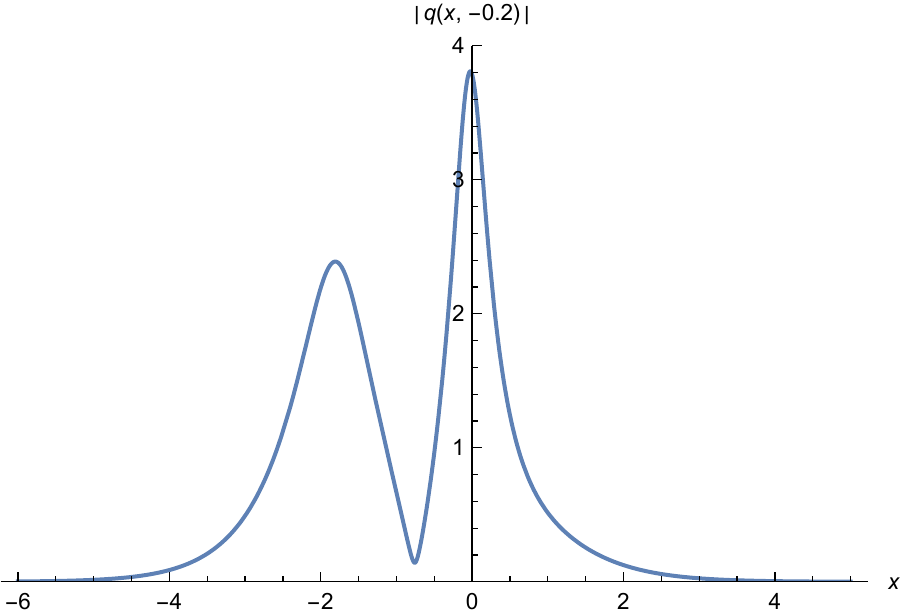} \hskip .1in
         \includegraphics[width=2in]{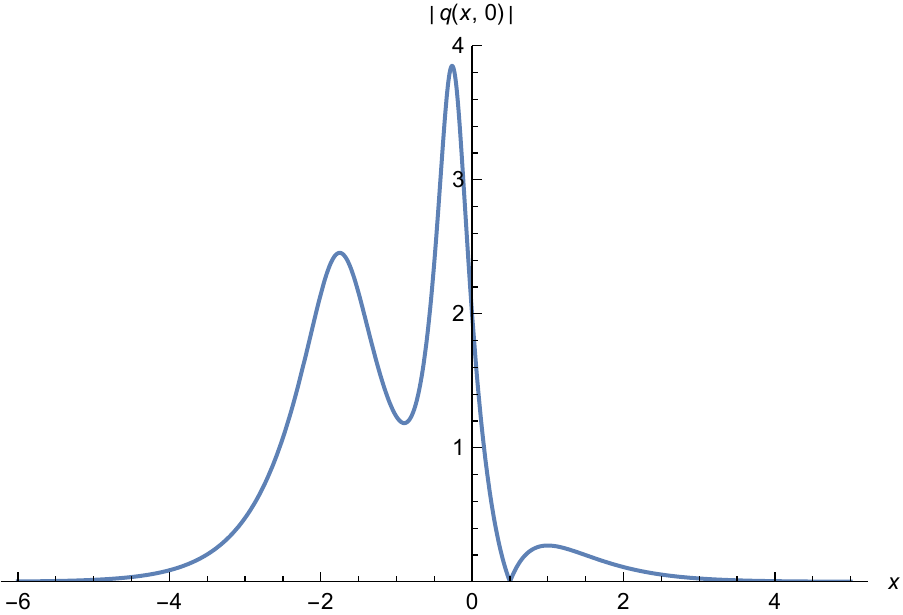}
         \vskip .2in
         \includegraphics[width=2in]{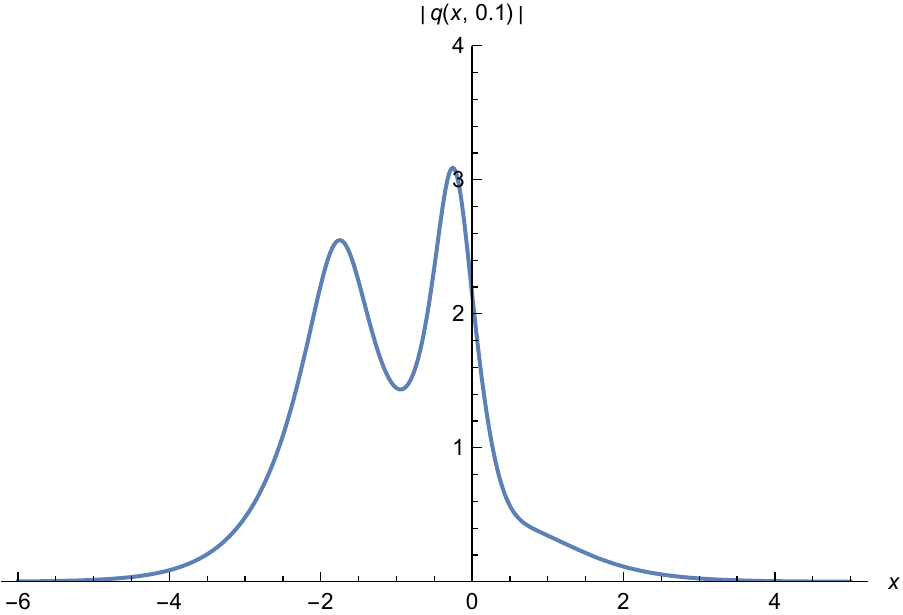} \hskip .1in
         \includegraphics[width=2in]{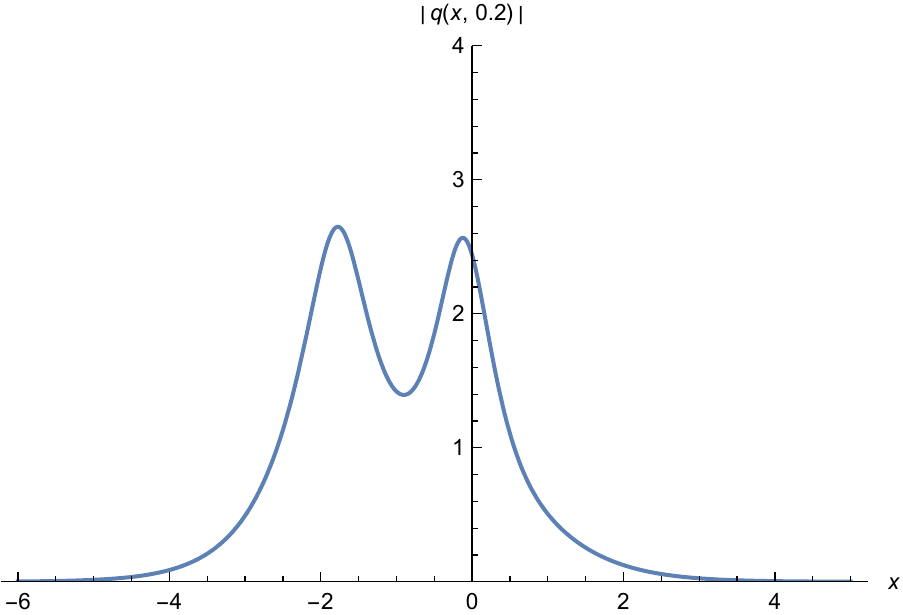} \hskip .1in
         \includegraphics[width=2in]{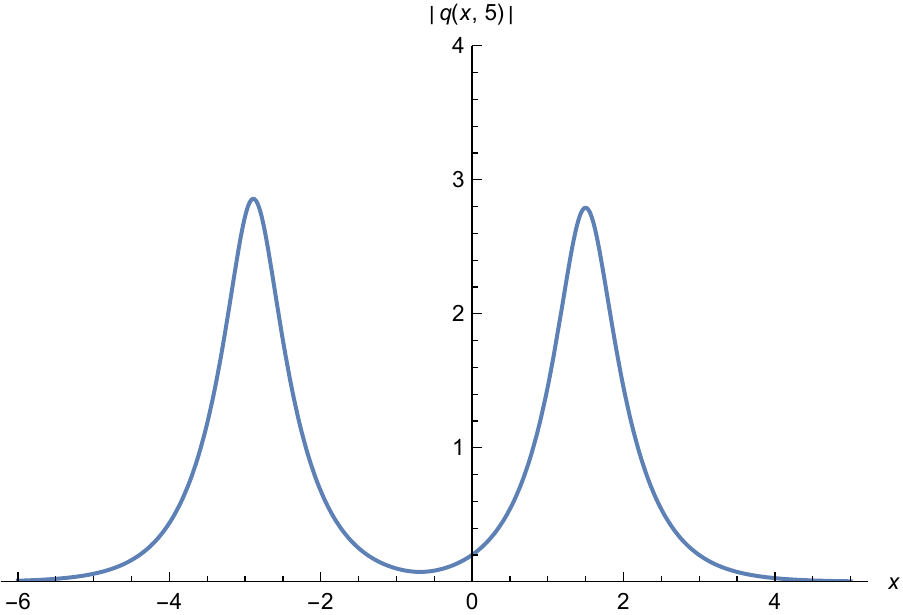}
\caption{The snapshots for $|q(x,t)|$ of \eqref{9.32} at several $t$-values in Example~\ref{example9.4}.}
\label{figure9.6}
\end{figure}

\begin{example}
\label{example9.4}
\normalfont
In the reflectionless case, let us use
\begin{equation}\label{9.30}
a=0, \quad b=0, \quad \kappa=1,
\end{equation}
\begin{equation}\label{9.31}
A=\begin{bmatrix}
i&1\\ \noalign{\medskip}
0&i
\end{bmatrix},\quad B=\begin{bmatrix}
0\\ \noalign{\medskip}
1
\end{bmatrix},\quad C=\begin{bmatrix}
i&1
\end{bmatrix},
\quad
\bar A =\begin{bmatrix}
-i&1\\ \noalign{\medskip}
0&-i
\end{bmatrix},\quad \bar B=\begin{bmatrix}
0\\ \noalign{\medskip}
1
\end{bmatrix},\quad\bar{ C}=\begin{bmatrix}
-i&1
\end{bmatrix},
\end{equation}
which agrees with \eqref{9.25} and differs from
\eqref{9.26} only by the values in the matrices $C$ and $\bar C.$
Using \eqref{9.30} and \eqref{9.31} as input to \eqref{4.31} and \eqref{4.32}, 
after unpacking the matrix exponentials, we obtain the corresponding potentials $q(x,t)$ and $r(x,t)$ as
\begin{equation}\label{9.32}
q(x,t)=\ds\frac{64\,w_{17}}{w_{18}} \exp\bigg(2x+4i\,t-4i\,\tan^{-1}(w_{19}/w_{20})\bigg), \quad r(x,t)=q(x,t)^*,
\end{equation}
where we recall that we use an asterisk to denote complex conjugation and we have defined
\begin{equation}\label{9.33}
w_{17}:=x+4i\,t-8\,e^{4x}(-i+8t+2i\,x),
\end{equation}
\begin{equation}\label{9.34}
w_{18}:=-i+256i\,e^{8x}+32e^{4x}\left(1-4x+8x^2+16it+128t^2\right),
\end{equation}
\begin{equation}\label{9.35}
w_{19}:=32e^{4x}\left(1-4x+8x^2+128t^2\right), \quad w_{20}:=-1+256e^{8x}+512e^{4x}t.
\end{equation}
Contrary to the potentials $q(x,t)$ and $r(x,t)$ given in \eqref{9.27} and \eqref{9.28},
the potentials $q(x,t)$ and $r(x,t)$ given in \eqref{9.32} have no singularities and
they belong to the Schwartz class for each $t\in\mathbb R.$
For the input data given in \eqref{9.30} and \eqref{9.31},
we obtain the scalar quantity $E(x,t)$ defined in \eqref{1.19} and the constant $\mu$ 
appearing in \eqref{2.14} as
\begin{equation}\label{9.36}
E(x,t)=\exp \bigg(-2i \tan^{-1}(w_{19}/w_{20})\bigg), \quad \mu=4\pi.
\end{equation}
Since the transmission coefficients are unaffected when we only change $C$ and $\bar C,$
the transmission coefficients corresponding to the input
in \eqref{9.30} and \eqref{9.31} are still given by
\eqref{9.43}.
Because of the symmetry expressed in the
second equality in \eqref{9.32}, we have $|r(x,t)|=|q(x,t)|,$ and hence
we only discuss the time evolution for $|q(x,t)|.$ 
As seen from Figure~\ref{figure9.6},
there are two solitons that are initially far apart. They move toward each other, and their speeds increase as they get closer.
Then, they interact with each other nonlinearly, and then they move away from each other.
As they move away from each other, they regain their individual shapes. 
Our Mathematica
notebook provides all the quantities related to the linear and nonlinear problems
by unpacking the matrix exponentials, and it
verifies that the corresponding linear and nonlinear equations are satisfied. It also confirms that
the integrals of $q(x,t)$ and $r(x,t)$ over all $x$-values at each fixed
$t\in\mathbb R$  are each zero, as stated in Theorem~\ref{theorem4.6}.
Let us slightly modify the input data given in \eqref{9.30} and \eqref{9.31}, by only changing
the matrices $C$ and $\bar C$ to the new values given as
\begin{equation*}
C=\begin{bmatrix}
-i&-1
\end{bmatrix},
\quad
\bar{ C}=\begin{bmatrix}
-i&1
\end{bmatrix}.
\end{equation*}
Then, with the help of
\eqref{4.31} and \eqref{4.32}, we obtain the corresponding potentials $q(x,t)$ and $r(x,t)$
as
\begin{equation}\label{9.37}
q(x,t)=\ds\frac{64\,w_{22}}{w_{23}} \,\exp\bigg(2x+4i\,t+4i\,\tan^{-1}(w_{19}/w_{21})\bigg), 
\end{equation}
\begin{equation*}
\quad r(x,t)=-q(x,t)^*,
\end{equation*}
and the quantities $E(x,t)$ and $\mu$ are given by
\begin{equation}\label{9.38}
E(x,t)=\exp \bigg(2i \tan^{-1}(w_{19}/w_{21})\bigg), \quad \mu=-4\pi,
\end{equation}
where we have defined
\begin{equation*}
w_{21}:=-1+256\,e^{8x}-512\, e^{4x}t,
\end{equation*}
\begin{equation*}
w_{22}:=x+4i\,t+8\,e^{4x}(-i+8t+2i\,x),
\end{equation*}
\begin{equation*}
w_{23}:=i-256i\,e^{8x}+32\, e^{4x}\left(1-4x+8x^2+16it+128 t^2\right).
\end{equation*}
Our Mathematica notebook provides an animation of $|q(x,t)|$ of \eqref{9.37}, from which we observe that the
time evolution of $|q(x,t)|$ in this modified case is similar to the time evolution of $|q(x,t)|$ of \eqref{9.32}
described earlier and illustrated in Figure~\ref{figure9.6}.
We remark that the $\mu$-value in \eqref{9.38} differs from
the $\mu$-value in \eqref{9.36} by $8\pi,$ which agrees with the result in Theorem~\ref{theorem4.4}.
Let us again modify 
the input data set given in \eqref{9.30} and \eqref{9.31}, by only changing
the matrices $C$ and $\bar C$ to the new values given by
\begin{equation*}
C=\begin{bmatrix}
-i&-1
\end{bmatrix},
\quad
\bar{ C}=\begin{bmatrix}
-i&-1
\end{bmatrix}.
\end{equation*}
In this case, the corresponding potentials $q(x,t)$ and $r(x,t)$ and the corresponding 
quantities $E(x,t)$ and $\mu$ are again explicitly determined, and we have
\begin{equation}\label{9.39}
q(x,t)=\ds\frac{64\,w_{27}\,w_{28}}{w_{29}+w_{30}} \,\exp\bigg(2x+4i\,t+2i\,\tan^{-1}(w_{24}/w_{25})+2i  \,\tan^{-1}(w_{26}/w_{20})    \bigg), 
\end{equation}
\begin{equation*}
r(x,t)=-q(x,t)^\ast,
\end{equation*}
\begin{equation}\label{9.40}
E(x,t)=\ds\frac{\sqrt{w_{31}+w_{32}}}{\sqrt{w_{29}+w_{30}}}\,
\exp \bigg(i \tan^{-1}(w_{24}/w_{25})+2i  \,\tan^{-1}(w_{26}/w_{20})    \bigg), \quad \mu=0,
\end{equation}
where we have defined
\begin{equation*}
w_{24}:=32e^{4x}\left(-3+4x+8x^2+128t^2\right),
\end{equation*}
\begin{equation*}
w_{25}:=-1+256e^{8x}-1536 e^{4x}t,
\end{equation*}
\begin{equation*}
w_{26}:=32e^{4x}\left(1+4x+8x^2+128t^2\right),
\end{equation*}
\begin{equation*}
w_{27}:=x+4i\,t+8\,e^{4x}(i+8t+2i\,x),
\end{equation*}
\begin{equation*}
w_{28}:=-i+256i\,e^{8x}+32e^{4x}\left(-3+4x+8x^2-48 it+128t^2\right),
\end{equation*}
\begin{equation*}
w_{29}:=1 + 65536 \,e ^{16 x}- 1024 \,e ^{4 x} t + 262144 \,e ^{12 x} t,
\end{equation*}
\begin{equation*}
w_{30}:=512 \,e ^{8 x)}\left[1 + 32768 t^4 + 16 x + 64 x^2 + 128 x^3 + 128 x^4 + 
   1024 t^2 (1 + 2 x + 4 x^2)\right],
\end{equation*}
\begin{equation*}
w_{31}:=1 + 65536\,e ^{16 x} + 3072 \,e ^{4 x} t - 786432 \,e ^{12 x} t,
\end{equation*}
\begin{equation*}
w_{32}:=512 \,e ^{8 x} \left[17 + 32768 t^4 - 48 x - 64 x^2 + 128 x^3 + 128 x^4 + 
   1024 t^2 (3 + 2 x + 4 x^2)\right].
\end{equation*}
We obtain an animation of $|q(x,t)|$ of \eqref{9.39} in our Mathematica notebook, from which we observe that the
time evolution of $|q(x,t)|$ also in this modified case is similar to the time evolution of $|q(x,t)|$ of \eqref{9.32}
described earlier and illustrated in Figure~\ref{figure9.6}.
We remark that the $\mu$-value in this case given in \eqref{9.40} differs by $4\pi$
from each of the two $\mu$-values given in \eqref{9.36} and \eqref{9.38},
and those differences are compatible with the result of Theorem~\ref{theorem4.4}.

\end{example}

In the next example, we demonstrate that, instead of choosing \eqref{1.2} 
as the unperturbed nonlinear problem, we could choose any particular case of \eqref{1.6} as
the unperturbed problem.

\begin{example}
\label{example9.5}
\normalfont
As far as the nonlinear problem is concerned,
in order to demonstrate that any particular case of
\eqref{1.6} can be chosen as the unperturbed problem and the rest as the perturbed problem, it is
sufficient to show that we can express $q(x,t)$ and $r(x,t)$
in terms of $\tilde q(x,t)$ and $\tilde r(x,t)$ in a way similar to that given in
\eqref{1.20}.
We proceed as follows. From \eqref{1.20} we observe that \eqref{6.4} holds.
Similar to \eqref{1.19} let us define the quantity $\tilde E(x,t)$ as
\begin{equation}
\label{9.41}
\tilde E(x,t):=\exp\left(\ds\frac{i}{2}\ds\int_{-\infty}^x dz\,\tilde q(z,t)\,\tilde r(z,t)\right).
\end{equation}
Note that the use of \eqref{6.4} in \eqref{9.41} implies that 
$\tilde E(x,t)=E(x,t),$ where $E(x,t)$ is the quantity defined in \eqref{1.19} in terms of
the unperturbed potentials $q(x,t)$ and $r(x,t).$
Thus, we obtain the inverses of the transformations given in \eqref{1.20}, and we have
\begin{equation*}
q(x,t):=\kappa\,\tilde q(x,t)\,\tilde E(x,t)^{a-b},\quad r(x,t):=\ds\frac{1}{\kappa}\, \tilde r(x,t) \, \tilde E(x,t)^{b-a},
\end{equation*}
which proves our claim that any particular case of
\eqref{1.6} can be chosen as the unperturbed problem.
Nevertheless, as far as the linear problem is concerned,
a comparison of \eqref{2.1} and \eqref{5.1} indicates that
the choice of \eqref{2.1} as the unperturbed problem
is the simplest. Since the analysis of the unperturbed and perturbed linear problems
is the crucial part in our paper, we have chosen the particular case with
$(a,b,\kappa)=(0,0,1)$ as our unperturbed problem.

\end{example}

In the following example, in the reflectionless case, we present some explicit solutions to the nonlinear
system \eqref{1.3}, which is also called the Chen--Lee--Liu system, and
to the nonlinear
system \eqref{1.4}, which is also called the Gerdjikov--Ivanov system.

\begin{example}
\label{example9.6}
\normalfont
Let us recall that \eqref{1.3} is obtained from \eqref{1.6} when the three parameters
in \eqref{1.6} are chosen as
$(a,b,\kappa)=(1,0,1).$ In order to illustrate some explicit solution to
\eqref{1.3}, we choose our input as
\begin{equation}\label{9.42}
a=1, \quad b=0, \quad \kappa=1,
\end{equation}
\begin{equation}\label{9.43}
A =\begin{bmatrix}
i&1\\ \noalign{\medskip}
0&i
\end{bmatrix},\quad B=\begin{bmatrix}
0\\ \noalign{\medskip}
1
\end{bmatrix},\quad C=\begin{bmatrix}
i&1
\end{bmatrix},
\quad
\bar A =\begin{bmatrix}
-i&1\\ \noalign{\medskip}
0&-i
\end{bmatrix},\quad \bar B=\begin{bmatrix}
0\\ \noalign{\medskip}
1
\end{bmatrix},\quad\bar{ C}=\begin{bmatrix}
-i&1
\end{bmatrix},
\end{equation}
where we remark that the matrix triplets $(A,B,C)$ and $(\bar A,\bar B,\bar C)$ in
\eqref{9.43} coincide with the triplets in \eqref{9.31}.
Using the unperturbed potentials $q(x,t)$ and $r(x,t)$ appearing in \eqref{9.32}, the quantity 
$E(x,t)$ given in the first equality of \eqref{9.36}, and 
the three parameters listed in \eqref{9.42} we obtain the
explicit solution to \eqref{1.3} as
\begin{equation}\label{9.44}
\tilde q(x,t)=\ds\frac{64\,w_{17}}{w_{18}} \exp\bigg(2x+4i\,t-2i\,\tan^{-1}(w_{19}/w_{20})\bigg), \quad \tilde r(x,t)=\tilde q(x,t)^*,
\end{equation}
where $w_{17},$ $w_{18},$ $w_{19},$ $w_{20}$  are the quantities defined in \eqref{9.33}, \eqref{9.34}, and \eqref{9.35}, respectively.
We remark that the second equality in \eqref{9.44} is compatible with
\eqref{8.17} with $\kappa=1$ because,
as seen from the second equality of \eqref{9.32},
the unperturbed potentials $q(x,t)$ and $r(x,t)$ are complex conjugates of each other.
Let us finally illustrate the explicit solutions to \eqref{1.4}, which corresponds to
choosing the three parameters
in \eqref{1.6} as
$(a,b,\kappa)=(1,-1,1).$ 
As our input, let us use the same matrix triplets 
$(A,B,C)$ and $(\bar A,\bar B,\bar C)$ appearing in
\eqref{9.43} and replace \eqref{9.42} with
\begin{equation}\label{9.45}
a=1, \quad b=-1, \quad \kappa=1.
\end{equation}
Using \eqref{9.45}, the unperturbed potentials $q(x,t)$ and $r(x,t)$ appearing in \eqref{9.32}, and the quantity 
$E(x,t)$ given in the first equality in \eqref{9.36},
we obtain the
explicit solution to \eqref{1.4} as
\begin{equation}\label{9.46}
\tilde q(x,t)=\ds\frac{64\,w_{17}}{w_{18}} \exp\bigg(2x+4i\,t\bigg), \quad \tilde r(x,t)=\tilde q(x,t)^*,
\end{equation}
where we remark that $\tilde q(x,t)$ appearing in \eqref{9.46} only differs from $\tilde q(x,t)$ listed in \eqref{9.44} by the absence of
the inverse tangent function and that the second equality in \eqref{9.46} is compatible with
\eqref{8.17} using the parameter $\kappa=1.$ 

\end{example}

\end{document}